\documentclass[12pt]{article}
\RequirePackage[colorlinks,citecolor=blue,urlcolor=blue]{hyperref}
\RequirePackage[OT1]{fontenc}
\usepackage{amsmath,amsthm,amsfonts,amssymb}
\usepackage{array}
\usepackage{bm}
\usepackage{multirow}
\usepackage{enumitem}
\usepackage[usenames,dvipsnames,svgnames,table]{xcolor} 
\usepackage{bbm}

\usepackage{graphicx}

\usepackage[margin=2.9cm]{geometry}
\usepackage{mathtools}
\usepackage{multirow}
\usepackage{colortbl}
\usepackage{arydshln}
\usepackage{xspace}
\usepackage{multirow}

\definecolor{kugray5}{RGB}{224,224,224}
\usepackage[style=numeric,%alphabetic,\begin{equation}\label{key}
			sorting=none,
			firstinits=true,
			backend=bibtex,
			doi=false,
			isbn=false,
			url=false,
			]{biblatex}

\bibliography{mbisbm_paper}
\renewbibmacro{in:}{}			
\DeclareFieldFormat[article,inbook,incollection,inproceedings,patent,thesis,unpublished]{citetitle}{#1}
\DeclareFieldFormat[article,inbook,incollection,inproceedings,patent,thesis,unpublished]{title}{#1} 
\usepackage{arash_macros}

\usepackage{algorithm}
\usepackage{algpseudocode}
\usepackage{xpatch}
\makeatletter
\xpatchcmd{\algorithmic}{\itemsep\z@}{\itemsep=.5ex plus2pt}{}{}
\makeatother

\usepackage{verbatim}

\usepackage[]{titlesec} % or RequirePackage[loadonly]{titlesec} in a cls file
\titleformat*{\subsection}{\large\bfseries}
\titleformat*{\subsubsection}{\normalsize\bfseries}
\titleformat*{\paragraph}{\normalsize\bfseries}
\titleformat*{\subparagraph}{\normalsize\bfseries}

\usepackage{comment}

\newcommand{\E}{\mathbb{E}}

\newcommand{\indep}{\stackrel{\text{ind}}{\sim}}

\IfFileExists{upquote.sty}{\usepackage{upquote}}{}

\DeclareMathOperator{\tr}{tr}

\newcommand{\mut}{\widetilde{\mu}}
\newcommand{\muc}{\bar{\mu}}
\newcommand{\Sigt}{\widetilde{\Sigma}}

\newcommand{\Psib}{{S}}
\newcommand{\Gamt}{\widetilde{\Gamma}}

\newcommand{\Pc}{\mathcal{P}}
\newcommand{\onev}{\bm{1}}

\newcommand{\bern}{\text{Ber}}

\newcommand{\poi}{\text{Poi}}

\newcommand{\At}{\widetilde{A}}
\newcommand{\Zt}{\widetilde{Z}}
\newcommand{\Psit}{\widetilde{\Psi}}

\DeclareMathOperator{\softmax}{softmax}

\newcommand{\tauc}{\bar{\tau}}
\newcommand{\xc}{\bar{x}}

\newcommand\mbisbm{\texttt{mbiSBM}\xspace}
\newcommand\bisc{\texttt{biSC}\xspace}
\newcommand\scp{\texttt{SCP}\xspace}
\newcommand\simrnd{\texttt{$\sim$rnd}\xspace}
\newcommand\rnd{\texttt{rnd}\xspace}

\newcommand{\mnorm}[1]{|\!|\!| #1 |\!|\!|}
\newcommand{\cities}{{\scshape  Cities}\xspace}
\newcommand{\topart}{{\scshape  TopArticles}\xspace}
\newcommand{\ER}{Erd\"{o}s--R\'{e}nyi\xspace}

\newcommand{\gbin}{g_{\text{ber}}}
\newcommand{\gpoi}{g_{\text{poi}}}

\title{Matched bipartite block model with covariates}
\author{Zahra S. Razaee, Arash A. Amini, Jingyi Jessica Li}
\date{}
\begin{document}
\maketitle
	\begin{abstract}
%		Community detection in bipartite networks is challenging. Motivated by an application from biology, we consider a model which allows for matched communities in the bipartite setting, in addition to node covariates with information about the matching. We derive a simple fast algorithm for fitting the model based on variational inference ideas and show its effectiveness in simulations.
		Community detection or clustering is a fundamental task in the analysis of network data. Many real networks have a bipartite structure which makes community detection challenging. In this paper, we consider a model which allows for matched communities in the bipartite setting, in addition to node covariates with information about the matching. We derive a simple fast algorithm for fitting the model based on variational inference ideas and show its effectiveness on both simulated and real data. A variation of the model to allow for degree-correction is also considered, in addition to a novel approach to fitting such degree-corrected models.
	\end{abstract}

%!TEX root = sbm_bip_arxiv.tex
\section{Introduction}\label{sec:intro}
Network analysis has been a very active area of research with applications to social sciences, biology and marketing, to name a few. 
%are many and varied, including social networks, biological networks, fMRI networks, etc. 
A fundamental problem in network data analysis is community detection, or clustering: Given a collection of nodes and a similarity matrix among them, interpreted as the adjacency matrix of a (weighted) network, one wants to partition the
nodes  into clusters, or communities, of high similarity. For undirected networks, a popular model for community-structured networks is the stochastic block model (SBM) ~\cite{Holland1983} and its variants~\cite{Karrer2011,Gopalan2013}, which have been extensively investigated in recent years both in terms of theoretical properties and efficient fitting algorithms.  See, for instance~\cite{Bickel2009, Decelle2011, Rohe2011,Mossel2012,Zhao2012,Amini2013,Qin2013,Mossel2013,Massoulie2014,Amini2014,Gao2015,Jing2015,Hajek2016, Abbe2016,Gao2016} for a sample of the work. On the other hand, a natural structure is often present in many  real networks, that of being bipartite, where nodes are divided into two sets, or \emph{sides}, and only connections between nodes of different sides are allowed. Examples include networks of actors and movies, scientific papers and their authors, shoppers and products, and transcription factors and their binding sites.  Block-modeling with the explicit aim of taking into account the bipartite nature of a network has received comparatively less attention. Interesting new modeling possibilities emerge in the bipartite case, chief among them being the issue of matching between the communities of the two sides. 

% where nodes are divided into two sets, and only connections between nodes of different sets are allowed. 
%
%Network analysis has become a very active area of research for analyzing complex data. The applications are many and varied, including social networks, biological networks, fMRI networks, etc. A natural structure present in many instances is that of a bipartite network,  where nodes are divided into two sets, and only connections between nodes of different sets are allowed. Examples of bipartite networks include networks of actors and movies, scientific papers and their authors, shoppers and products and transcription factors and their binding sites.  

The problem of community detection in bipartite networks is closely related to that of co-clustering, also known as bi-clustering, which goes back at least to~\cite{Hartingan1972}. Co-clustering refers to simultaneous clustering of the rows and columns of a matrix, the \emph{ bi-adjacency matrix} of a bipartite graph. It has been extensively used in biological applications~\cite{Cheng2000, Madeira2010} and text mining~\cite{Dhillon2001, Dhillon2003,Bisson2008}. Recently, \cite{Choi2014,Flynn2012} studied likelihood-based co-clustering. In \cite{Rohe2016},  a spectral co-clustering algorithm has been proposed for directed networks and discussed how it can be applied to bipartite setting. Often, the co-clustering formulation ignores the issue of matching of the clusters, in the sense that in general any row cluster can be in relation to any column cluster.

Another common approach is to reduce community detection in the bipartite setting to two separate instances of usual clustering of (undirected) unipartite networks, by forming \emph{one-mode projections} of the network onto the two sides~\cite{Zhou2007}.
%
%One fundamental problem in network data analysis is community detection, or clustering. Given a collection of nodes and a similarity matrix among them, interpreted as the adjacency matrix of a (weighted) graph, one wants to partition the
%nodes  into clusters (or communities) of high similarity. However, community detection in bipartite networks is tricky since the nodes in the same set (or side) are not connected. One common approach is to first form one-mode projections of the network into the two sets and then apply standard community detection algorithms to each mode (or side) individually.
%%
  Despite a moderate reduction in the dimension (having to deal with two smaller networks), the projection approach suffers from information loss and identifiability issues~\cite{Zhou2007}. Projection can also turn a structured bipartite network into unstructured unipartite ones or vice versa~\cite{Larremore2014}. Another major difficulty is establishing a link between the communities on the two sides. One can come up with ad-hoc association measures between communities of the two sides, e.g., by counting links between each pair. This, however, leads to another bipartite graph on the communities, leading to the difficulty of interpretation. In effect, the problem transfers from community detection on the individual nodes, to that on the newly-discovered communities, or supernodes. 
  
%  One can try to match communities of both sides by running a matching algorithm on these supernodes, however, the overall approach is somewhat ad-hoc and dependent on arbitrary choices (i.e., how to form the community-level bipartite graph). It is worth noting that there is a deeper more unsettling problem with this approach. Recalling the reduction from the bipartite network on original nodes to a bipartite network on discovered communities, achieved by  projection followed by community detection, one is tempted to repeat the process, i.e., apply projection to this new bipartite network and do community detection again. This process can be continued and there is no clear stopping point.
Block-modeling in the bipartite setting has recently gained more attention. In \cite{Wyse2014}, a method was proposed to infer both community memberships as well as the number of communities in a bipartite network using a blockmodel and an algorithm similar to the  iterated conditional modes~\cite{Besag1986}. A bipartite stochastic block model (BiSBM)\cite{Larremore2014}, built on the work of~\cite{Karrer2011}, has been proposed to infer bipartite community structure in both degree-corrected and uncorrected regimes by maximizing a profile likelihood over all partitions. In both cases, the issue of matching of the communities on the two sides is not the main concern.

Motivated by the matching problem, in this paper, we consider the problem of \emph{matched community detection} in a bipartite network. In many practical examples, one either expects a one-to-one correspondence between the communities of the two sides, or it is reasonable to postulate such structure, due to ease of interpretation (Section~\ref{sec:mbiSBM}). The problem of ``finding communities in a bipartite network that are in one-to-one correspondence between the two sides'' is what we refer to as matched community detection. We will propose a generative model for such networks where there is a hidden matched community structure. This avoids the need for post-hoc matching of the communities: the matching is built into the model and inferred simultaneously with the communities in the process of fitting the model. Our model is a natural extension of the well-known stochastic block model (SBM) and is discussed in details in Section~\ref{sec:mbiSBM}. We also discuss an extension to allow for degree-correction (Section~\ref{sec:dc:extension}), providing a matched version of degree-corrected block model (DC-SBM).

In another direction, many networks come with metadata, often in the form of node features, or covariates.
%
% or sometimes edge. 
%However, most studies ignore the covariate information and only focus on the network. 
%
The potential for improving quality of the clusters by incorporating node covariates has been explored in recent work, in the context of unipartite networks~\cite{Binkiewicz,Zhang2015,Yan2016, Newman2016}. Bipartite setting adds another challenge to modeling node covariates, in particular, how to jointly model the covariates from the two sides, considering that one often has covariates of different dimensions on each side. (The extreme case is when only once side has node covariates.)
%
% The challenge in jointly analyzing the network and covariate information in bipartite setting is that node covariates for the two modes are (usually) of different dimensions.
%%and there might be a statistical linkage between and within the covariates of the two modes. 
%
%
We extend our proposed model to allow for the presence of node covariates that are aware of the matching between communities. In other words, covariates corresponding to nodes in matched communities are statistically linked. The linkage can be tunned using a general cross-covariance matrix, allowing for varying degrees of covariate influence on the community detection problem (Section~\ref{sec:mbiSBM}). 

To fit the model, we derive an algorithm based on variational inference, also known as variational Bayes~\cite{Jordan1999,Blei2016} ideas (Section~\ref{sec:model:fitting}). We derive both the degree-corrected and uncorrected versions of algorithm within the same unified framework, namely, sequential block-coordinate ascent on the variational likelihood. This in particular leads to a novel approach to fitting degree-corrected likelihoods using methods of continuous optimization (as opposed to profiling out the degree-correction parameters and optimizing over the space of discrete labels.). As part of the initialization of the algorithm, we revisit a bipartite spectral clustering algorithm, first proposed in~\cite{Dhillon2001}, and identify it as an effective algorithm for matched bipartite clustering.  We show the effectiveness of our approach on simulated (Section~\ref{sec:simulation}) and real data (Section~\ref{sec:realdata}), namely, page-user networks collected from Wikipedia (Section~\ref{sec:realdata}).

 To summarize, our contributions in this paper are the following:
 \begin{itemize}
 	\item[(i)] Identify the matching problem in bipartite community detection more clearly and give it the prominent role, by showing that it is possible to consider matched communities from the start in the modeling process. Bringing attention to matched bipartite clustering (or community detection) as a well-defined problem also allows us to identify an earlier spectral algorithm, namely that of~\cite{Dhillon2001}, originally proposed in the context of topic modeling, as effectively solving the matched version of bipartite clustering. At present, we are unaware of any other algorithm that attempts to solve this problem directly.
 	
 	\item[(ii)] Propose a natural bipartite extension of the SBM and DC-SBM, which we refer to as {\it{matched bipartite stochastic block model}} (\mbisbm), has a latent structure of matched communities and allows for node covariates that are potentially informative  about the matching (see Section~\ref{sec:mbiSBM}). Some of the challenges involved in joint modeling of the node covariates of the two sides are resolved by appealing to hierarchical Bayesian modeling ideas~\cite{Gelman2004}. 
 	
 	\item[(iii)] Show the effectiveness of the variational Bayes approach in fitting the overall \mbisbm model, when combined with good initialization, especially a variant of \bisc algorithm of~\cite{Dhillon2001}. The algorithm is a block-coordinate ascent with a closed-form, fairly cheap iterations, and can be scaled to large networks. 
 \end{itemize}

%\aaa{This could go here or elsewhere.}

% \subsection{Related work}
% The problem of community detection in bipartite network is closely related to co-clustering(a.k.a bi-clustering) which was proposed by \cite{Hartingan1972}. Co-clustering refers to simultaneous clustering of the rows and columns of a matrix (incidence matrix in bipartite setting).  {}

% gained much attention in the social networks and machine
% learning literature, see for example~\cite{Dhillon2001} Dhillon(2001), Doreian et al. (2004), Brusco et al. (2013) Doreian et al. (2013). \cite{Rohe2015} Rohe et al. (2015) discuss how their stochastic co-blockmodel may be extended to a bipartite setting. Wyse and Friel (2012) have proposed a method to infer structure in bipartite networks using the latent blockmodel and exact ICL. Larremore et. al have proposed an stochastic block model (BiSBM) that is built upon the work of Karrer and Newman (2011) to infer bipartite community structure in both degree-corrected and uncorrected regimes.  However, to the best of our knowledge, no one has considered the matched community detection in a bipartite network. Our method solves
% this problem with and without the presence of node covariates. @Amini2013 

\bigskip
\noindent \textbf{Notation.} We write $[K] := \{1,\dots,K\}$ and $\Pc_K := \{ p \in \reals^K_+ :\; \onev^T p = 1 \}$, for the set of probability vectors on $[K]$. Here, $\onev$ is the all-ones vector of dimension $K$. We identify $[K]$ with $\{0,1\}^K \cap \Pc_K$, the set of binary vectors of length $K$ having exactly a single entry equal to one. The identification is via the so-called \emph{one-hot} encoding: $z = k$ as an element of $[K]$ iff $z_{k} = 1$, treating $z$ as element of $\{0,1\}^K \cap \Pc_K$. 
 $x \mapsto N(x;\mu,\Sigma)$ denotes the PDF of a multivariate Gaussian distribution with mean $\mu$ and covariance $\Sigma$. The constant terms in an expression are denoted as ``const.''. We write $\doteq$ for equality up to additive constants. We use $[Z_1;Z_2]$ to denote the vertical concatenation of two matrices $Z_1$ and $Z_2$, having the same number of columns.

% Our proposed model has two assumptions: 1) The number of communities (clusters) in both groups should be equal. 2) There is an implicit matching between clusters of both sides.

%!TEX root = sbm_bip_arxiv.tex
\section{Matched bipartite SBM}\label{sec:mbiSBM}
% \aaa{At some point you should talk about SBM, but not clear where.}
%The stochastic block model (SBM) is a generative model for networks with communities (or blocks). In the most basic SBM, each node is assigned to one of the $K$ communities and the edges are placed independently between node $i$ and node $j$ with probability $p_{ij}$. If $p_{ij} = p$ for all $i$ and $j$, then the result is the famous Erd{\H o}s–R{\'e}nyi model. This case is degenerative and the partition into communities becomes irrelevant. Another special case happens when all the nodes within the same community are connected with the constant probability $p$ while any two nodes in different communities are connected with constant probability $q$. This special case is called the {\it{planted partition model}}. If $p > q$, the model is called {\it{assortative}}, while the case $p < q$ is called {\it{dissortative}}. 

The stochastic block model (SBM) is a generative model for networks with communities (or blocks). In the most basic SBM, sometimes called the {\it{planted partition model}}, each node is assigned to one of the $K$ communities and the edges are placed independently between two nodes $i$ and $j$, with probability $p$ if $i$ and $j$ belong to the same community, and with probability $q$ otherwise. When $p = q$, one recovers the famous Erd{\H o}s–R{\'e}nyi model, where there is no genuine community structure. The interesting cases are the \emph{assortative} model where $p > q$ and the \emph{dissortative} model where $p < q$. Our focus in this paper is mainly on the assortative case, though the results can be easily adapted to the other case.

% \paragraph{The model.} Motivated by the aforementioned biological example, we propose a new generative model based on the stochastic block model which is suited to the bipartite nature of the network and also accommodates node covariates. 
% The bipartite network in our biological example is the orthologus relation between the two species and the node covariates are gene expression over the developmental stages of the species.
%
We start with the ingredients needed to define the matched bipartite SBM (\mbisbm). Assume that we have two groups of nodes $[N_1] = \{1,\dots,N_1\}$ and $[N_2]= \{1,\dots,N_2\}$, representing nodes on the two sides of a bipartite network.  %Compactly, we have the groups $[N_r], r=1,2$. 
We assume that there is a partition $\{C_{rk}\}_{k=1}^K$ of $[N_r]$, for each $r=1,2$. This is our latent  community structure.  In referring to  $C_{rk}$, we will use the terms {\it{community}} and {\it{cluster}} interchangeably. We assume the following implicit (true) \emph{one-to-one matching} between these communities:
\begin{align}\label{eq:C:match}
  C_{1k} \leftrightarrow C_{2k}, \quad k=1,\dots,K.
\end{align}
%That is, $\{C_{rk}: r=1,2\}$ to be matched for all $k=1,\dots,K$.
To each node $i$ in group $r$, we assign a community membership variable $z_{ri}$ showing which community it belongs to:
\begin{align*}
   z_{ri} = k\; &\iff i \in C_{rk}, \quad \forall i \in [N_r], \; r = 1,2.
\end{align*}
Recalling the identification $[K] \cong \{0,1\}^K \cap \Pc^K$, we treat $z_{ri}$ as both an element of $[K]$ and a binary vector of length $K$, hence, with some abuse of notation, $z_{ri} = k$ and $z_{rik} = 1$ are equivalent.
We collect these labels in \emph{membership matrices}  $ Z_r := (z_{ri}: i \in[N_r]) \in \{0,1\}^{N_r \times K}$,  $r=1,2$, where each $z_{ri}$, treated as a binary vector, appears as a row in $Z_r$. %, where each $z_{ri}$ is identified with a binary vector in $\{0,1\}^K$ (via $z_{rik} = 1 \iff z_{ri} = k$) and appears as a row in $Z_r$.
 We also let $Z = [Z_1;Z_2] \in \{0,1\}^{(N_1+ N_2) \times K}$ be the \emph{ matched membership matrix} obtained by vertical concatenation of $Z_1$ and $Z_2$. 

% after padding with zeros on left or right: $z_{1i}$ form rows $(z_{1i},0_{K})$ for $i \in [N_1]$ and $x_{2j}$ form rows $(0_{K},z_{2j})$ for $j \in [N_2]$.

For each node $i$ in group $r$, we observe a covariate vector $x_{ri} \in \reals^{d_r}$. If we want to specify the components of this vector we write $x_{rij}, j=1,\dots,d_r$. 
Let $X := (x_{ri},\, i\in [N_r],\, r = 1,2)$. We often think of $X$ as a matrix in $\reals^{(N_1+ N_2) \times (d_1+d_2)}$, by padding covariate vectors with zeros on the left or right: $x_{1i}$ form rows $(x_{1i},0_{d_2})$ for $i \in [N_1]$ and $x_{2j}$ form rows $(0_{d_1},x_{2j})$ for $j \in [N_2]$.

In addition to the covariate matrix $X$, we also observe a bipartite network on $[N_1] \times [N_2]$ represented as a \emph{bi-adjacency matrix} $A \in \{0,1\}^{N_1\times N_2}$. Thus, the observed data is $(X,A)$. We  assume that given the latent community labels $Z$, $X$ is independent of $A$. Below, we outline how each of these components are generated given $Z$.

%\medskip
%\simplesubsec{}
\begin{figure}
	\centering

		\includegraphics[scale=.9]{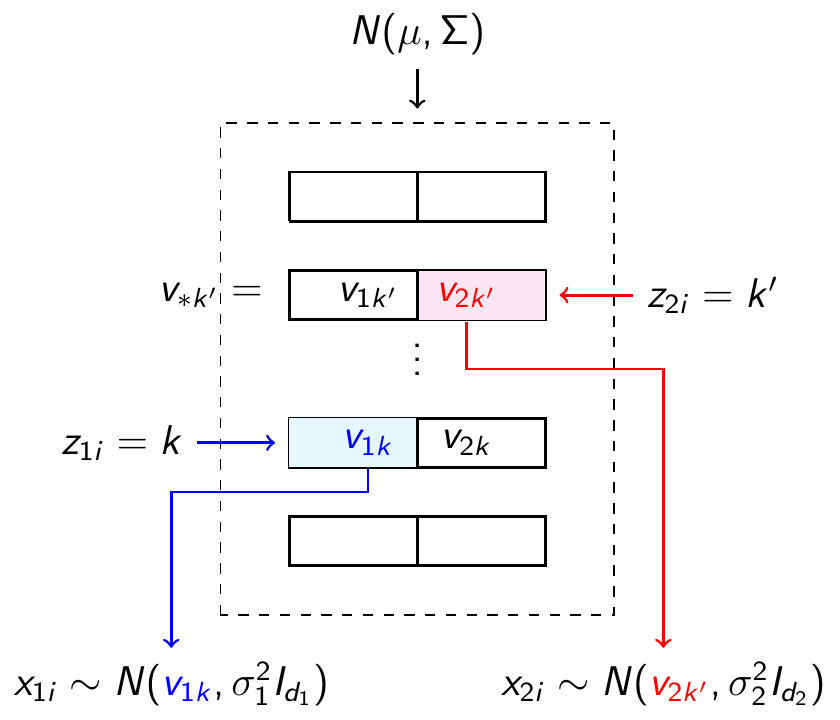}\hskip2em
		 \raisebox{.75\height}{\includegraphics[scale=1]{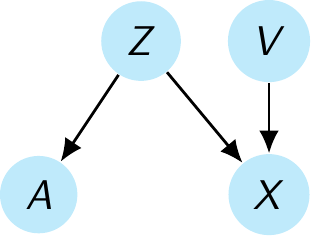}}

	\caption{(Left) Schematic diagram for the hierarchical generation of node covariates. (Right) Overall graphical representation of the model.}
	\label{fig:covar:gen}
\end{figure}
\subsection{Generating covariates}
To generate $x_{ri}$, we use a hierarchical mixture model: First we generate the mean vector $v_{rk}\in \mathbb{R}^{d_r}$ associated with each cluster $C_{rk}$, and then we draw $x_{ri}$ from a normal distribution with mean $v_{rk}$ when $z_{ri} = k$; see Figure~\ref{fig:covar:gen}. In order to model the correlation (i.e., a statistical link)  between covariates of matched clusters, we draw the entire vector $v_{*k} := (v_{rk},r=1,2) = (v_{1k},v_{2k})$ from a multivariate normal distribution with possibly nonzero covariance matrix between the two components $v_{1k},v_{2k}$. We have the following model: 
\begin{align}\label{eq:model:1}
\begin{split}
  z_{ri} &\sim \text{Mult}(1,\pi_r ), \\
  (v_{rk},r=1,2) &\iid N(\mu,\Sigma), \;\; k=1,\dots,K.\\
  x_{ri} | z_{ri} = k , v_{rk} &\sim N(v_{rk}, \sigma_r^2 I_{d_r}), \;\; i \in [N_r], \; r=1,2
\end{split}
\end{align}
where the draws are independent over $r$ and $i \in [N_r]$, on each line.
Here, $\pi_r = (\pi_{r1},\dots,\pi_{rK})$ is the prior on cluster proportions for group $r$ ($\pi_r \in [0,1]^K$ with $\sum_{k=1}^K \pi_{rk} = 1$).  $(\sigma^2_r, r=1,2)$ models the variance of measurement noise in the two groups.
To make the correlation structure in $(v_{rk},r=1,2)$ more explicit, we can partition $\mu = (\mu_{r},r=1,2)$ and $\Sigma$, so that
\begin{align*}
  \begin{pmatrix}
     v_{1k} \\ v_{2k}
  \end{pmatrix}
  &\indep
   N \bigg[ 
  \begin{pmatrix}
    \mu_{1} \\ \mu_{2}
  \end{pmatrix},
  \begin{pmatrix}
    \Sigma_{11} & \Sigma_{12} \\
    \Sigma_{12}^T & \Sigma_{22}
  \end{pmatrix}
  \bigg], \;\; k=1,\dots,K.
\end{align*}
Note that $\mu_{r} \in \reals^{d_r}$. 
%Figure~\ref{fig:covar:gen} shows a schematic diagram of hierarchical generation of covariates according to~\eqref{eq:model:1}. 
%
In subsection~\ref{sec:covariate:corr}, we discuss how this model provides a statistical link between covariates of the two groups. For future reference, $v_{*k} := (v_{rk},r=1,2)$ collects the matched hidden covariate means of the clusters $C_{1k}$ and $C_{2k}$. On the other hand, we write $v_{r*} = (v_{rk}, k \in [K])$ which collects all the hidden covariate means for side $r=1,2$ of the network.

%Let $E_{d_1,d_2} = 1_{d_1} 1_{d_2}^T$ be the $d_1 \times d_2$ matrix of all ones. We can assume for simplicity that $\Sigma_{12} = \rho_{12} E_{d_1,d_2}$ and $\Sigma_{rr} = \alpha_r I_{d_r} + (\rho_{r}-\alpha_r) E_{d_r,d_r}$ for $r=1,2$.

%\medskip
%\simplesubsec{Generating the graph ($A$)}
\subsection{Generating the network}
 Given $Z$, the bipartite graph is generated as follows: For each node $i$ in $[N_1]$ and each node $j$ in $[N_2]$, we put an edge between them with probability $p$ if they belong to matched clusters, and with probability $q \neq p$ otherwise. With $A = (A_{ij}) \in \{0,1\}^{N_1 \times N_2}$ denoting the resulting bi-adjacency matrix, we have
\begin{align}\label{eq:model:2}
  A_{ij} \mid Z \; \indep\;
  \begin{cases}
    \text{Ber($p$)} & z_{1i} = z_{2j} \\
    \text{Ber($q$)} & z_{1i} \neq z_{2j}
  \end{cases},
  \;\;\; i \in [N_1],\; j \in [N_2].
\end{align}
Combined,~\eqref{eq:model:1} and~\eqref{eq:model:2} describe our full matched bipartite SBM model.
The objective is to find the posterior probability of $Z$ given $A$ and $X = \{x_{ri}: i \in [N_r], r=1,2\}$.

Although we will focus on the simple model~\eqref{eq:model:2} in deriving the algorithms, it is possible to allow for a more general edge probability structure as in the usual SBM, by assuming 
\begin{align}\label{eq:model:3}
	A_{ij} \mid Z \; \indep\; \bern(\Psi_{z_{1i},z_{2j}})
	\;\;\; i \in [N_1],\; j \in [N_2].
\end{align}
where $\Psi \in [0,1]^{K \times K}$ is a
%\emph{symmetric}
 connectivity (or edge probability) matrix. Model~\eqref{eq:model:2} corresponds to the case where $\Psi_{kk} = p$ and $\Psi_{k\ell} = q$ for $k\neq \ell$. We will refer to this model as \texttt{mbiSBM} for matched bipartite SBM.

\begin{rem}\label{rem:tuning:cov:info}
	Parameter $\Sigma$ in~\eqref{eq:model:1} is key in tuning the effect of the node covariates on community detection. Assume for simplicity that $\sigma_r^2 = 0, \, r=1,2$. Then, when $\Sigma=0$, $v_{*k} = \mu$ a.s. for all $k$, hence $x_{*i} = \mu$ for all $i$, and the  covariates carry no information about communities. When, $\Sigma \neq 0$ there is variability in $v_{*k}$ across $k$, hence community detection benefits from the covariate information. On the other hand, it is well-known that the information in the adjacency matrix $A$ about community structure is roughly controlled by the expected degree of the network, i.e., the scaling of $Q = (p,q)$, in addition to the separation of $p$ and $q$. Thus, by rescaling $\Sigma$ and $Q = (p,q)$, we can control the balance of these two sources of information. This is explored in Section~\ref{sec:simulation} through simulation studies.
\end{rem}

\subsection{Connection with the usual SBM}
Ignoring the covariate part of the model, one might wonder whether \texttt{mbiSBM}, introduced in~\eqref{eq:model:3}, can be thought of as a sub-model of a usual SBM with perhaps increased number of communities. First, it should be clear that the model is not a usual SBM with $K$ communities. However, it can be thought of as a SBM with $2K$ communities with \emph{restrictions} imposed on both its membership and connectivity matrix. To see this, let us recall the matrix representation of the usual SBM with $K$ blocks, where one has the connectivity matrix $\Psi \in [0,1]^{K \times K}$ and binary membership matrix $Z \in \{0,1\}^{N\times K}$. Such model can be compactly represented as $\ex[A |Z] = Z \Psi Z^T$.

Now, consider model~\eqref{eq:model:3}, and let $Z_r = (z_{ri}) \in \{0,1\}^{N_r \times K}$ for $r=1,2$. We express the model compactly as
\begin{align}\label{eq:At:Zt:Psit}
\ex \underbrace{\begin{pmatrix}
	0 & A \\
	A^T & 0
	\end{pmatrix}}_{\At} = 
\underbrace{\begin{pmatrix}
	Z_1 & 0 \\
	0 & Z_2
	\end{pmatrix}}_{\Zt}
\underbrace{\begin{pmatrix}
	0 & \Psi \\
	\Psi^T & 0
	\end{pmatrix}}_{\Psit}
\begin{pmatrix}
Z_1^T & 0 \\
0 & Z_2^T
\end{pmatrix}.
\end{align}
given $Z_1,Z_2$. Letting $N := N_1  + N_2$ and 
defining the matrices $\At \in \{0,1\}^{N \times N}$, $\Zt \in \{0,1\}^{N \times 2K}$ and $\Psit \in [0,1]^{2K \times 2K}$ as in~\eqref{eq:At:Zt:Psit}, it is clear that model~\eqref{eq:model:3} is equivalent to $\ex[\At|\Zt] = \Zt \Psit \Zt^T$. This is a SBM with restrictions on both $\Zt$ and $\Psit$:  Nodes $1,\dots,N_1$ can only belong to communities $1,\dots,K$ and nodes $N_1+1,\dots,N_1+N_2$ can only belong to communities $K+1,\dots,2K$. As for $\Psit$, the restriction imposes zero connectivity among communities $1,\dots,K$ and among those of $K+1,\dots,2K$. With these restrictions in place, we have a natural bipartite matching between communities: $\ell \leftrightarrow K+\ell$ for $\ell \in [K]$. 

\subsection{Degree-corrected version}\label{sec:dc:model}
A limitation of the SBM is that nodes in the same community have the same expected degree. To allow for degree heterogeneity within communities, bringing the model closer to real networks, a common approach is to use the DC-SBM~\cite{Dasgupta2004,Karrer2011}. It is fairly straightforward to introduce degree-correction in our setup. Consider the form of the matched SBM introduced in~\eqref{eq:model:3}. To each node $i$ in group $r$, we associate a propensity parameter $\theta_{ri} > 0$. Thus, we have additional parameters $\theta_r := (\theta_{ri}, \;i \in [N_r])$ for $r=0,1$. The degree-corrected (DC) version of the model replaces~\eqref{eq:model:3} with
\begin{align}\label{eq:dc:model}
A_{ij} \mid Z \; \indep\; \poi(\theta_{1i} \theta_{2j}\Psi_{z_{1i}, \,z_{2j}})
\;\;\; i \in [N_1],\; j \in [N_2].
\end{align}
Replacing the Bernoulli with Poisson is for convenience in later derivations, and is common in dealing with DC-SBM~\cite{Karrer2011}. In order for the parameters $(\theta_1,\theta_2,\Psi)$ to be identifiable, we need to agree on a normalization of $\theta_{ri}$ per each community. Here, we adopt the following:
\begin{align}\label{eq:theta:normalization}
\frac{1}{|C_{rk}|} \sum_{i \,\in\, C_{rk}} \theta_{ri} = 1  \iff  \sum_{i=1}^{N_r} (\theta_{ri} -1) z_{rik} = 0, \quad k \in [K], \; r=1,2.
\end{align}
With this normalization, we recover the original model when $\theta_{ri} = 1$ for all $i$ and $r$. Our normalization is similar to the one considered in~\cite{Gao2016}.

\begin{rem}
	A normalization of the form~\eqref{eq:theta:normalization} is often assumed when one considers both $\theta = (\theta_1,\theta_2)$ and $Z$ to be deterministic unknown parameters, or alternatively when working conditioned on $\theta$ and $Z$. Throughout, we assume $\theta$ to be an unknown parameter. However, to be pedantic, \eqref{eq:theta:normalization} is inconsistent with i.i.d. random generation of $z_{ri}$ from a $\text{Mult}(1,\pi_r)$ as in~\eqref{eq:model:1}. One way to get around this is to assume that the labels are generated a priori from the product multinomial distribution described in~\eqref{eq:model:1} \emph{conditioned} on the set of labels satisfying~\eqref{eq:theta:normalization}. We will ignore the change in the label prior this conditioning makes in deriving the algorithms. In the end, we enforce~\eqref{eq:theta:normalization} in an ``averaged'' sense, replacing $z_{rik}$ with the corresponding (approximate) posterior $\tau_{rik}$, as detailed in subsection~\ref{sec:dc:extension}. Viewed as a set of constraints on the collection of soft-labels $(\tau_{rik})$,~\eqref{eq:theta:normalization} is not that restrictive.
\end{rem}

\subsection{Covariate correlation on matched clusters}\label{sec:covariate:corr}
One desirable feature in modeling covariates, in the context of a matched bipartite network, is the ability to gain some information about whether a pair $(i,j) \in [N_1] \times [N_2]$ belongs to a matched cluster, by just looking at their respective covariates $x_{1i}$ and $x_{2j}$. Assume for the moment that there is no measurement noise, i.e., $\sigma_r^2 = 0, \; r=1,2$. Then, the question boils down to whether we can tell  $(v_{1k},v_{2k'})$ for $k \neq k'$ apart from $(v_{1k},v_{2k})$.
According to the model, $(v_{1k},v_{2k})$ and $(v_{1k'},v_{2k'})$ are independent Gaussian vectors, hence $	(v_{1k},v_{2k},v_{1k'},v_{2k'})$ is Gaussian with mean $(\mu,\mu) = (\mu_1, \mu_2, \mu_1, \mu_2)$ and covariance $
(\begin{smallmatrix}
\Sigma & 0 \\
0 & \Sigma
\end{smallmatrix})$.
Recalling the decomposition of $\Sigma$, it follows that $(v_{1k},v_{2k'}) \sim N(\mu, (\begin{smallmatrix}
\Sigma_{11} & 0 \\
0 & \Sigma_{22}
\end{smallmatrix} ))$ for $k \neq k'$ whereas $(v_{1k},v_{2k}) \sim N(\mu, (\begin{smallmatrix}
\Sigma_{11} & \Sigma_{12} \\
\Sigma_{21} & \Sigma_{22}
\end{smallmatrix} ))$. As long as $\Sigma_{12}\neq 0$, these two distributions are different, hence the model is able to distinguish between the two cases. In other words, there is information in the covariates about the matching of the clusters in the two groups. However, this information (in itself) is quite weak since it amounts to distinguishing between two multivariate Gaussian distributions, based only on a single draw from each. Fortunately, the model also carries information about the matching in the adjacency matrix $A$.

\begin{figure}[t]
	\centering
	\begin{tabular}{p{.5in}p{.5in}p{.5in}p{.5in}}
		\includegraphics[height=1.3in]{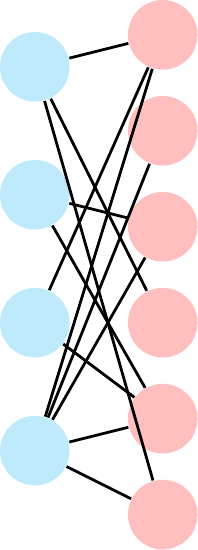} &
		\includegraphics[height=1.3in]{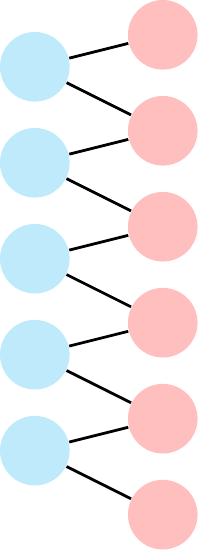}	&
		\includegraphics[height=1.3in]{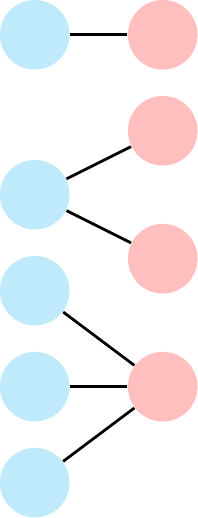} &
		\includegraphics[height=1.3in]{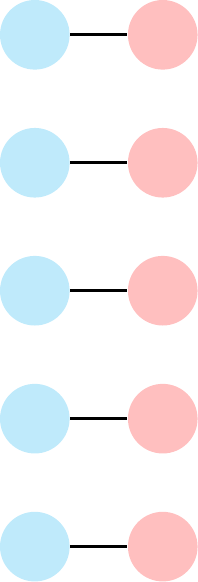}	\\
		(a) & (b) & (c) & (d)			
	\end{tabular}
	\caption{Possible relations between communities of the two sides, in a bipartite network. (a) and (b) are hard to interpret. Structures like (c), i.e., collections of disjoint stars, are interpretable and (d) is the simplest within this class.}
	\label{fig:interp}
\end{figure}

\subsection{Interpretability and identifiability}\label{sec:interp}
We alluded earlier to the merits of having a 1-1 matching between the communities of the two sides built into the model. Our main argument for the advantage of a 1-1 matching is interpretability. Figure~\ref{fig:interp} shows some possible relations that could exist between communities of the two side (each circle represents a community). The closer this relation is to a complete bipartite graph, the harder it is to interpret; Figure~\ref{fig:interp}(a) thus is perhaps the least informative relation among the four. In Figure~(b), the relation is much sparser. However, it is still hard to interpret: all communities seem to be related, albeit indirectly. We would like to argue that structures like (c) where the graph is a collection of disjoint \emph{stars} is interpretable. One would like to fit such models, though in full generality, this seems to be a difficult task. One has to somehow control the branching numbers of the stars, which indirectly control the number of communities on either side. Thus, the problem is at least as hard as deciding the number of communities in the usual SBM. 
%This goes back to the issue of identifiability: one can always split into more branches having the same probability of connection. 
Our 1-1 matching relation, Figure~\ref{fig:interp}(d), is the simplest structure of the type depicted in~(c). It is in a sense a first-order approximation of the models in this class. It is easiest to fit and is the most interpretable.

%\subsection{Identifiability and one-to-many matchings}

%!TEX root = sbm_bip_arxiv.tex
\section{Model fitting}\label{sec:model:fitting}%by variational inference}
In order to fit the model, we derive algorithms based on variational inference ideas. The algorithm starts from some initial guess of the labels and parameters and proceeds to improve the likelihood via simple iterative updates to the parameters and an approximate posterior on the labels. We first discuss the case with no degree correction. The extension to the degree-corrected model is discussed in subsection~\ref{sec:dc:extension}.

\subsection{The likelihood}\label{subsec:likelihood} Let us introduce some notation. We write $v_{* k} = (v_{rk}, r=1,2) \in \reals^{d_1+d_2}$ and $v_{r *} = (v_{rk}, k \in [K]) \in \reals^{K d_r}$, and $V = (v_{rk}, \,r = 1,2,\, k \in [K]) \in \reals^{K(d_1+d_2)}$.
Similarly, $Z = (z_{ri},\, i\in [N_r],\, r = 1,2)$ and $X = (x_{ri},\, i\in [N_r],\, r = 1,2)$. (In this section, the particular matrix form of $Z$ and $X$ are not of interest. $Z$ and $X$ are simply placeholders for the collections of labels and covariates.)~Let
\begin{align*}
y_{ij} := \mathbbm{1}\{z_{1i} = z_{2j}\}, \quad
z_{rik} := \mathbbm{1}\{z_{ri} = k\}.
\end{align*}
The joint distribution of all the variables in the model factorizes as follows:
\begin{align*}
	p(A,X,Z,V) &= p(A|Z) \, P(X|Z,V)\, p(Z)\, p(V) \\
	&= \prod_{i=1}^{N_1} \prod_{j=1}^{N_2} p(A_{ij} | z_{1i},z_{2j}) \prod_{r=1}^2 \prod_{i=1}^{N_r} \Big\{ p(x_{ri}|z_{ri} , v_{r*}) 
	p(z_{ri}) \Big\} \prod_{k=1}^K p(v_{*k}). %\\
\end{align*}

We have $p(x_{ri}|z_{ri} , v_{r*}) = \prod_{k=1}^K  [ f_r(x_{ri}; v_{rk})]^{z_{rik}}$ where  $f_r(x_{ri}; v_{rk}) := N(x_{ri}; v_{rk},\sigma^2_r I_{d_r})$. In addition, $p(z_{ri}) = \prod_{k=1}^K {\pi_{rk}}^{z_{rik}}$. For the network part, we in general have 
\begin{align}\label{eq:ell1:def}
	\ell_1(\Psi) =: \log p(A|Z) = \sum_{ij} \sum_{k\ell} z_{1ik} z_{2j\ell}
	\;g(\Psi_{k\ell},A_{ij}) 
\end{align}
where $g$ is either the log-likelihood of the Bernoulli, $\gbin(p,\alpha) = \alpha \log \frac{p}{1-p} + \log (1-p)$, or Poisson, $\gpoi(p,\alpha) = \alpha \log p  - p$. 

In the special planted partition case, $\log p(A|Z)$ greatly simplifies: By breaking up over $k = \ell$ and $k \neq \ell$, we obtain
\begin{align}\label{eq:ell1_longer}
	\ell_1(\Psi) = \log p(A|Z)  
		&= \sum_{ij} \Big[ \sum_{k} z_{1ik} z_{2jk}
			\;g(p,A_{ij}) + \sum_{k\neq \ell} z_{1ik} z_{2j\ell}
			\;g(q,A_{ij}) \Big] \notag \\ 
		&= \sum_{ij} y_{ij} g(p,A_{ij}) + (1-y_{ij}) g(q,A_{ij}).
\end{align}
where we have used $\sum_{k\ell} z_{1ik} z_{2j\ell} = 1$. The complete log-likehood of the model, i.e., assuming we observe the latent variables $(Z,V)$, is then 
\begin{align*}
\ell(\mu,\Sigma,\sigma,\pi,\Psi) = \ell_1(\Psi ) + \ell_2(\mu,\Sigma,\sigma,\pi)
\end{align*}
where $\ell_1(\Psi )$ is as defined in~\eqref{eq:ell1:def} and 
% $\ell_1(\Psi ) =  \sum_{ij} \sum_{k\ell} z_{1ik} z_{2j\ell}\;g(\Psi_{k\ell},A_{ij})$
\begin{align*}
	 \ell_2(\mu,\Sigma,\sigma,\pi) = \sum_{r=1}^2 \sum_{i=1}^{N_r} \sum_k z_{rik}  \log \big[ \pi_{rk}\, f_r(x_{ri}; v_{rk}) \big] + \sum_{k} \log p(v_{*k} | \mu, \Sigma).
\end{align*}

\subsection{Mean-field Approximation}\label{subsec:meanfield}
Variational inference is often regarded as the approximation of a posterior distribution by solving an optimization problem \cite{Wainwright2008,Blei2016}. The idea is to pick an approximation $q$ from some tractable family of distributions over the latent variables $(Z,V)$ and try to make this approximation as close as possible in KL divergence to the true posterior. We prefer to think of the approach as a generalization of the EM  algorithm, i.e., a general approach to maximize the incomplete likelihood by maximizing a lower bound on it. This lower bound, which we call variational likelihood, also known as the  evidence lower bound (ELBO)~\cite{Jordan1999}, involves both the likelihood parameters and a distribution $q$, namely,
\begin{align}\label{eq:elbo:def}
  J := \E_q[\ell(\mu,\Sigma,\sigma,\pi,\Psi) - \log q(Z,V)].
\end{align}
Here the expectation is taken,  assuming $(Z,V) \sim q$. One maximizes $J$ by alternating between maximizing over likelihood parameters $(\mu,\Sigma,\sigma,\pi,\Psi)$ and the variational posterior $q$. Without additional constraints, the optimization over $q$ leads to the posterior distribution of $(Z,V)$ given $(X,A)$, resulting in the EM algorithm. A genuine variational inference procedure, however, imposes some simplifying constraints on $q$. In particular, we impose the following factorized form, often referred to as the \emph{mean-field approximation}:
% Minimizing the KL divergence is equivalent to maximizing the Evidence Lower BOund (ELBO)(cite). 
%
%To do variational inference, we need a surrogate for the posterior distribution of $(Z,V)$. 
%
%We can assume that the variational distribution factorizes as follows:
\begin{align}\label{eq:var:dist}
	\textstyle 
  q(Z,V) = q_V(V) q_Z(Z), \quad q_Z(Z) = \prod_{r,i} q_{ri}(z_{ri}), \quad q_V(V)=\prod_{k=1}^K N(v_{*k}; \tilde{\mu}_k,\tilde{\Sigma}_k)
\end{align}
where $q_{ri}(z_{ri}) = \prod_{k=1}^K \tau_{rik}^{z_{rik}}$ is a multinomial distribution. In keeping up with our notation we write $\tau_{ri} = (\tau_{rik}, \; k\in [K])$. Note that $\tau = (\tau_{ri})$ collects  the approximate posteriors on node labels. They are the key parameters in our inference. 

  The particular form assumed for $q_V$ in~\eqref{eq:var:dist} is motivated by looking at the (true) posterior of $V$ given $Z$. We could have assumed a factorized form $q(Z,V) = q_Z(Z) p(V|Z)$ where $p(V|Z)$ is the the true posterior of $V$ given $Z$. However, the parameters of $p(V|Z)$ have a complicated dependence on $Z$. We have kept the form of $p(V|Z)$ while freeing the parameters, letting them be optimized by the algorithm.

\medskip
%Plugging the variational distribution in \elbo, we can compute $J$ as follows:
To simplify notation, let us define $\Gamt := ((\Sigt_k, \mut_k), k=1,\dots,K)$, collecting the parameters for the variational posterior $q_V$. Plugging in the variational distribution~\eqref{eq:var:dist} into the variational likelihood~\eqref{eq:elbo:def} using expression \eqref{eq:ell1_longer}  for $l_1(\psi)$, after some algebra detailed in Appendix~\ref{sec:elbo:deriv}, we get
% \begin{align*}
% 	J &= \E_q\Big[\ell(\mu,\Sigma,\sigma,Q)] -\log q(Z,V)\Big] 
% \end{align*}
% where
\begin{align}\label{eq:elbo:expr:1}
\begin{split}
	J &= \sum_{i,j} \Big[\gamma_{ij}(\tau) g(p;A_{ij}) + (1-\gamma_{ij}(\tau)) g(q; A_{ij}) \Big] 
+ \sum_{r,i,k} \tau_{rik} \big[ \beta_{rik}(\Gamt, \sigma^2) + \log \frac{\pi_{rk}}{\tau_{rik}}\big]\\
  &-\frac{1}{2} \sum_r d_r N_r \log \sigma_r^2 -\frac{K}2 \big\{ \log |\Sigma| + \tr[\Sigma^{-1} \Psib(\Gamt,\mu)] \big\}
+ \frac12\sum_k\log|\Sigt_k| + \text{const.}
\end{split}
\end{align}
where 
\begin{align}
\gamma_{ij}(\tau) := \E_{q_Z}(y_{ij}) = \sum_{k=1}^K \tau_{1ik} \tau_{2jk},\quad
%\begin{align}
%	\gamma_{ij}(\tau) &:= \E_{q_Z}(y_{ij}) = \sum_{k=1}^K \tau_{1ik} \tau_{2jk}, \label{eq:gam:def}\\
%	\beta_{rik}(\Gamt, \sigma^2) &:= 
%  - \frac1{2\sigma^2_r} \Big[\tr\big((\Sigt_k)_{rr}\big) + \|x_{ri} - \mut_{rk}\|^2 \Big], \label{eq:bet:def}\\
%  \Psib(\Gamt,\mu) &:= \frac{1}{K} \sum_{k=1}^K \big[\Sigt_k + (\mut_k - \mu) (\mut_k - \mu)^T \big]. \notag
%\end{align}
\beta_{rik}(\Gamt, \sigma^2) := 
		- \frac1{2\sigma^2_r} \Big[\tr\big((\Sigt_k)_{rr}\big) + \|x_{ri} - \mut_{rk}\|^2 \Big], 	\label{eq:bet:def}
\end{align}
$(\Sigt_k)_{rr}, r=1,2$ refers to the two diagonal blocks of $\Sigt_k$ of sizes $d_r \times d_r$, and 
\begin{align}\label{eq:S:def}
		\Psib(\Gamt,\mu) := \frac{1}{K} \sum_{k=1}^K \big[\Sigt_k + (\mut_k - \mu) (\mut_k - \mu)^T \big].
\end{align}

%\subsection{Optimizing ELBO}
\subsection{Optimizing the variational likelihood}\label{subsec:optimization}
We proceed to maximize $J$ by alternating between the likelihood parameters $(\mu,\Sigma,\sigma,\pi,\Psi)$ and variational parameters $(\tau,\Gamt)$. Each of these two sets of parameters is also optimized by alternating maximization. In other words, the overall optimization algorithm is a  block coordinate ascent. The key update is that of label distributions $\tau$, which we describe in details below. The other updates are more or less standard and detailed in Appendix~\ref{sec:sig2:etc}.

\subparagraph{Updating node labels ($\tau$).} To optimize $\tau$, we  use block coordinate ascent, by fixing $\tau_2 := (\tau_{2j}, j \in [N_2])$ and optimizing over $\tau_1 := (\tau_{1j}, j \in [N_1])$ and vice versa. Here we only consider optimization over $\tau_1$ given $\tau_2$.
To simplify notation, let $	h(p,q;\alpha) := g(p;\alpha) - g(q;\alpha)$.
%\begin{align*}
%	h(p,q;\alpha) := g(p;\alpha) - g(q;\alpha) = 
%	%\log \frac{p^{\alpha}(1-p)^{1-\alpha}}{q^\alpha(1-q)^{1-\alpha}}
%	\log \big(\frac{p}{q}\big)^{\alpha}
%		\big(\frac{1-p}{1-q}\big)^{1-\alpha}
%\end{align*}
Considering only the terms in $J$ that depend on $\tau$, we have 
\begin{align*}
	J = \sum_{ij} \gamma_{ij}(\tau)\, h(p,q;A_{ij}) + \sum_{r,i,k} \tau_{rik} \big[ \beta_{rik}(\Gamt, \sigma^2) + \log \frac{\pi_{rk}}{\tau_{rik}}\big] + \text{const.}
\end{align*} 
where const. collects terms that do not depend on $\tau$. Let $\xi_{rik} := \beta_{rik}(\Gamt, \sigma^2) + \log \pi_{rk}$. Using the definition of $\gamma_{ij}(\tau) 
= \sum_{k=1}^K \tau_{1ik} \tau_{2jk}$, 
\begin{align*}
	J = \sum_{ij}  \Big(\sum_k \tau_{1ik} \tau_{2jk}\Big) \, h(p,q;A_{ij}) + \sum_{r,i,k} \tau_{rik} \big[ \xi_{rik} - \log \tau_{rik}\big] + \text{const.}
\end{align*} 
Now, assume further that $\tau_2$ is constant. Then,
\begin{align*}
	J &= \sum_{i}  \sum_k  \tau_{1ik} \Big(\sum_j  \tau_{2jk}\, h(p,q;A_{ij}) \Big) + \sum_{i,k} \tau_{1ik} \big[ \xi_{1ik} - \log \tau_{1ik}\big] + \text{const.} \\
	&= \sum_{i} \sum_k    \tau_{1ik} \Big(\sum_j  \tau_{2jk}\, h(p,q;A_{ij}) + \xi_{1ik} - \log \tau_{1ik} \Big) +  \text{const.}
\end{align*}
where const. includes terms also dependent on $\tau_2$, but not on $\tau_1$. The cost function above is separable over $i$, and for each $i$ we have an instance of the problem given in the following lemma. Recall that $\Pc_K$ is the set of probability vectors in $\reals^K$.
\begin{lem}\label{lem:tau:update}
 For any nonnegative vector $(a_1,\dots,a_K)$, let $f_a: \Pc_K \to \reals$ be defined by $f_a(p) := \sum_{k=1}^K p_k (a_k   - \log p_k)$. Then the maximizer of $f_a$ over $\Pc_K$ is given by the softmax operation:
	\begin{align}\label{eq:softmax:op}
		\argmax_{p \in \Pc_k} f_a(p) =  \softmax(a) := \frac{e^{a_k}}{\sum_\ell e^{a_{\ell}}}.
	\end{align}
\end{lem}
%For the proof see Appendix~?.
\begin{proof}
	We have $f_a(p) = -\sum_k p_k \log(p_k/e^{a_k})$. If $\sum_k e^{a_k} = 1$, then $-f_a(p)$ is the KL-divergence between $(p_k)$ and $(e^{a_k})$ and the result follows. Otherwise, normalizing only adds a constant to $f_a$, that is, with $C = 1/\sum_k e^{a_k}$ and $q_k = C e^{a_k}$, we have  $f_a(p) = -\sum_k p_k \log(p_k/q_k) -\log C $ and the result follows.
\end{proof}
We write the solution of Lemma~\ref{lem:tau:update} simply as $p_k \propto_k e^{a_k}$ where $\propto_k$ means proportional as a function of $k$. Then, $\tau_1$ update is $\tau_{1ik} \propto_k \exp \big[ \sum_j  \tau_{2jk}\, h(p,q;A_{ij}) + \xi_{1ik}  \big]$, or after unpacking $ \xi_{1ik}$,
\begin{align}\label{eq:tau:update:v1}
	\tau_{1ik} 
	%\;&\propto_k\; \exp \Big[ \sum_j  \tau_{2jk}\, h(p,q;A_{ij}) + \xi_{1ik}  \Big],  \\
	\;&\propto_k\; \pi_{1k} \exp \Big[ \sum_j  \tau_{2jk}\, h(p,q;A_{ij}) + \beta_{1ik}(\Gamt, \sigma^2)  \Big]\quad i=1,\dots,N_1.
\end{align}
The update for $\tau_2$ is similar. %See Algorithm~?.

%\begin{rem}[Improving the speed]
%    For binary likelihood, we can write $h(p,q;\alpha) = \alpha \phi_1 + \phi_0$.
%	where $\phi_1 = \log \frac{p(1-q)}{q(1-p)}$ and $\phi_0 = \log \frac{1-p}{1-q}$. Then, in matrix notation 
%%	\begin{align*}
%%	\tau_{1ik}
%%	\;&\propto_k\ \pi_{1k} \exp \Big[ \phi_1 \sum_j  \tau_{2jk}A_{ij} + \phi_0\tauc_{2k}+ \beta_{1ik}(\Gamt, \sigma^2)  \Big]\quad i=1,\dots,N_1
%%	\end{align*}
%%	or in matrix notation
%	\begin{align*}
%	\tau_{1ik}
%	\;&\propto_k\ \pi_{1k} \exp \Big(\phi_1 [A \tau_2]_{ik} + \phi_0\tauc_{2k}+ \beta_{1ik}(\Gamt, \sigma^2)  \Big)\quad i=1,\dots,N_1
%	\end{align*}
%	where $[A \tau_2]_{ik} = \sum_j \tau_{2jk} A_{ij}$. When $A$ is sparse, the matrix-vector product $A\tau_2$ can be computed quite fast.
%\end{rem}

%\subsection{Updating $\Sigt$ and $\mut$}
%\label{subsec:Sigt:mut}
\subparagraph{Updating $\Sigt$ and $\mut$.}
Let us define $\tauc_{rk} := \sum_{i=1}^{N_r}\tau_{rik}$ and $D^{-1}_k := \diag\big( \frac{\tauc_{1k}}{\sigma^2_1}I_{d_1}, \frac{\tauc_{2k}}{\sigma^2_2}I_{d_2}\big)$
%\begin{align}\label{eq:tau:check:def}
%  \tauc_{rk} = \sum_{i=1}^{N_r}\tau_{rik}, \quad 
%  \text{and} \quad
%  D^{-1}_k = \diag\Big( \frac{\tauc_{1k}}{\sigma^2_1}I_{d_1}, \frac{\tauc_{2k}}{\sigma^2_2}I_{d_2}\Big)
%\end{align}
Then, as a function of $\Sigt$, $J$ can be written as (see Appendix~\ref{sec:details:Gamt})
\begin{align}\label{eq:J:Sigt:expr}
  J(\Sigt)\doteq -\frac12 \sum_k \tr\big[(D^{-1}_k +\Sigma^{-1})\Sigt_k\big] -\log|\Sigt_k|.
\end{align}
This is separable over $k$, with each term being the likelihood of a multivariate Gaussian with covariance parameter. The maximizers are then simply $ \Sigt_k = (D^{-1}_k +\Sigma^{-1})^{-1}$ for $k \in [K]$.
%\begin{align}\label{eq:Sigt:update}
%  \Sigt_k = (D^{-1}_k +\Sigma^{-1})^{-1}, \quad k \in [K].
%\end{align}
%

To derive the updates for $\mut$, let $\xc_{rk} := \sum_{i=1}^{N_r} \tau_{rik} x_{ri}$ and $\muc_{rk} := \xc_{rk}/\tauc_{rk}$.
%\begin{align}\label{eq:x:check:def}
% \xc_{rk} := \sum_{i=1}^{N_r} \tau_{rik} x_{ri}, \quad \muc_{rk} := \frac{\xc_{rk}}{\tauc_{rk}}
%\end{align}
Then, as a function of $\mut$,  $J$ can be written as (see~Appendix~\ref{sec:details:Gamt}):
\begin{align}\label{eq:J:mut:expr}
J(\mut) \doteq -\frac12\sum_k (\mut_k - m_k)^T (D_k^{-1}+\Sigma^{-1}) (\mut_k - m_k) 
  \end{align}
where $m_k = (D_k^{-1}+\Sigma^{-1})^{-1} (D_k^{-1}\muc_k +\Sigma^{-1}\mu)$. It is clear that the optimal value of $\mut_k$ is equal to $m_k$, which using the optimal value of $\Sigt_k$, can be written as $\mut_k = \Sigt_k (D_k^{-1}\muc_k +\Sigma^{-1}\mu)$.
%\begin{align}\label{eq:mut:update}
%  \mut_k = \Sigt_k (D_k^{-1}\muc_k +\Sigma^{-1}\mu).
%\end{align}

% So the optimal value of $\mut_k = m_k = (D_k^{-1}+\Sigma^{-1})^{-1} (D_k^{-1}\muc_k +\Sigma^{-1}\mu)$. \\
% Note that based on our derivation in the last section, the optimal value of $\Sigt_k$ was found to be $(D_k^{-1}+\Sigma^{-1})^{-1}$. 

% \subsection{$\mut$}

%\subsection{Updating $\Sigma$, $\mu$ and $\sigma^2$}
\subparagraph{Updating $\Sigma$, $\mu$ and $\sigma^2$.}
As a function of $\Sigma$, we have $J(\Sigma) \doteq -\frac{K}2 \big[ \log |\Sigma| + \tr(\Sigma^{-1} \Psib(\Gamt)) \big]$ which is the standard Gaussian likelihood, giving the optimal value $\Sigma = \Psib(\Gamt,\mu)$.
%\begin{align}\label{eq:Sigma:update}
%  \Sigma &= \Psib(\Gamt,\mu) = \frac{1}{K} \sum_{k=1}^K \Psi_k (\Gamt) = \frac{1}{K} \sum_{k=1}^K \big[\Sigt_k + (\mut_k - \mu) (\mut_k - \mu)^T\big].
%\end{align}
Similarly, as a function of $\mu$, $ J(\mu) \doteq -\frac{K}2 \big[\tr(\Sigma^{-1} \Psib(\Gamt)) \big] \doteq
  -\frac12 \sum_k \big[(\mut_k - \mu)^T  \Sigma^{-1} (\mut_k - \mu)\big]$
giving the optimal solution $\mu = \frac{1}{K} \sum_k \mut_k$.
%\begin{align}\label{eq:mu:update}
%    \mu &= \frac{1}{K} \sum_k \mut_k.
%\end{align}
%
% Now we derive the optimal value for $\Sigma$:
% \begin{align*}
% J(\Sigma) &= -\frac{K}2 \big[ \log |\Sigma| + \tr(\Sigma^{-1} \Psib(\Gamt)) \big]\\
% & = \frac{K}2 \big[ \log |\Sigma^{-1}| - \tr(\Sigma^{-1} \Psib(\Gamt)) \big]\\
% \frac{\partial J(\Sigma)}{\partial \Sigma^{-1}}&= \frac{K}{2}\big[\Sigma - \Psib(\Gamt)\big] = 0\\
% \text{So};\ 
% \end{align*}
%
The update for $\sigma^2 = (\sigma^2_1,\sigma^2_2)$ can be easily obtained too (see Appendix~\ref{sec:sig2:etc})
\begin{align}\label{eq:sig2:update}
     \sigma^2_r &= \frac1{N_r d_r}\Big[\sum_k \tauc_{rk}\tr\big((\Sigt_k)_{rr}\big)+\sum_{i,k} \tau_{rik}\|x_{ri} - \mut_{rk}\|^2 \Big], \quad r = 1,2.
\end{align}

%\subsection{Updating $\pi$ and $Q = (p,q)$}
%\bigskip
\subparagraph{Updating $\pi$ and $\Psi \equiv (p,q)$.}
Updating these parameters is standard (See Appendix~\ref{sec:sig2:etc}): 
\begin{align}\label{eq:pi:Q:update}
  \pi_r= \frac{(\tauc_{r1} ,\ldots, \tauc_{rK})}{\sum_k \tauc_{rk}},\;r =1,2, \quad p = \frac{\sum_{ij} \gamma_{ij}(\tau) A_{ij}}{\sum_{ij}\gamma_{ij}(\tau)}, \quad q= \frac{\sum_{ij} \big(1-\gamma_{ij}(\tau)\big) A_{ij}}{\sum_{ij} \big(1-\gamma_{ij}(\tau)\big)}.
\end{align}
%where $\tauc_{rk}$ is defined in~\eqref{eq:tau:check:def} and $\gamma_{ij}(\tau)$ defined in~\eqref{eq:gam:def}..

%\paragraph{Matrix form of the updates.}
\subsubsection{Improving the speed}
	For $r=0,1$, we treat each $(\tau_{rik})_{ik}$ as a matrix $\tau_r \in [0,1]^{N_r \times K}$. The $\tau$-update in~\eqref{eq:tau:update:v1} can be simplified to improve computational complexity for sparse networks $A$.
	We can write $h(p,q;\alpha) = \alpha \phi_1 + \phi_0$,
	where for the binary likelihood, $\phi_1 = \log \frac{p(1-q)}{q(1-p)}$ and $\phi_0 = \log \frac{1-p}{1-q}$, and for the Poisson likelihood considered in Section~\ref{sec:dc:extension} below, $\phi_0 = q-p$ and $\phi_1 = \log(p/q)$. Then, in matrix notation 
	%	\begin{align*}
	%	\tau_{1ik}
	%	\;&\propto_k\ \pi_{1k} \exp \Big[ \phi_1 \sum_j  \tau_{2jk}A_{ij} + \phi_0\tauc_{2k}+ \beta_{1ik}(\Gamt, \sigma^2)  \Big]\quad i=1,\dots,N_1
	%	\end{align*}
	%	or in matrix notation
	\begin{align*}
	\tau_{1ik}
	\;&\propto_k\ \pi_{1k} \exp \Big(\phi_1 [A \tau_2]_{ik} + \phi_0\tauc_{2k}+ \beta_{1ik}(\Gamt, \sigma^2)  \Big)\quad i=1,\dots,N_1
	\end{align*}
	where $[A \tau_2]_{ik} = \sum_j \tau_{2jk} A_{ij}$. When $A$ is sparse, the matrix-vector product $A\tau_2$ can be computed quite fast. Letting $\beta_{r} = (\beta_{rik})_{ik} \in \reals^{N_r \times K}$, we have the $\tau$-update in vector form:
	\begin{align}\label{eq:tau:update:matrix:form}
		\tau_1 = \texttt{row-softmax}\big[ \phi_1 A \tau_{2} + \phi_0 \onev_{N_1} (\tauc_2 + \log \pi_1)^T + \beta_1 \big]
	\end{align}
	and similarly for $\tau_2$. Here, \texttt{row-softmax} is the row-wise softmax operator, applying~\eqref{eq:softmax:op} to each row of a matrix.
	
	\medskip
	Further improvements are possible in estimating $p$ and $q$. Note that we can write $\tauc_{rk} := [\tau_r^{T}\onev_{N_r}]_k$. Let us treat $\tauc_r$ as a $K$-vector, with elements $\tauc_{rk}$. Then, we have  $\sum_{ij} \gamma_{ij}(\tau) A_{ij} = \tr(\tau_1^T A \tau_2)$ and $\sum_{ij} \gamma_{ij}(\tau) =\ip{\tauc_1, \tauc_2}$ , and
	\begin{align}\label{eq:improved:p:q:estim}
		p = \frac{\tr(\tau_1^T A \tau_2)}{\ip{\tauc_1, \tauc_2}}, \quad q = \frac{r \rho - p}{r - 1}, \quad \text{where}\; \frac1r = \ip{\frac{\tauc_1}{N_1},\frac{\tauc_2}{N_2}} \; \text{and}\; 
		\rho = \frac1{N_1N_2} \sum_{ij} A_{ij}.
	\end{align}
	Note that $\rho$ is the density of the graph (or $A$) and that $\ip{\tauc_1, \tauc_2} = \tr(\tau_1^T E \tau_2)$ where $E$ is the all-ones matrix of appropriate dimension. Finally, let us define $\beta'_{rik} = \tr\big((\Sigt_k)_{rr}\big) + \|x_{ri} - \mut_{rk}\|^2$ noting that $\beta_{rik} = \frac1{2\sigma_r^2}\beta'_{rik}$ from definition~\eqref{eq:bet:def}. Letting $\beta'_r = (\beta'_{rik})_{ik}$ be its matrix form, we can write the update for $\sigma_r^2$ compactly as
	\begin{align}\label{eq:improved:sigma2:estim}
		\sigma_r^2 = \frac1{N_r d_r} \sum_{ik} \tau_{rik} \beta'_{rik} = 
		\frac{1}{N_r d_r} \tr(\tau_r^T \beta'_{r}).
	\end{align}

\subsection{Extension to the degree-corrected case}\label{sec:dc:extension}
%The general log-likelihood based on $A$ only will be
%\begin{align*}
%\sum_{ij} \sum_{k\ell} z_{1ik} z_{2j\ell}
%\;g(\Psi_{k\ell},A_{ij})
%\end{align*}
%where $g(p,\alpha) =  \alpha \log \frac{p}{1-p} + \log (1-p)$. In the special case, by breaking up over $k = \ell$ and $k \neq \ell$ we obtain
%\begin{align*}
%\sum_{ij} \Big[ \sum_{k} z_{1ik} z_{2jk}
%\;g(p,A_{ij}) + \sum_{k\neq \ell} z_{1ik} z_{2j\ell}
%\;g(q,A_{ij}) \Big] = 	\sum_{ij} \Big[y_{ij} g(p,A_{ij}) + (1-y_{ij}) g(q,A_{ij})\Big]
%\end{align*}
%where we have used $\sum_{k\ell} z_{1ik} z_{2j\ell} = 1$.
In this case the network-dependent part of the likelihood is replaced with
$$\ell_1(\Psi,\theta ) =  \sum_{ij} \sum_{k\ell} z_{1ik} z_{2j\ell}
\;g(\theta_{1i} \theta_{2j}\Psi_{k\ell},A_{ij}).$$
 Again, we focus on the case where $\Psi_{kk} = p$ and $\Psi_{k\ell} = q$ for $k \neq \ell$. Recalling the notation $h(p,q,\alpha) = g(p,\alpha) - g(q,\alpha)$, we have
\begin{align}\label{eq:ell1:dc}
\ell_1(\Psi,\theta )  = \sum_{ij} \big[y_{ij} \,h\big(p \theta_{1i} \theta_{2j},\,q\theta_{1i} \theta_{2j},\, A_{ij}\big) + g\big(\theta_{1i} \theta_{2j} q,A_{ij} \big)\big].
\end{align}
 We assume a Poisson log-likelihood with $g(p,\alpha) = \alpha \log p - p$ for which $h(p,q,\alpha) = \alpha \log (p/q) + q-p$. We also recall the normalization assumption~\eqref{eq:theta:normalization}, $\sum_{i} \theta_{ri} z_{rik} = \sum_{i}z_{rik}$, which implies $\sum_{ij} y_{ij} \theta_{1i} \theta_{2j} = \sum_{ij} y_{ij}$ and $\sum_{ij} \theta_{1i}\theta_{2j} = N_1N_2$. Using these two implications, the first term in~\eqref{eq:ell1:dc} simplifies to
\begin{multline*}
\sum_{ij} y_{ij}\big[ (q-p) \theta_{1i} \theta_{2j} + A_{ij} \log(p/q)\big] + \sum_{ij} \big[ {-\theta_{1i}\theta_{2j} q} + A_{ij} \log (\theta_{1i}\theta_{2j} q)\big]\\ = \sum_{ij} y_{ij}\big[ (q-p)+ A_{ij} \log(p/q)\big]  - qN_1N_2 + \sum_{ij}  A_{ij} \log (\theta_{1i}\theta_{2j} q)
\end{multline*}
%where we have used the implications of the constraints: $\sum_{ij} y_{ij} \theta_{1i} \theta_{2j} = \sum_{ij} y_{ij}$, and $\sum_{ij} \theta_{1i}\theta_{2j} = N_1N_2$. 
Let $\phi_0 = q-p$ and $\phi_1 = \log(p/q)$. 

\subparagraph{$\tau$-update.} Let us fix $\theta$ and obtain updates for the label posteriors $\tau$. Taking expectations of the objective and the constraints, the $\tau$-portion of the update is equivalent to maximizing 
\begin{align*}
\sum_{ij} \gamma_{ij}(\tau) \big[\phi_0 + A_{ij} \phi_1 \big] 
+ \sum_{rik} \tau_{rik} \big[ \beta_{rik}(\Gamt, \sigma^2) + \log \frac{\pi_{rk}}{\tau_{rik}}\big]
\end{align*}
subject to constraints $\sum_{i} \tau_{rik} (\theta_{ri} -1) = 0$ for all $k$. Note that these constraints follow by taking expectations of the normalization constraints~\eqref{eq:theta:normalization} under $Z \sim q$. Focusing on updating $\tau_{1k}$, we have the following optimization problem:
\begin{align}\label{eq:dc:tau:update:optim}
\begin{array}{ll}
\displaystyle \max_{\tau_1} & \sum_{i,k}    \tau_{1ik} \big(\sum_j  \tau_{2jk}\big[\phi_0 + A_{ij} \phi_1 \big]  + \xi_{1ik} - \log \tau_{1ik} \big)  \\	
\text{subject to} & \sum_{i} \tau_{1ik} (\theta_{1i} -1) = 0
\end{array}
\end{align}
where $\xi_{1ik} = \beta_{1ik} + \log \pi_{1k}$ as before.
In Appendix~\ref{sec:dc:tau:details}, we derive a dual ascent algorithm for solving this problem with the following updates:
\begin{align}\label{eq:dc:tau:update}
	\begin{split}
	\tau_1(\lambda) &= \texttt{row-softmax}\big[ \phi_1 A \tau_{2} + \phi_0 \onev_{N_1} (\tauc_2 + \log \pi_1)^T + \beta_1  + (\theta_1-\onev) \lambda^T\big],\\
	\lambda^{+} &= \lambda - \mu [\tau_1(\lambda)]^T (\theta_1-\onev).
	\end{split}
\end{align}
Here, $\lambda\in \reals^K$ is the dual variable, $\lambda^+$ is its update, $\theta_1 =(\theta_{1i})\in \reals^n$, and $\mu$ is a proper step-size. These two iterations are repeated till convergence, before updating other parameters. Note that when $\theta_1 = \onev$, the dual ascent algorithm reduces to the single step of~\eqref{eq:tau:update:matrix:form} obtained for the case without degree correction.

\subparagraph{$\theta$-update.}
Let us now fix $\tau$ and the rest of the parameters and optimize over $\theta$. The relevant portion of the objective function is 
\begin{align*}
\sum_{ij} \gamma_{ij}(\tau) \big[ q-p + A_{ij} \log(p/q)\big] - qN_1N_2 + \sum_{ij}  A_{ij} \log (\theta_{1i}\theta_{2j} q).
\end{align*}
Consider optimizing over $(\theta_{1i})$, which is equivalent to maximizing $\sum_{i}  d_{1i} \log \theta_{1i}$, subject to $\sum_{i} \tau_{1ik}(\theta_{1i}-1) = 0$ for all $k$, and $\theta_{1i} \ge 0$ for all $i$. This problem is suitable for an application of the Douglas--Rachford (DR) splitting algorithm~\cite{Douglas1956,OConnor2014}. Let $f_t(\cdot ; d) : \reals_+^n \to \reals_+^n$ with $d \in \reals_+^n$ and $t > 0$, be defined by 
\begin{align}\label{eq:prox:of:log}
	[f_t(x;d)]_i := \frac12 \big[ x_i +  \sqrt{ x_i^2 + 4 t d_i } \big].
\end{align}
Also, let $H_1:= \tau_1 (\tau_1^T\tau_1)^{-1} \tau_1^T$ be the projection operator onto the span of $\tau_1 \in \reals^{N_1 \times K}$. The algorithm performs the following iterations for updating $(\xi_1,\theta_1)$ to $(\xi_1^+,\theta_1^+)$:
\begin{align}\label{eq:dc:theta:update}
	\begin{split}
	\theta_1^+ &= f_t(\xi_1; d_1) \\
	\xi_1^+ &= \theta_1^+ - H_1(2 \theta_1^+ - \xi_1 - \onev) 
	\end{split}
\end{align}
where $\xi_1 \in \reals^{N_1}$ is an auxiliary variable, $d_1 = (d_{1i}) \in \reals^{N_1}$ collects the degrees of side $1$,  and $t > 0$ is the fixed parameter of DR algorithm (often set to $1$). The details for the derivation of this algorithm can be found in Appendix~\ref{sec:dc:theta:details}. The same updates apply to $\theta_2$, replacing subscript 1 with 2.

\begin{rem}
	Note that if $\tau_1 = (\tau_{1ik})$ was a hard label assignment, then the optimization for $(\theta_{1i})$ would have a simple solution. To see this, let $C_{k}(\tau_1)$ be the $k$th cluster of hard label $\tau_1$. Then, the optimal value of $\theta_1$ is given by 
	\begin{align*}
		\theta_{1i} = \frac{d_{1i}}{ \sum_{i' \in C_k(\tau_1)} d_{1i'} }, \quad \text{for}\; i \in C_{1k}(\tau_1).
	\end{align*}
	This is in fact, the choice in profile-likelihood approaches to fitting DC-SBM, where one replaced $\theta_1$ with this optimal value, in addition to optimal values of edge probabilities and class priors, all in terms of $\{C_{1k}(\tau_1)\}$, and then optimize the resulting profile likelihood over $\{C_{1k}(\tau_1)\}$. See for example ~\cite{Karrer2011}. Our approach here, allows us to keep a soft-label assignment $\tau_1$ throughout the algorithm, viewing optimization over $\theta_1$ as another phase of block-coordinate ascent for the overall constrained optimization problem.
\end{rem}

\subparagraph{$(p,q)$-update.}
%\paragraph{Optimizing over $\theta_{ri}$ and $(p,q)$.}
%Here, $d_{1i} = \sum_{j} A_{ij}$. Optimizing the Lagrangian $\sum_{i}  \big[d_{1i} \log \theta_{1i} - \sum_k \xi_k \tau_{1ik}(\theta_{1i}-1)\big]$ (keeping the constraint $\theta_{1} \ge 0$), via a dual ascent algorithm leads to the following updates
%\begin{align*}
%\theta_{1i}(\xi) = \max\Big\{0,\frac{d_i}{\sum_k \xi_k \tau_{1ik}} \Big\}, \quad \xi_k^+ = \xi_k +\alpha \sum_i \tau_{1ik}\big[\theta_{1i}(\xi)-1\big], \; k \in [K]
%\end{align*}
To optimize over $p$ and $q$ we note that because of the Poisson model, $p$ and $q$ are not tied together and the only constraint we have is $p,q \ge 0$. Optimizing over $p$ is equivalent to maximizing   $-p \sum_{ij}\gamma_{ij} + \log p \sum_{ij}\gamma_{ij}A_{ij}$ and optimizing over $q$, is equivalent to maximizing over $-q \sum_{ij}(1-\gamma_{ij}) + \log q \sum_{ij}(1-\gamma_{ij})A_{ij}$, both giving the same updates as those in~\eqref{eq:pi:Q:update}.
%\begin{align*}
%p = \frac{\sum_{ij}\gamma_{ij}A_{ij}}{ \sum_{ij}\gamma_{ij}}, \quad q = \frac{\sum_{ij}(1-\gamma_{ij})A_{ij}}{\sum_{ij}(1-\gamma_{ij})}
%\end{align*}
%as before.

\begin{algorithm}[t]\label{alg:1}
	\caption{Variational block coordinate ascent for fitting \mbisbm}
	\label{CHalgorithm}
	\begin{algorithmic}[1]
		%\Input Incidence matrix $A$ and covariate matrices $X_1$ and $X_2$.
		\State Initialize $\tau_r$  using \bisc, and $\theta_r = \onev_{N_r}$ for $r=1,2$. Pick tolerance $\eps \in (0,1]$.
		\State Initialize $\Sigma, \Sigt_k$ with $I_{d_1+d_2}$ and $\mu, \mut_k$ with $0$, for $k \in [K]$, and $\sigma_r^2=1$ for $r=1,2$.
		\While {not \texttt{CONVERGED}, nor maximum iterations reached}
		
		\State Update $(p,q)$ using~\eqref{eq:improved:p:q:estim} and $\pi_r, r=1,2$ using~\eqref{eq:pi:Q:update}.
		\State Update $(\phi_0,\phi_1) \gets (q-p, \,\log (p/q))$. 
		\If {DC-version}
		\State Update $\theta_r$ by repeating~\eqref{eq:dc:theta:update} till convergence.
		\EndIf
		
		\State Update $\beta_r, r=1,2$ using~\eqref{eq:bet:def}.
		\State $\tau_r^{\text{old}} \gets \tau_r$, $r=1,2$.
		\State Update $\tau_1$ by repeating~\eqref{eq:dc:tau:update} till convergence.
		\State Update $\tau_2$ by repeating~\eqref{eq:dc:tau:update}, with subscripts 1 and 2 switched and $A$ replaced with $A^T$, till convergence.
		\medskip
		\State Update the following for for $r=1,2$ and $k \in [K]$: \Comment{Update parameters}
		\State $\tauc_{rk} \gets \sum_{i=1}^{N_r}\tau_{rik}$, and $D^{-1}_k \gets \diag\big(\tauc_{1k}I_{d_1}/\sigma^2_1,
		\tauc_{2k}I_{d_2}/\sigma^2_2\big)$, 
		\State
		$\xc_{rk} \gets \sum_{i=1}^{N_r} \tau_{rik} x_{ri}$, and $\muc_{rk} \gets \xc_{rk}/\tauc_{rk}$.
		\State Update $ \Sigt_k \gets (D^{-1}_k +\Sigma^{-1})^{-1}$ and $\mut_k \gets \Sigt_k (D_k^{-1}\muc_k +\Sigma^{-1}\mu)$.
		\State Update $\mu \gets \frac1K \sum_k \mut_k$ and $\Sigma \gets \frac{1}{K} \sum_{k=1}^K \big[\Sigt_k + (\mut_k - \mu) (\mut_k - \mu)^T \big]$.
		\State Update $\sigma_r^2, r=1,2$ using~\eqref{eq:improved:sigma2:estim}.
		%	Evaluate $\beta_{rik}(\Gamt, \sigma^2)$ using~\eqref{eq:bet:def}.
		\State \texttt{CONVERGED} $  \gets \big[\max\{\delta_1,\delta_2\} < \eps/K\big]$,\; where $\delta_r := \mnorm{\tau_r -\tau_r^{\text{old}}}_\infty, r=1,2$
		\EndWhile
	\end{algorithmic}
	\label{alg:mbisbm}
\end{algorithm}

\subsection{Summary of the algorithms}
Algorithm~\ref{alg:mbisbm} summarizes the updates for fitting the proposed matched bipartite SBM model, to which we refer as \mbisbm. We have stated the general form of the algorithm with degree correction (DC) and covariates. Note for example that if no degree-correction is desired, $\theta_r$ remains equal to $\onev_{N_r}$ and the iterations~\eqref{eq:dc:tau:update} in steps 11 and 12, for updating the label distributions ($\tau_r$), automatically reduce to the simple update~\eqref{eq:tau:update:matrix:form} (that is, the iterations converge in one step.). There are other variations available. For example, if desired, step~5 can be replaced with $(\phi_0,\phi_1) \gets(\log \frac{1-p}{1-q}, \log\frac{p(1-q)}{q(1-p)})$ to use values based on a Bernoulli likelihood instead of a Poisson. Empirically, we have not found much difference between the two. With minor modifications, the algorithm can be used when only one side has covariates or without covariates for either side.% \aaa{Need to test reordering params. estimate.}

%If there are no node covarites in the bipartite network, we either can set $\Sigma = 0$ and use the first algorithm, or use the simplified form detailed in Algorithm~\ref{alg:2}.

%\begin{algorithm}[H]\label{alg:2}
%\caption{Variational block coordinate ascent when there are no node covarites}
%\label{NOXalgorithm}
%\begin{algorithmic}[1]
%%\Procedure{CH\textendash Election}{}
%\Function{Update variational posterior}{$\pi,(p,q)$}
%
%\State Update $\tau_{1i}$ using:
%  %
%  $\tau_{1ik} \,\propto_k\,  \pi_{1k} \exp \big[ \sum_j  \tau_{2jk}\, h(p,q;A_{ij}) \big],\; i=1,\dots,N_1$.
%
%\State Update $\tau_{2i}$ using:
%  %
%  $\tau_{2ik} \,\propto_k\,  \pi_{2k} \exp \big[ \sum_j  \tau_{1jk}\, h(p,q;A_{ij})  \big],\; i=1,\dots,N_2$.
%
%\Return $\tau$
%\EndFunction
%
%\Function{Update likelihood parameters}{$\tau$}
%
%\State Update $\pi_r, r=1,2$ and $(p,q)$, in any order, using~\eqref{eq:pi:update} and~\eqref{eq:Q:update}.
%
%% \State Update $\pi_{rk} =\frac{\tauc_{1k}}{\sum_l \tauc_{1l}} \quad r = 1,2$.
%
%% \State Update $p = \frac{\sum_{ij} \gamma_{ij}(\tau) A_{ij}}{\sum_{i,j}\gamma_{ij}(\tau)}$.
%
%% \State Update $q = \frac{\sum_{i,j} \big(1-\gamma_{ij}(\tau)\big) A_{ij}}{\sum_{i,j} \big(1-\gamma_{ij}(\tau)\big)}$.
%
%\Return $\pi,(p,q)$
%\EndFunction
%\end{algorithmic}
%\end{algorithm}

%The initialization of the algorithms is discussed in Section~?.
\subsubsection{Initialization of the algorithms}\label{subsubsec:initialization}

It is known that variational inference is sensitive to initialization \cite{Blei2016}.
% Thus, it is important to choose the choose the initial labels carefully. 
%
The main component of the algorithm that needs careful initialization is the matrix of (approximate) posterior node labels $\tau = [\tau_1;\tau_2]$.
We propose to initialize $\tau$ using a bipartite spectral clustering algorithm, \bisc for short, which is a variant of the approach of~\cite{Dhillon2001}. The difference between our version and that of~\cite{Dhillon2001} is that  \cite{Dhillon2001} does not normalize the rows of the singular vectors and keeps top $\ceil{\log_2 K}$
singular vectors, as opposed to $K$. We have found that row normalization greatly improves the performance, and it is fairly standard in usual (non-bipartite) Laplacian-based spectral clustering. Algorithm~\ref{alg:bisc} summarizes our version.

\newcommand\Ut{\widetilde{U}}
\newcommand\Vt{\widetilde{V}}
\begin{algorithm}[t]
	\caption{Bipartite Spectral Clustering (\bisc)}
	%\label{BiSC}
	\begin{algorithmic}[1]
		%\Procedure{CH\textendash Election}{}
		{}
		\State Input: bi-adjacency matrix $A \in \{0,1\}^{N_1 \times N_2}$.
		\State Let $D_1 = \diag(\sum_{j}A_{ij}, i=1,\dots,N_1)$ and $D_2 = \diag(\sum_{i} A_{ij}, j=1,\dots,N_2)$.
		%be the diagonal matrices of the degrees of the two sides.
		\State Form  $L = D_1^{-1/2} A D_2^{-1/2}$.
		\State Let $L = USV^T$ be the SVD of $L$ truncated to $K$ largest singular values ($U \in \reals^{N_1 \times K}$ and $V \in \reals^{N_2 \times K}$). %Compute $K$ left and right singular vectors of $A_n$, $u_1,\cdots, u_{K}$ and $v_1, \cdots, v_{K}$ 
		
		\State  Normalize each row of $U$ and $V$ to unit $\ell_2$ norm to get $\Ut$ and $\Vt$, resp., then form
		\begin{align*}
			Z = \begin{bmatrix}
				D_1^{-1/2} \Ut \\
				D_2^{-1/2} \Vt 
				\end{bmatrix}. %\in \reals^{(N_1+N_2) \times K}
		\end{align*}
			
		\State  Run $k$-means  with $K$ clusters on the rows of $Z$.
	\end{algorithmic}
	\label{alg:bisc}
\end{algorithm}

In simulation studies, we also consider a couple of competing initializations. One interesting choice is to use the usual Laplacian-based spectral clustering, which is oblivious to the bipartite nature of the problem. For this choice, we use the regularized version described in~\cite{Amini2013} as \scp. Note that \scp will be applied to the (symmetric) adjacency matrix $\At$; see~\eqref{eq:At:Zt:Psit}. It is also possible to regularize \bisc using similar ideas, though surprisingly, we found the simple unregularized version of \bisc is quite robust, and we have used this simple version when reporting results.

 When working with simulated data, since we have access to the true labels, we will also consider a perturbed version of truth as an initialization. Specifically, we generate from a mixture of the true label distribution and Dirichlet noise, i.e. $\tau_{ri} = \omega z_{ri} +(1- \omega) \eps_{ri}$ where $\eps_{ri} \sim \text{Dir}(.5 {\bf{1}}_K)$. 
 %\aaa{Not a true mixture.} 
 Here, we treat $z_{ri}$, the true label of node $i$ in group $r$, as a distribution on the $K$ labels. Parameter $\omega \in [0,1]$ measures the degree of initial perturbation towards noise. For example, with $\omega = 0.1$, about $10\%$ of the initial labels are correct. We will refer to this initialization as \texttt{$\sim$rnd}, for approximately random. This initialization will act as a proxy for a ``good enough'' initialization and allows us to study the behavior of our variational inference procedure decoupled from specific initializations produced by spectral clustering (or other methods). 
 
 Let us say a few words about the initialization of other parameters. The algorithm is moderately sensitive to the initialization of $p$, $q$ and $\pi_r, r=1,2$. When the quality of the initial labels ($\tau_1$ and $\tau_2$) is good, one can initialize these parameters, based on $(\tau_r)$, by running the corresponding updates first, as is done in Algorithm~\ref{alg:mbisbm}, lines 4--5. This is the form we suggest in practice when using the \bisc initialization, and is used in the real data application (Section~\ref{sec:realdata}). However, when the quality of the initial labels is not good, for example, when using \scp in the simulations, $p$ and $q$ obtained based on initial $(\tau_r)$ can become quite close leading to numerical instability. We have found in those cases that initializing these parameters with fixed values, say $(p,q) = (0.1,0.01)$ and $\pi_r$ set to uniform distribution of $[K]$, greatly improves the stability of the algorithm. (This is since even one iteration of the algorithm could significantly improve upon initial labels.) This fixed initialization of $(p,q,\pi_r)$, independent of $\tau_r$ is what we have used in Monte Carlo simulations on synthetic data, when comparing different label initializations (Section~\ref{sec:simulation}).
 
 %\aaa{Need link to the code.}
 %we show simulations with $\omega = 0.99$, where noise effectively dominates the initialization.

%The algorithm is not as sensitive to initial values of other parameters.
%One could initialize means by zero and covariance matrices by identity. Initial value of $\tau$ can be used to estimate other parameters of the network such as $(p,q)$ by running the corresponding updates first. Although, we have found that initializing $(p,q)$ to fixed values, say $(0.1,0.01)$ gives superior results. We refer to the code for more details.

%!TEX root = sbm_bip_arxiv.tex
\section{Simulations}\label{sec:simulation}
In this section, we  show that effectiveness of our proposed algorithm in recovering the true labels in synthetic bipartite networks. For the most part, we generate data from our proposed model~\eqref{eq:model:1}--\eqref{eq:model:2}. In the plots investigating the degree-corrected version of the algorithm, we generate from the degree-corrected version of the network described in subsection~\ref{sec:dc:model}.

%Throughout the simulations, we have initialized with $(p,q) = 0.01(8,1)$, $\Sigma = I_{d_1 + d_2}$, $\mu = 0$, $\sigma^2_r = 10$, for $r=1,2$, and $\Sigt_k = 0$, for $k \in [K]$.

\subparagraph{Data generation.}
Key parameters regarding covariate generation in~\eqref{eq:model:1} are $(\mu,\Sigma)$ for generating $v_{*k}$. We take $\mu = 0$ and $\Sigma = \nu I_{d_1 + d_2}$ throughout. Varying $\nu$ (or dimensions $d_r$) changes the information provided by the covariates (Appendix~\ref{sec:extra:sim}). Larger $\nu$ causes $v_{*k}$ to be further apart, hence covariates are more informative. $\nu = 0$ corresponds to zero covariate information. We also fix covariate noise levels at $\sigma_r = 0.5$ for $r=1,2$, and the network size at $N= (N_1,N_2) = (200,800)$.

% Covariate dimensions are $(d_1,d_2) = (35,30)$. The network size $N= (N_1,N_2)$ is varied and indicated in each case.

Key parameters regarding network generation in~\eqref{eq:model:2} are $p$ and $q$. We reparametrize our planted partition model in terms of expected average degree
\begin{align}\label{eq:avg:ex:degree:formula}
	\lambda = \frac{2 N_1 N_2}{N_1+ N_2}\big[q + (p-q)\sum_k \pi_{1k} \pi_{2k} \big]
\end{align}
(see Appendix~\ref{sec:expec:avg:degree}) and the out-in-ratio $\alpha = q/p \in[0,1)$. Estimation becomes harder when $\lambda$ decreases (few edges) or when $\alpha$ increases (communities are not well separated). We fix $\alpha = 1/7$ and vary $\lambda$ in the subsequent simulations.

When generating from the degree-corrected version, we draw $(\theta_i, i\in C_{rk})$ from a Pareto (i.e., power-law) distribution, for each $k \in [K]$ and $r=0,1$. Real networks are frequently reported to have power-law degree distributions~\cite{Barabasi1999}. The Pareto$(a,R)$ in general has density $\theta \mapsto (a R^{a}) \theta^{-a-1} 1\{\theta > a\}$, with mean $ a R/(a-1)$ for $a > 1$ and variance $R^2 a /[(a-1)^2 (a-2)]$ for $a > 2$. Since $|C_{rk}|^{-1} \sum_{i \in C_{rk}} \theta_i$ will be approximately equal to the mean of the Pareto, and we want this average to be $1$, we have to choose $R = (a-1)/a$, that is, we generate $\theta_i \iid \text{Pareto}(a,(a-1)/a)$ for $i \in C_{rk}$. (To comply with our model specification, we further normalize $\theta_i$ for their within-community averages to be exactly one; this will have little effect since the average is already close to $1$.) The variance in this case is $[a(a-2)]^{-1}$ which is decreasing in $a$ over $(2,\infty)$. In order to get maximum degree heterogeneity (i.e., the worse case in terms of the difficulty of fitting), we take $a = 2$, corresponding to infinite variance. We note that expression~\eqref{eq:avg:ex:degree:formula} remains valid for the degree-corrected case without modification, assuming normalization~\eqref{eq:theta:normalization}; see Appendix~\ref{sec:expec:avg:degree}. 

% We first need to set the values for $K, N$ and $d$ and initialize $\mu, \sigma^2, \Sigma , p, q$ and $\pi$. For simplicity we initialize $\Sigma$ as $\nu I_{d_1 + d_2}$ for some $\nu$. 
% %The bigger $\nu$ is, the more separated $v$ and as a result 
%  Then we generate the true labels and the adjacency matrix $A$ from \text{Ber($p$)} for the matched clusters and \text{Ber($q$)} for the unmatched ones. After generating $v_{k*}$ for $k\in[K]$ from $N(\mu,\Sigma)$, we generate $X_r$ from $N(v_{*r},\sigma_r^2 I_{d_r})$ for $r = 1,2$. 

   %and when block sizes vary. 
  
%\begin{figure}
%	\centering
%    \raisebox{-0.5\height}{\includegraphics[scale=.5]{figs2/adj_lambda1p4.eps}}
%    \raisebox{-0.5\height}{\includegraphics[scale=.45]{figs2/uneq_avd1p4_pert90_3.eps}}
%	\caption{Adjacency matrix (left), true labels, perturbed initial labels with $\omega = 0.9$ and estimated labels $\tau$ (right). The expected average degree of the network is $\lambda = 1.4$.}
%	\label{fig:block}
%\end{figure}

\newcommand{\tauh}{\widehat{\tau}}
\begin{figure}[ht]
	\centering
	\hskip2ex\includegraphics[width=.6\linewidth]{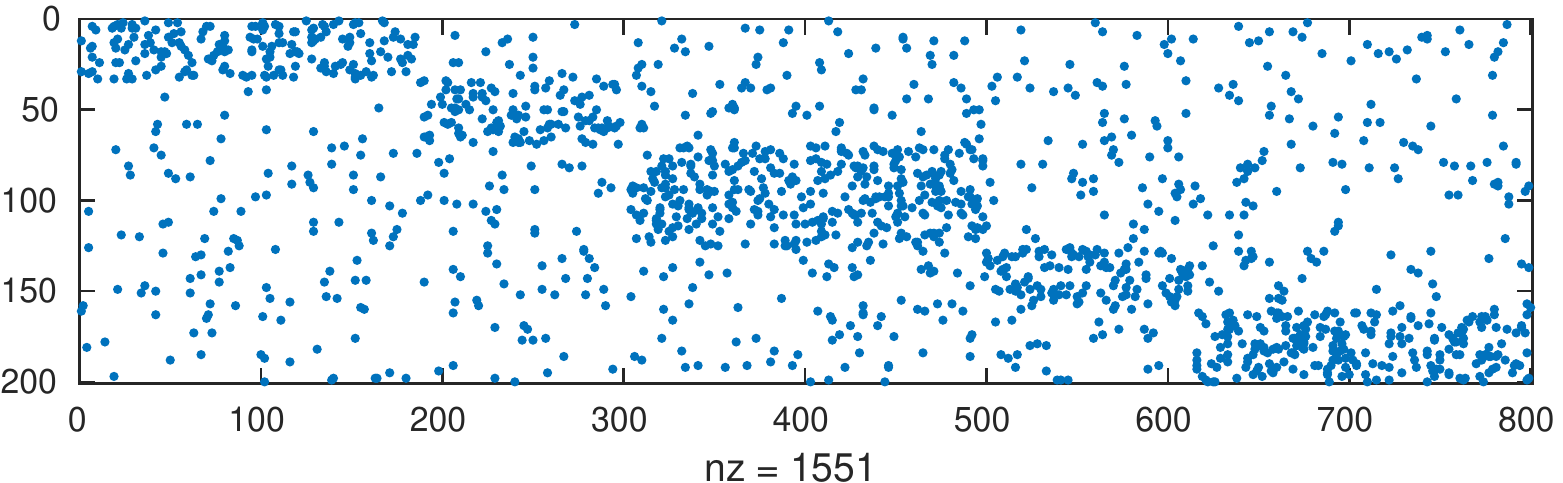}
	\vskip1ex
	\begin{tabular}{cccc}
		\includegraphics[width=0.103\linewidth]{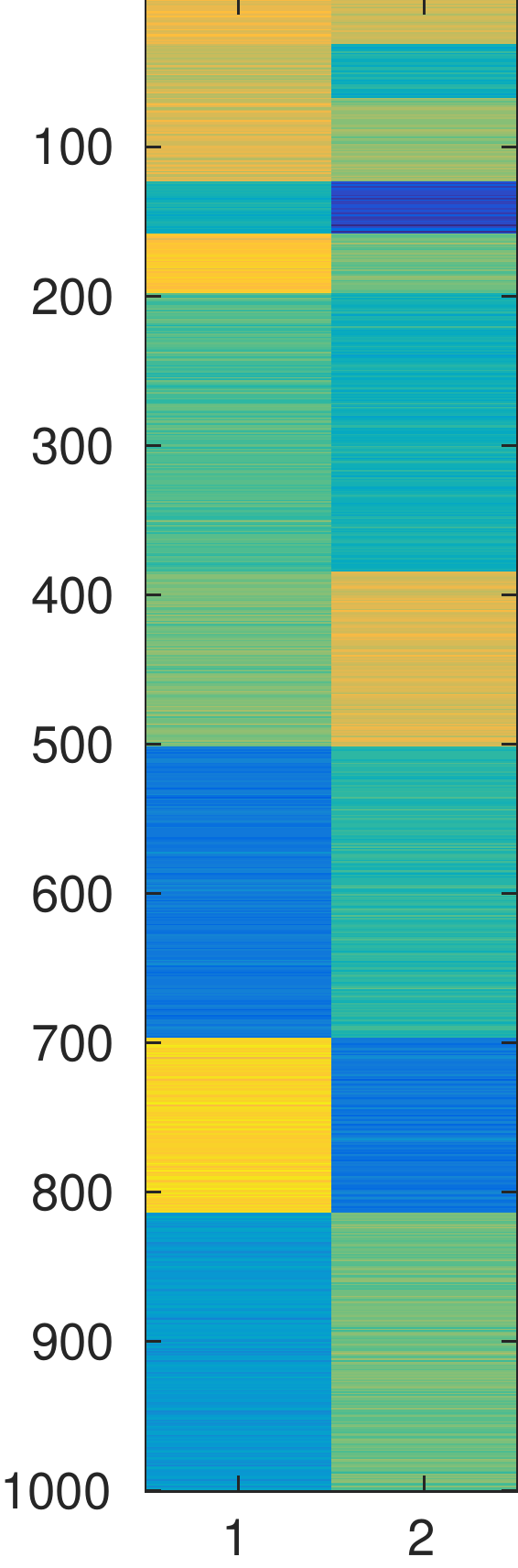}\hskip2ex &
		\includegraphics[width=0.2\linewidth]{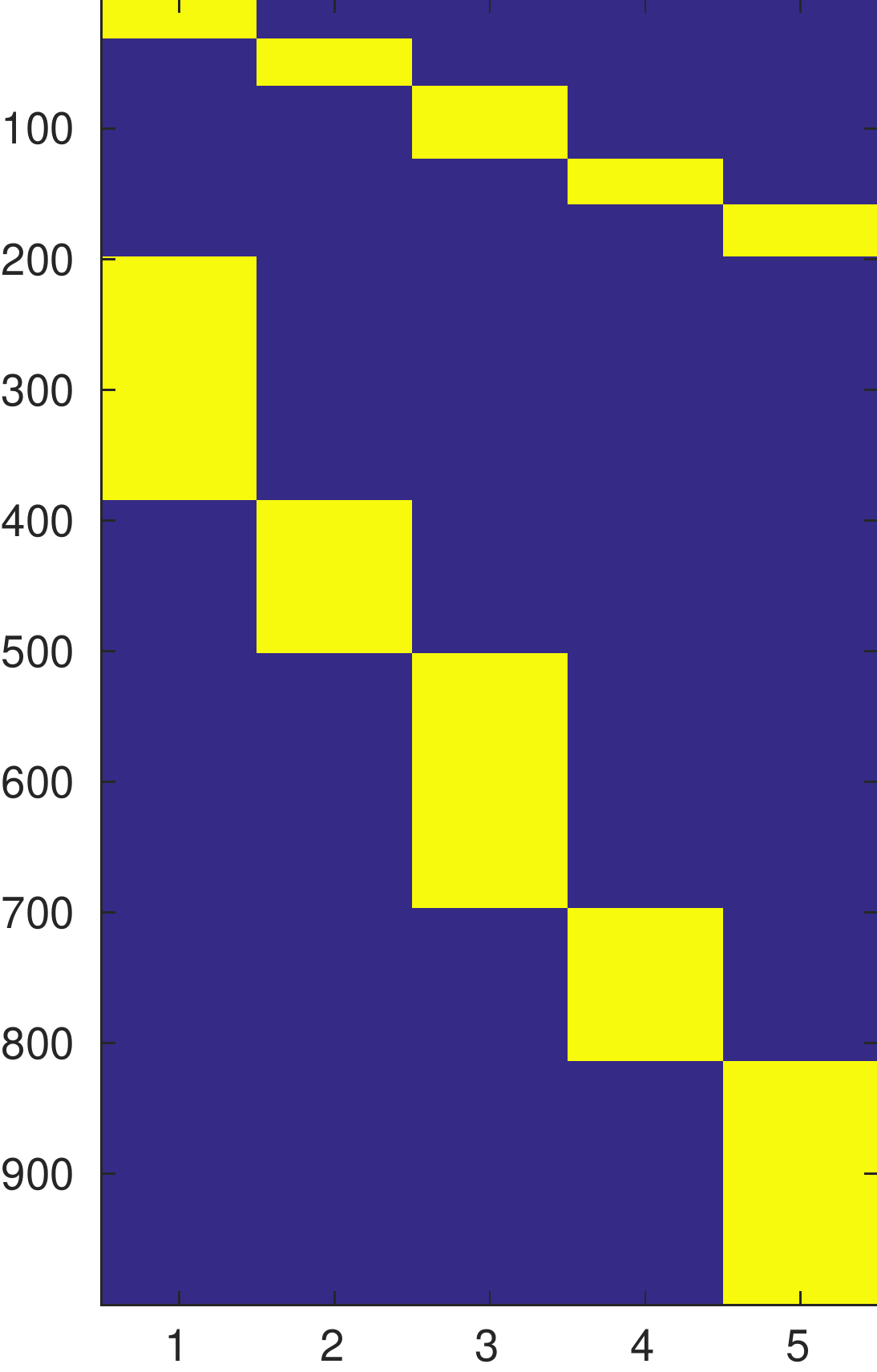}\hskip2ex &
		\includegraphics[width=0.2\linewidth]{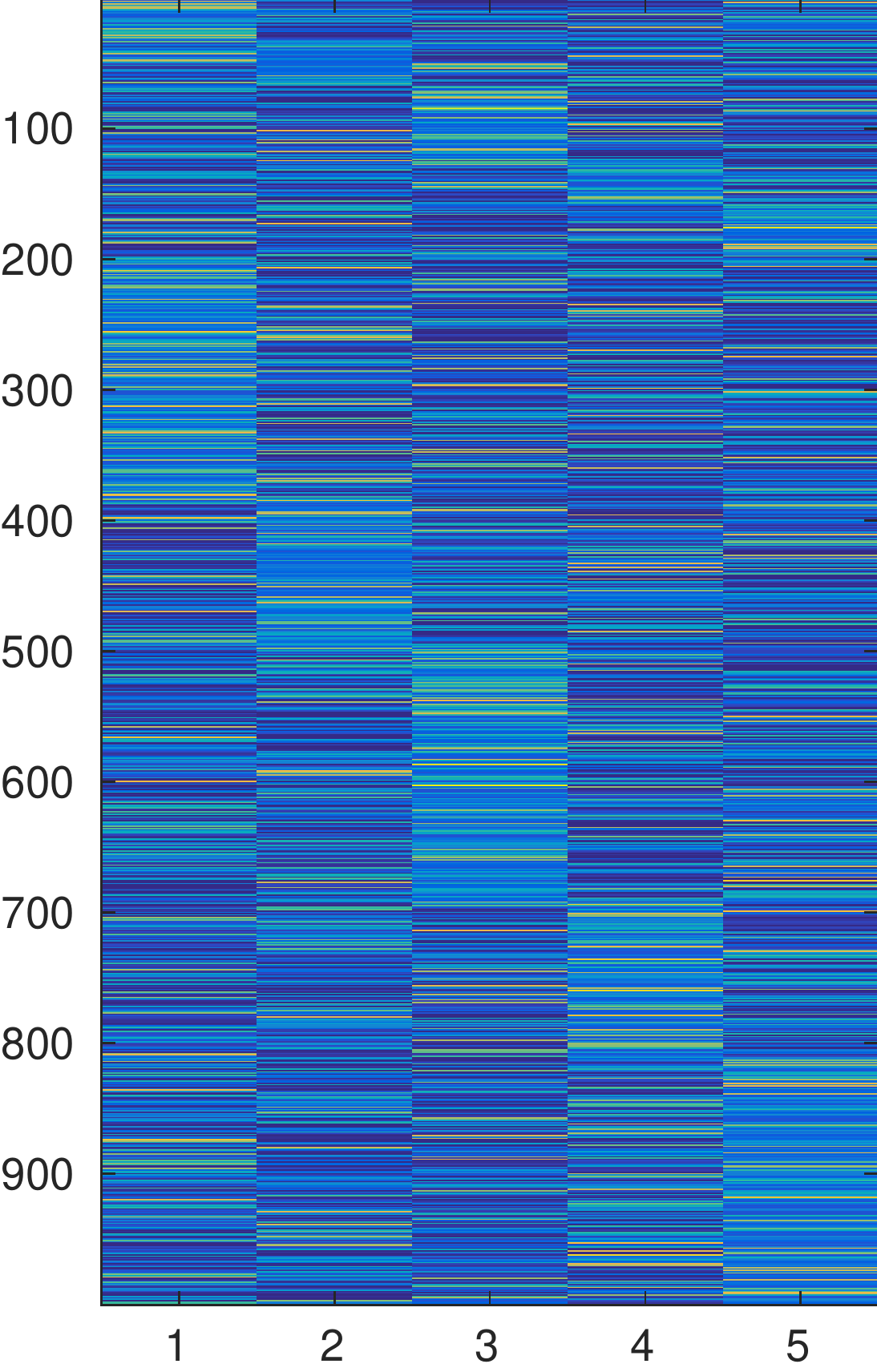}\hskip2ex &
		\includegraphics[width=0.2\linewidth]{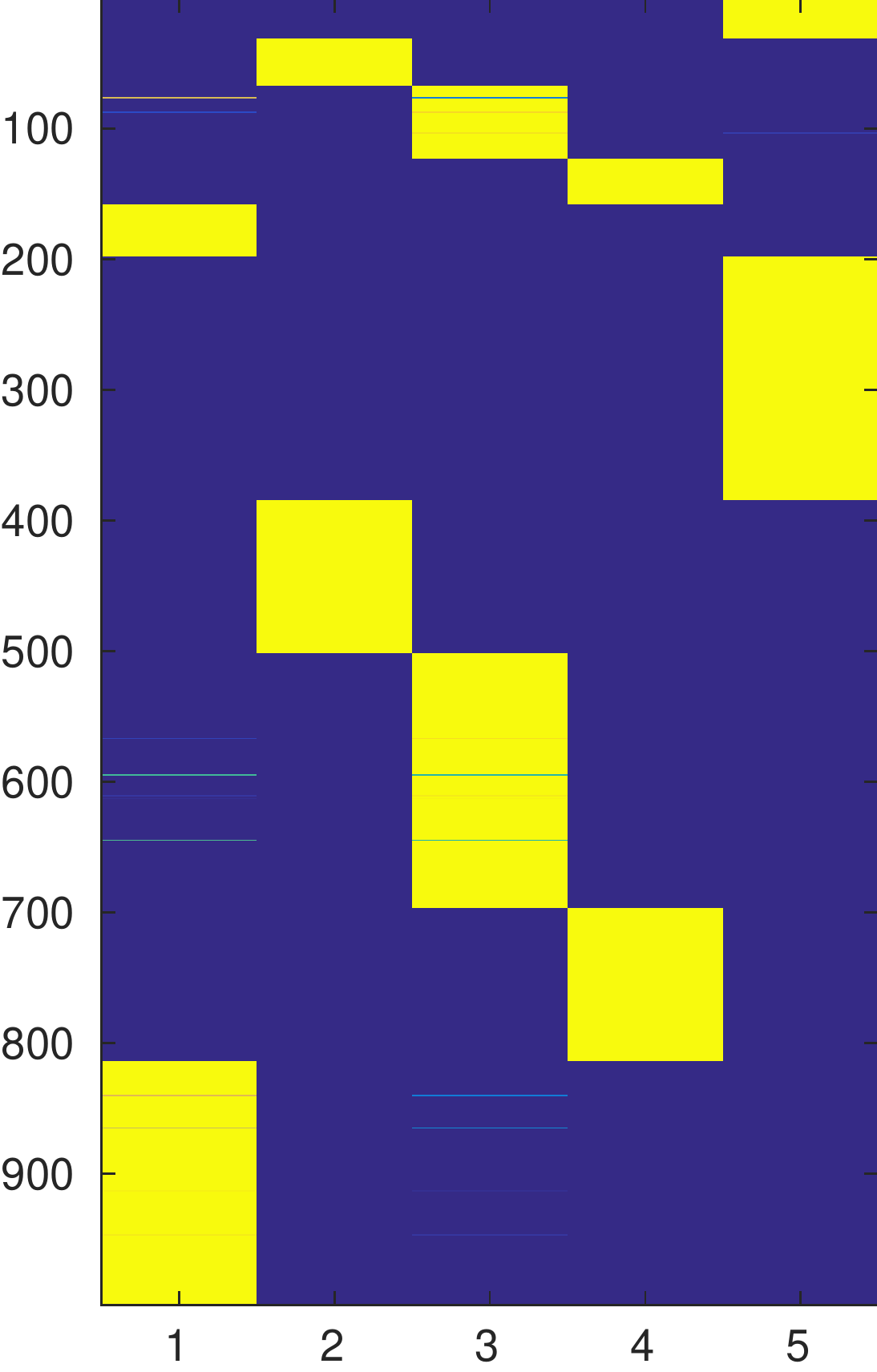}
		\\
		(a) & (b) & (c) & (d)
	\end{tabular} 
	\caption{Typical output of the algorithm. Top row: bi-adjacency matrix. Bottom row: (a) Concatenated covariate matrix $[X_1;X_2]$. (b) Concatenated true labels $[Z_1;Z_2]$. (c) Initial labels for the algorithm, Dirichlet-perturbed truth $0.1[Z_1;Z_2] + 0.9\,\text{Dir}(0.5\onev_K)$. (d) Concatenated output of the algorithm $[\tauh_1;\tauh_2]$. }% generated by demo3.m
	\label{fig:typical}
\end{figure}

\subparagraph{Matched NMI for evaluation.} In general, we measure the accuracy of the algorithms by the normalized mutual information (NMI) between the inferred and correct communities which is defined as the mutual information of the (empirical) joint distribution of the two label assignments divided by the joint entropy~\cite{Malvestuto1986}. NMI has a maximum value of 1 for perfect agreement and a minimum of 0 for no agreement. One could measure NMI individually between $Z_r \in \{0,1\}^{N_r \times K}$ (the true label matrix) and $\tau_r \in [0,1]^{N_r \times K}$ (the estimated soft-label matrix) for each $r=0,1$. However, one can also measure a \emph{matched NMI} by concatenating the labels of two sides vertically, i.e., forming $[Z_1;Z_2]$ and $[\tau_1;\tau_2]$ and measuring a single NMI between the resulting $(N_1 + N_2) \times K$ matrices. Some thought should convince the reader that this the natural way to also measure the effectiveness of the matching between the communities of the two sides: We have a matched NMI of $1$, if the true and estimated clusters on each side are in perfect agreement, and the matching between them is perfectly recovered.

\subparagraph{Typical output.} Figure~\ref{fig:typical} shows the typical output of the algorithm on the data generated from the model without degree correction (DC), i.e., $a = \infty$. Here the average degree is $\hat \lambda = 3.1$, $K= 5$, $\nu = 10$ and $d = (2,2)$, the dimensions of the covariates. Concatenated matrices of the true labels and the initial and final labels are shown. Vertical concatenation is used as discussed earlier, giving matrices of dimension $(N_1 + N_2) \times K$. Initial labels are the Dirichlet-perturbed truth with $\rho = 0.1$, i.e. $90\%$ noise, as discussed in Section~\ref{subsubsec:initialization}.
It is interesting to note that the output of the algorithm has recovered the communities with a nontrivial permutation of the community labels. 

In other words, the perturbation of the initial labels is high enough that the convergence of the algorithm cannot  simply be explained by a local perturbation analysis: the algorithm has not converged to the original labels, but to a perfectly valid permuted version of the original labels. That is, $\tauh_r \approx Z_r Q_r$ for $r=0,1$ where $Q_1$ and $Q_2$ are $K\times K$ permutation matrices. The matched NMI and misclassification rate for the algorithm are $0.98$ and $0.30\%$ in this case. If one runs $k$-means on the concatenated matrix of covariates $[X_1;X_2]$, disregarding the network information, one gets matched NMI and misclassification rate, $0.44$  and $38.50\%$. That is, the covariates themselves are not as informative alone as in combination with the network.

%\subsection{Evaluation of the algorithm}
%Fig.\ref{fig:block} shows a typical result of the algorithm. The adjacency matrix, the perturbed labels and the estimated $\tau$are shown for the model with $N_1 = 500, N_2 = 400$, $\alpha = .1494$, $\omega = .9$, $\lambda = 1.4$ . According to the figure, our algorithm can recover the labels (perfectly) even for a very sparse network (small $\lambda$).

\begin{figure}[t]
	\centering
	\includegraphics[width=.45\linewidth]{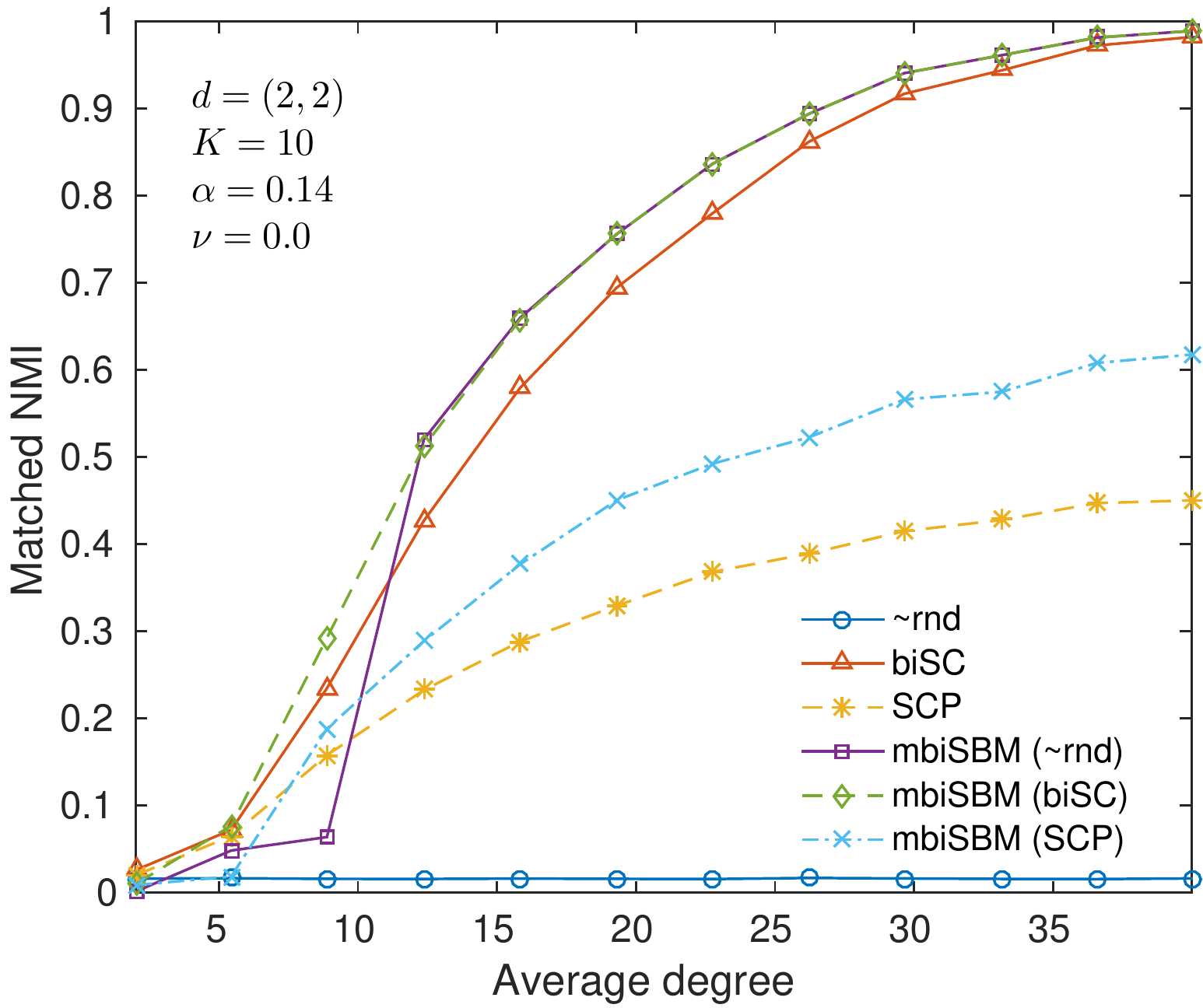}\hspace{1ex}
	\includegraphics[width=.45\linewidth]{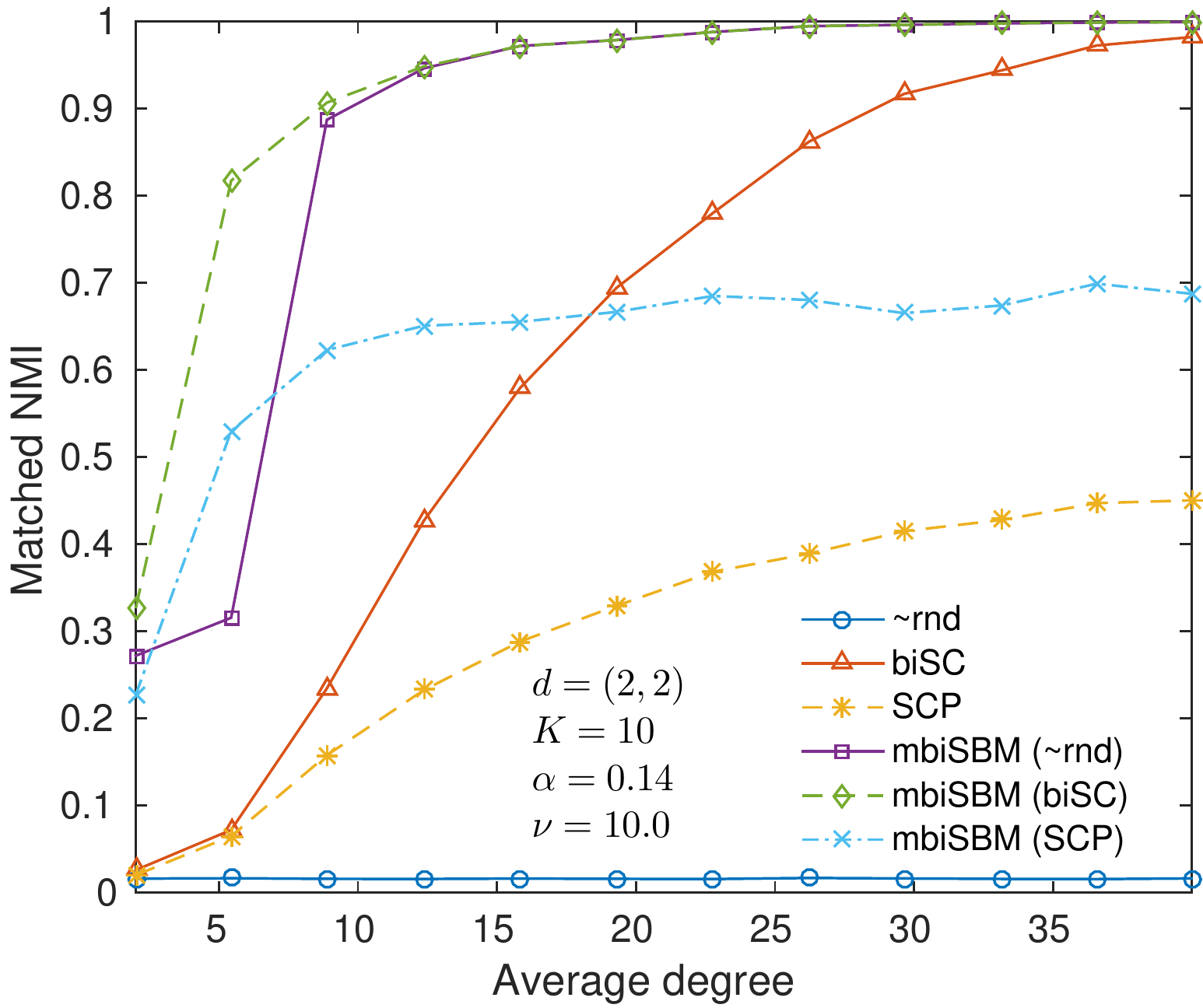}

	\caption{Effect of different initialization methods on mbiSBM with (left) $\nu = 0$ and (right) $\nu = 10$. The case $\nu = 0$ corresponds to no covariate information whereas $\nu > 0$ gives some covariates information.}
	\label{fig:nmi:undc}
\end{figure}

\subparagraph{Average behavior.}
Figure~\ref{fig:nmi:undc} shows the mathched NMI versus average expected degree $\lambda$ for various methods. The results are  averaged over 50 Monte Carlo replications. Naive (regularized) spectral clustering, denoted as \scp, is shown in addition to \bisc as discussed in Section~\ref{subsubsec:initialization}. Moreover, the plots show our algorithm \texttt{mbiSBM}, initialized with both spectral methods and with Dirichlet-perturbed truth ($\rho = 0.1$) denoted as \texttt{$\sim$rnd}. The two plots correspond to the case with no covariate information, $\nu = 0$, and the case with covariate information $\nu = 10$. In both cases, covariate dimensions are $d = (2,2)$, number of communities $K=10$ and out-in-ratio is $\alpha = 1/7$. There is no degree-correction in the model or \texttt{mbiSBM} algorithm. 

 As can be seen, \bisc outperforms SCP significantly. Without covariates, \texttt{mbiSBM} started with \bisc slightly improves upon \bisc; initializing with $\approx 10\%$ truth (\texttt{mbiSBM} \texttt{($\sim$rnd)}) has similar performance for sufficiently large $\lambda$, showing that \texttt{mbiSBM}  behaves well with any sufficiently good initialization. Note also that \texttt{mbiSBM} initialized with SCP, improves upon SCP for large $\lambda$. With covariate information ($\nu = 10$), \texttt{mbiSBM} significantly outperforms \bisc which does not incorporate the covariates. 

%Fig.\ref{fig:meandeg} and Fig.\ref{fig:ratio} show the agreement of estimated labels with the truth in various settings, as measured by the normalized mutual information (NMI) averaged over 50 Monte Carlo replications. Note that when $\nu = 0$, our algorithm will be downgraded to a matched bipartite block model with no node covariates. \aaa{What are these plots showing? Need to discuss.}

\begin{figure}[ht]
	\centering
	\begin{tabular}{cc}
	\includegraphics[width=.45\linewidth]{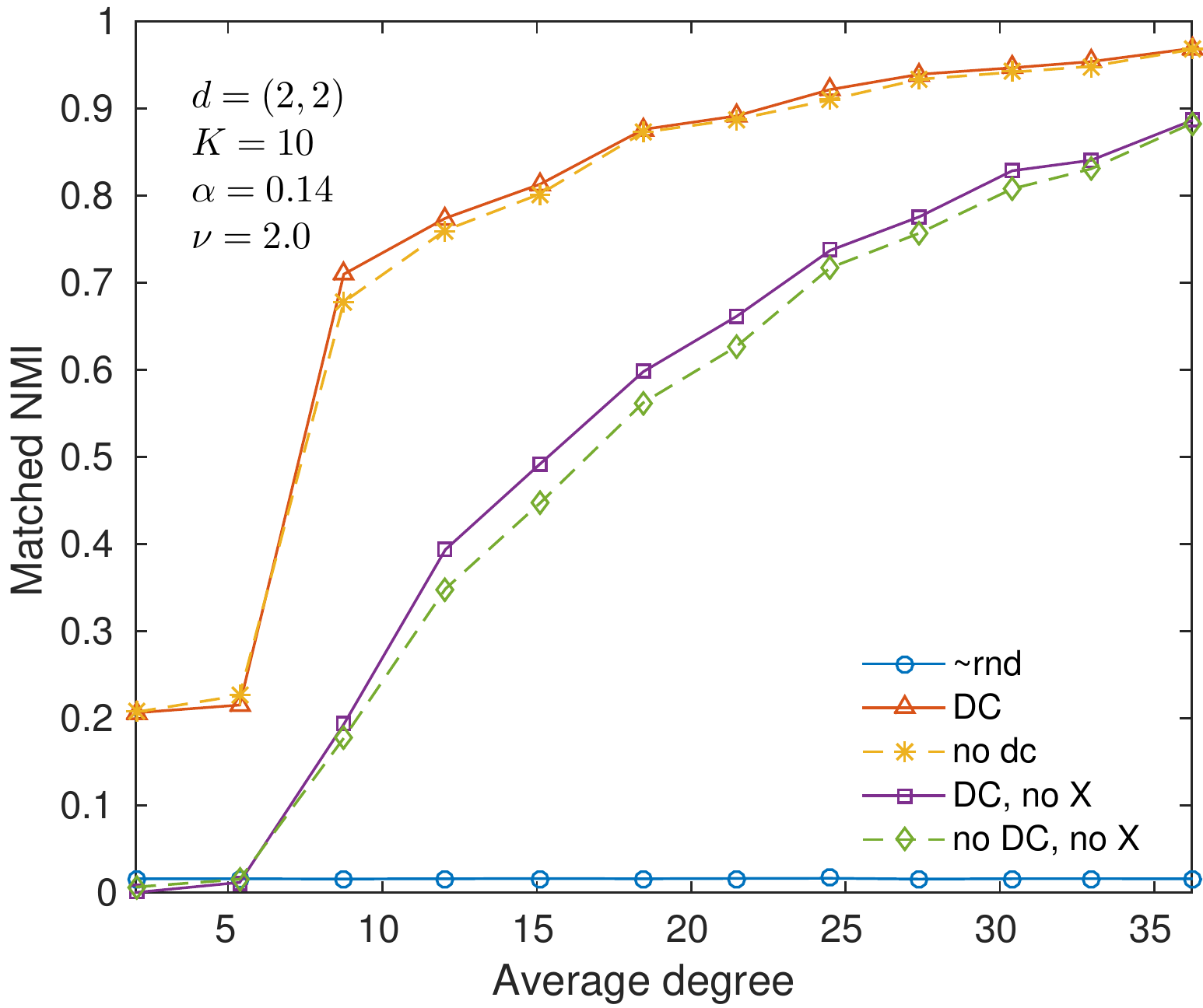}\hspace{1ex} &
	\includegraphics[width=.45\linewidth]{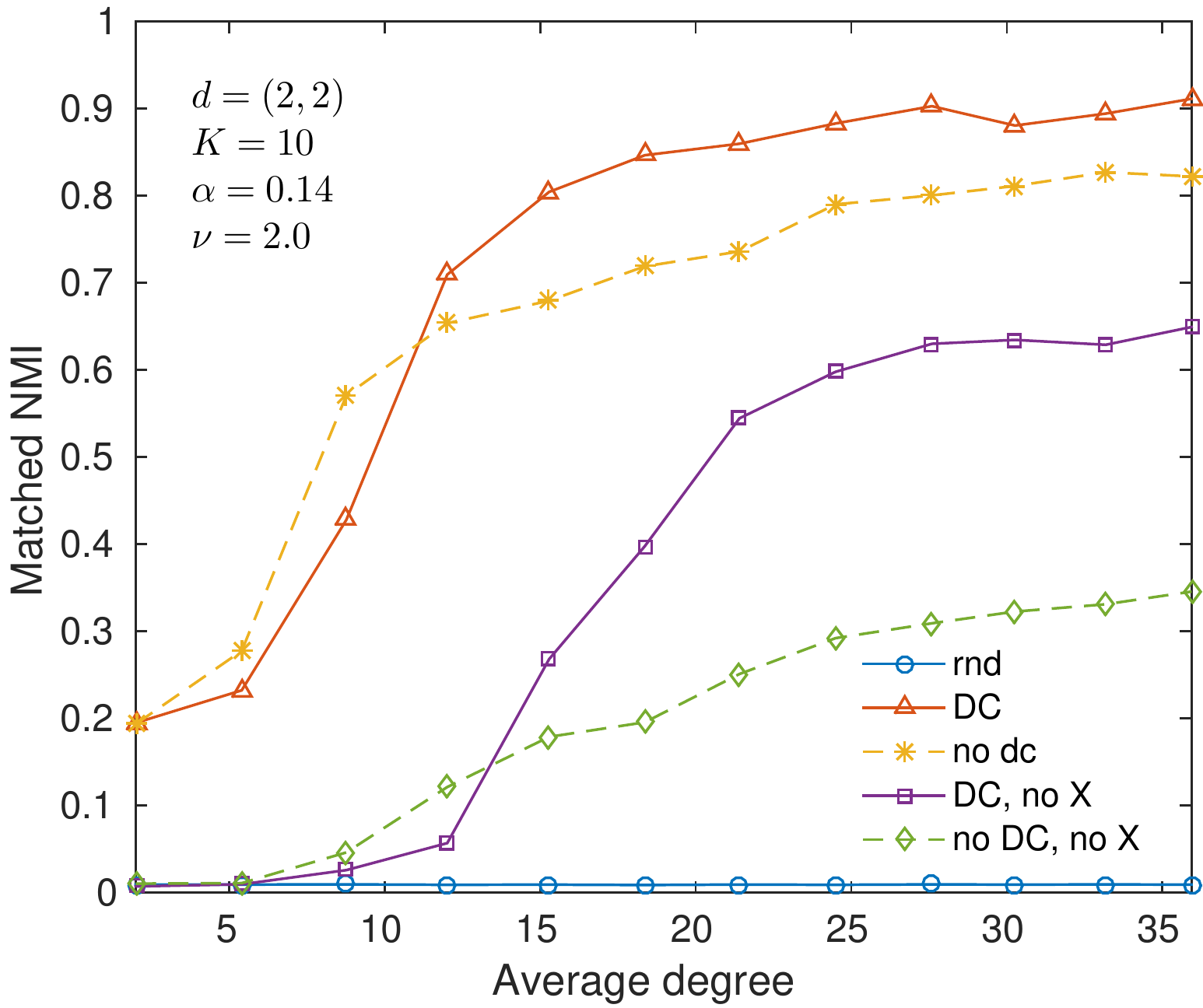}  \\
	(a) & (b)
	\end{tabular}

	\caption{Effect of the degree correction steps 6--8 in the algorithm, for a power-law network. Data is generated from DC version of the model (subsection~\ref{sec:dc:model}) with Pareto distribution $p(\theta)  \propto \theta^{-3},\;\theta > 2$ for degree parameters within each community. (Left) shows the results for a good initialization, the Dirichlet perturbed truth, denoted as~\simrnd (with $\approx$ 10\% true labels) and the (right)  shows the results for a completely random initialization, denoted as \rnd.}
	\label{fig:dc:effect}
\end{figure}

\subparagraph{Effect of degree correction.}
Figure~\ref{fig:dc:effect} investigates the effect of employing degree-correction in the algorithm. In both plots of the figure, we are generating from the same DC-version of the model using within-community Pareto degree distribution with parameter $a = 2$ as described earlier, $d=(2,2)$, $K= 10$, $\alpha = 1/7$, and $\nu = 2$. The difference between the two plots is how we initialize \texttt{mbiSBM} algorithm. The left panel corresponds to ``Dirichlet-perturbed truth'' initialization ($\rho = 0.1$) denoted as \texttt{$\sim$rnd}, whereas the right panel corresponds to completely random initialization, denoted as \texttt{rnd}. Four versions of the algorithm are considered, with or without covariate ($X$) incorporation, and with or without degree-correction (DC).

Surprisingly, as the left panel shows, with sufficiently good initialization (\texttt{$\sim$rnd}), degree-correction step of the algorithm provides only a slight improvement. However, the improvement of degree-correction is quite significant when starting from a poor initialization (\texttt{rnd}). In general, it is advisable to use the DC version since its solution has less variance. Figure~\ref{fig:dc:effect:boxplots}, illustrates the algorithm with DC correction and without, in the same setup of the left panel of Figure~\ref{fig:dc:effect}, that is, both cases initialized with \texttt{$\sim$rnd} (and both incorporating covariates). Though Figure~\ref{fig:dc:effect}(a) shows that mean behaviors are close, Figure~\ref{fig:dc:effect:boxplots} shows that the distributions of the outputs are quite different, with the solution of DC version having less variability. This is expected as the DC version is solving an optimization problem with much restricted feasible region.

\begin{figure}[t]
	\centering
	\includegraphics[width=.5\linewidth]{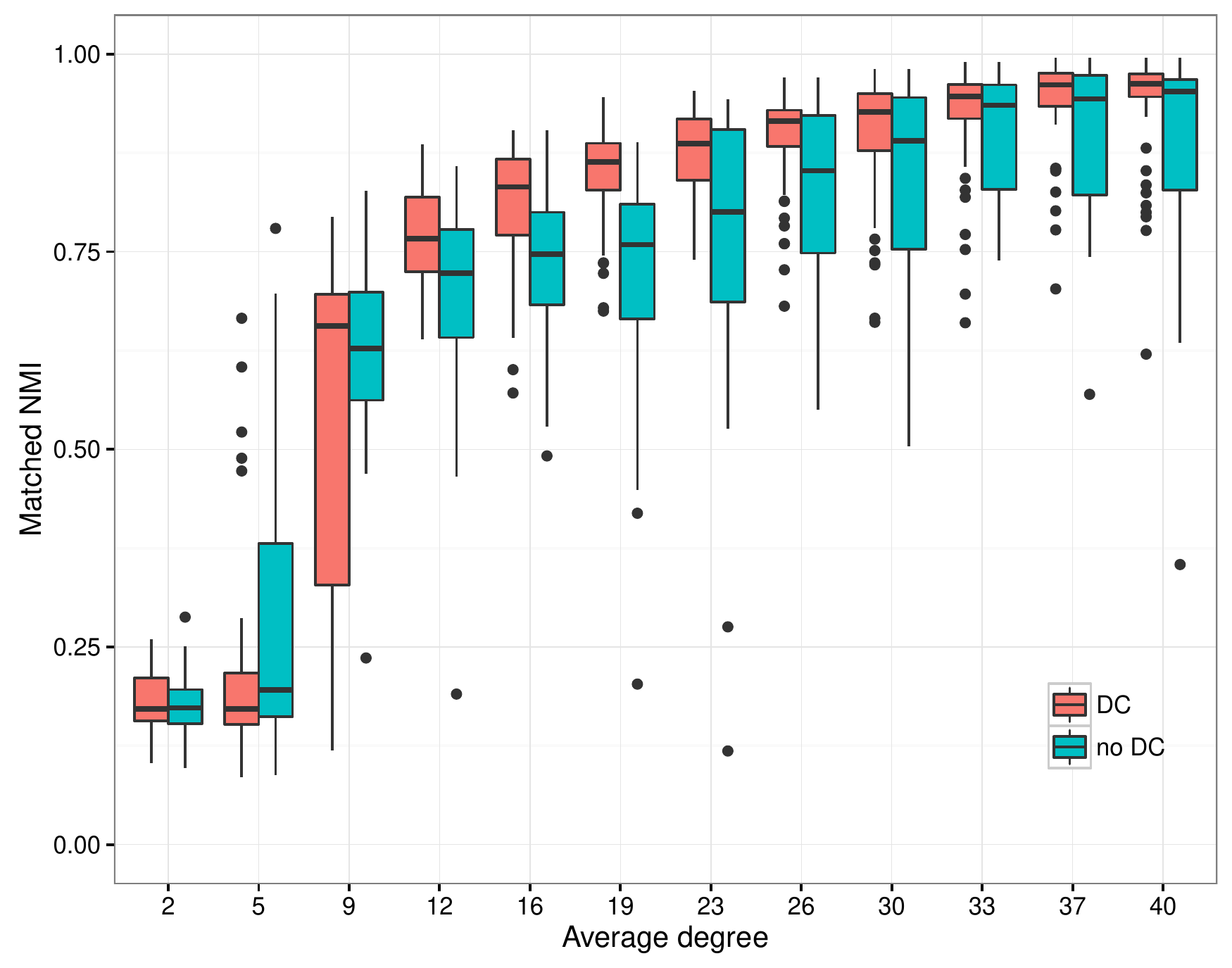}
	\caption{Effect of degree correction correction steps on variability, for a power-law network.}
	\label{fig:dc:effect:boxplots}
\end{figure}

%\subparagraph{Model misspecification (outliers).}

%Simulations
%\begin{itemize}
%	\item Plot of mNMI as a function of degrees (covar and no covar)
%	\item Effect of increasing the dimension of covariates (Plot mibSBM (~tru))
%	\item Effect of varying $\nu$ 
%	
%	\item Individual versus matched NMI
%	\item Effect of varying beta?
%	\item One sided 
%\end{itemize}

%\begin{figure}
%	\centering
%	\includegraphics[width=0.43\linewidth]{figs2/NMI_deg_covar_sepNMI_B}
%	\includegraphics[width=0.43\linewidth]{figs2/NMI_deg_covar_sepNMI_Ca}
%	\includegraphics[width=0.43\linewidth]{figs2/NMI_deg_covar_sepNMI_Cb}
%	\includegraphics[width=0.43\linewidth]{figs2/NMI_deg_covar_sepNMI_F}
%	\includegraphics[width=0.43\linewidth]{figs2/NMI_deg_covar_sepNMI_Fa}
%	
%\end{figure}

%\section{Simulations under model misspecification}
\section{Application to real data}\label{sec:realdata}
We have applied the algorithm to two wikipedia page--user networks, which we will call \topart and \cities. Each is a bipartite network between a collection of Wikipedia pages and the users who edited them: An edge is placed between a user and a page if the user has edited that page (at least once).  In the \topart, the pages are selected from the top articles (based on monthly contributions) from Chinese (CN), Korean (KR) and Japanese (JP) language Wikipedia, corresponding to the period from January to October 2016. In the \cities network, the pages correspond to city names in English language Wikipedia; the cities were chosen from five countries: Unites States (US), United Kingdom (GB), Australia (AU), India (IN), Japan (JP). In both cases, on the user side, only those with IP addresses were retained. Although, not perfect, IP addresses were the only means by which we could obtain additional information about each user, esp. geo-location data. Wikipedia usage statistics were scraped from~\cite{wikimedia} using code inspired by~\cite{Keegan2014}. For geo-data we used both the \texttt{ggmap} R package~\cite{Kahle2013} and the API provided by~\cite{ipapi}.

 In both networks, the true labels are the language assigned to each page and each user, that is, matched communities are specified by common language. The user language was assigned based on the dominant language of the country from which the IP address originates. The IPs were also used to obtain latitude and longitude coordinates on each user, providing us with user covariate matrix $X_2 \in \reals^{N_2 \times 2}$,
 
 The page language was assigned differently for the two networks. For \topart, it is the language in which the page was written. For \cities, it is the language of the country to which the city (that the page is about) belongs. For \topart, we do not have any page covariate. For \cities, we use the geo-location data of the city (latitude and longitude) to give us the page covariate matrix $X_1 \in \reals^{N_1 \times 2}$.

 % Below, we describe details of how the two networks are constructed and the primary language of each page and user determined. 

Figure~\ref{fig:wiki:nets} shows the two networks along with the true communities. Note that \cities is specially hard to cluster based only on network data due to the presence of nodes of different communities among each community (as positioned by the layout algorithm).  Tables~\ref{tab:top:art} and ~\ref{tab:cities} show the break-down of pages/users based on community (i.e., language) for the two networks. Also shown are the average degrees of each side of the network, as well as the overall average degree. For each of the two networks, we first obtained a 2-core, restricted to the giant component, then removed users from countries not under consideration. (If the last step created disjoint components we restricted again to the giant component. This only happened for \cities and only removed 5 nodes.)

% The primary language of each page is clear from the text in this case. The primary language of the user was assigned based on the dominant language of the country from which the IP address originates. (If the IP originated from outside Chinese, Japanese or Korean speaking countries, we removed the user from the network.) For \topart, we filtered the original extracted network to removed nodes with degree $< 2$ and then restricted ourselves to the resulting giant component.

\begin{table}[t]
	\begin{center}
		\begin{tabular}{l*{3}{c}|c|c|c}
			& CN & JP & KR &  Total & Avg. deg. & Covariates\\
			\hline
			Pages & 139 & 143 & 171 & 453 & 14.2 & N/A\\
			Users & 579 & 695 & 828  & 2102 & 3.1& $X_2 = $ user (lat.,lon.)\\
			\hline
			Total &  718 & 838 & 999 & 2555 & 5\\
		\end{tabular}
	\end{center}
	\caption{\topart page--user network}
	\label{tab:top:art}
\end{table}
\begin{table}[t]
	\begin{center}
		\begin{tabular}{l*{5}{c}|c|c|c}
			& US & IN & AU & JP & GB & Total & Avg. deg. & Covariates\\
			\hline
      Pages & 267 & 235 & 182 & 113 &  59 & 856& 10.2 &$X_1 = $ city (lat.,lon.)\\
			Users & 1054 & 1029 & 705 & 101 & 201  & 3090& 2.8& $X_2 = $ user (lat.,lon.)\\
			\hline
			Total & 1321 & 1264 & 887 & 214 & 260 & 3946& 4.4\\
		\end{tabular}
	\end{center}
	\caption{\cities page--user network}
	\label{tab:cities}
\end{table}

\begin{figure}[t]
	\begin{center}
		\includegraphics[scale=0.5]{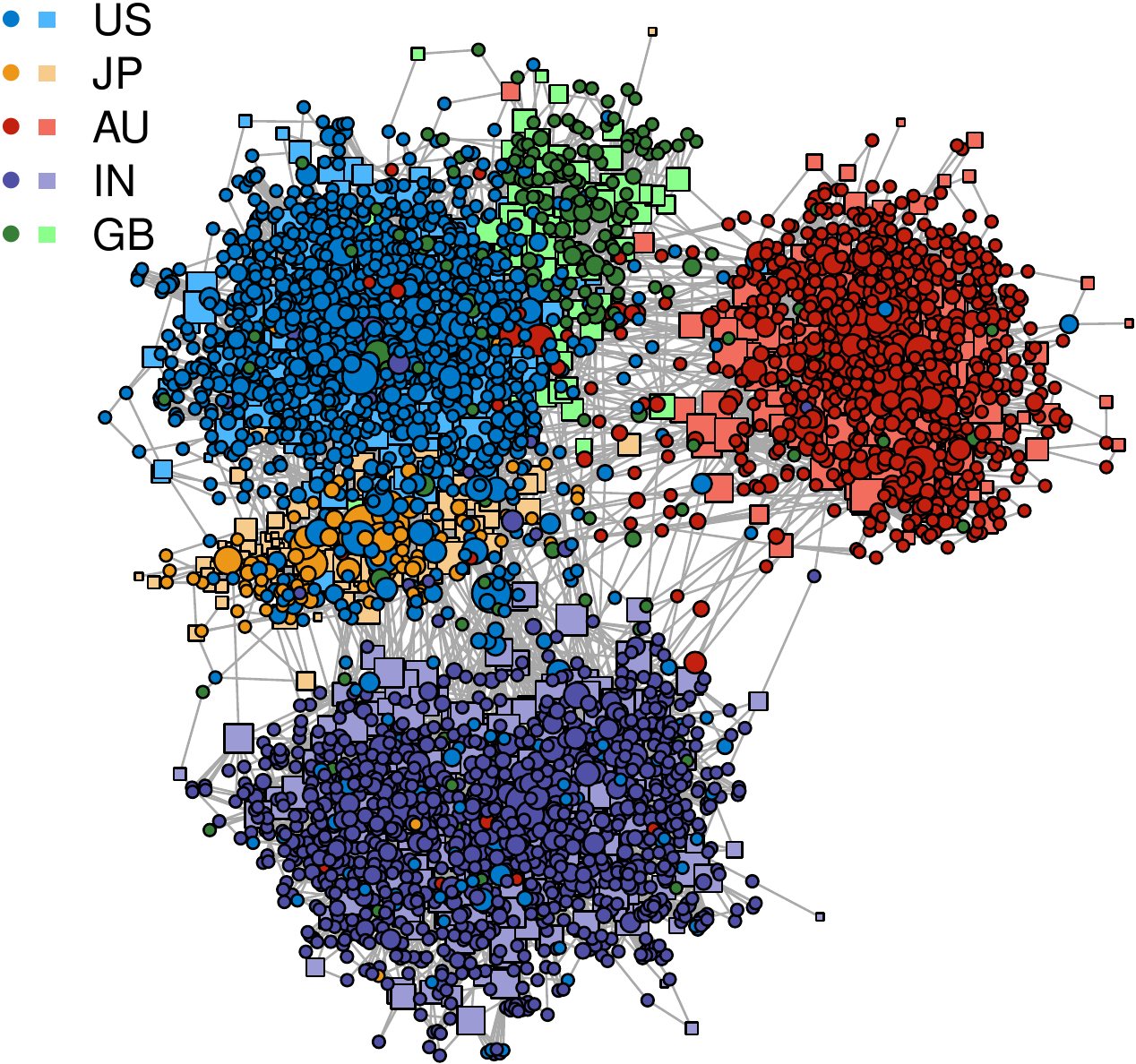}
		\includegraphics[scale=0.5]{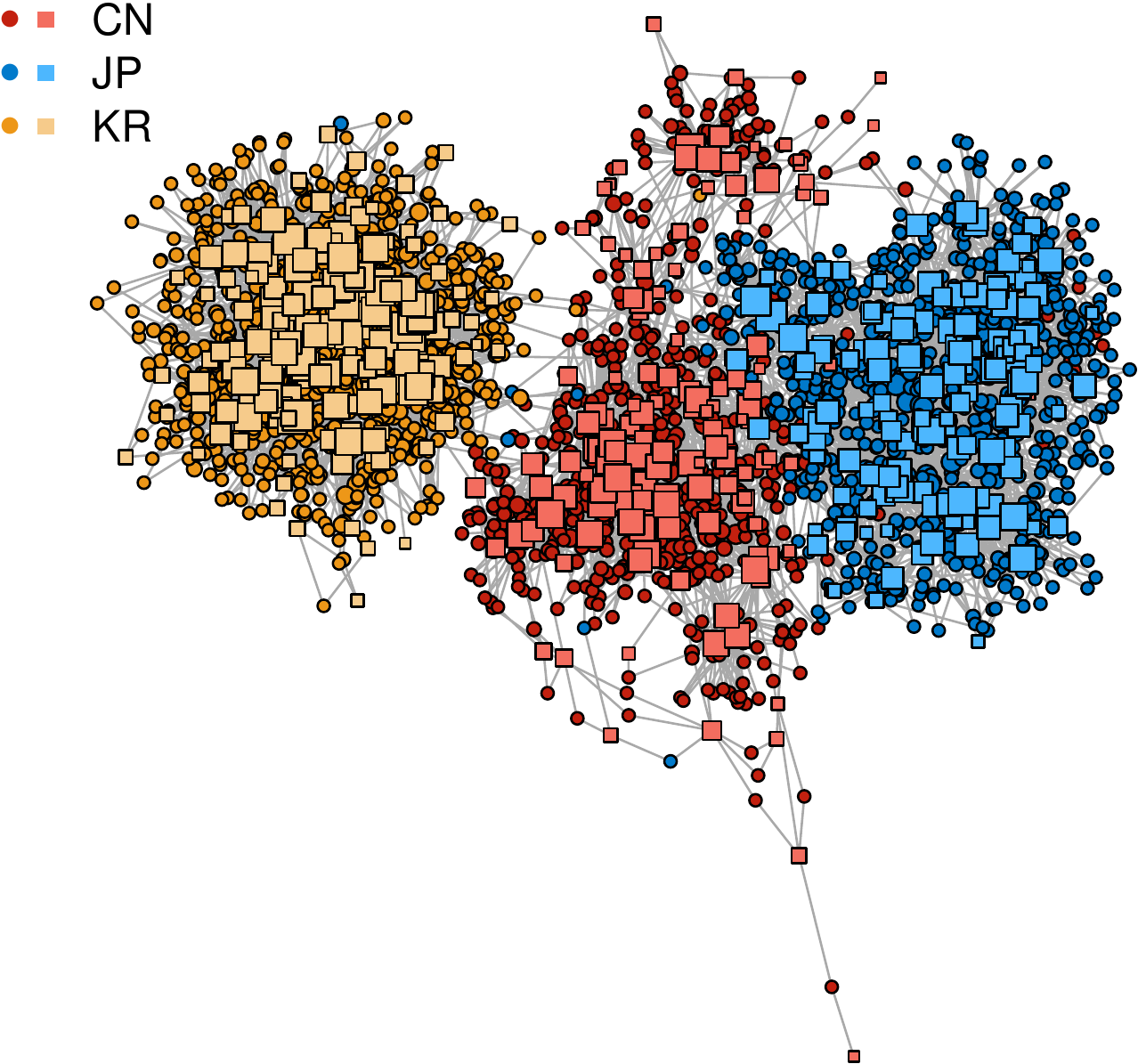}
	\end{center}
	\caption{The two Wikipedia networks: (left) \cities (right) \topart. Nodes are colored according to true communities (assigned languages). Pages are denoted with squares and users with circles. Node sizes are proportional to log-degrees.}
	\label{fig:wiki:nets}
\end{figure}

\subparagraph{Results on Wikipedia networks.}
Table~\ref{tab:wiki:res} illustrates the result of the application of \bisc, and the \mbisbm~(\bisc) algorithm with various combination of covariates. In all cases the degree-corrected (DC) version of \mbisbm is used. For \cities, without using covariates, there is no improvement on \bisc while using $X_2$ gives significant boost to \mbisbm. For \topart, \bisc outperforms \mbisbm with covariates. One the other hand, adding $X_2$ or both $X_1$ and $X_2$ significantly improves the result of \mbisbm.
\begin{table}
	\centering
	\renewcommand{\arraystretch}{1.2}
	\begin{tabular}{cc|*{4}{c}}
		Network & \bisc & \multicolumn{4}{c}{\mbisbm (\bisc), DC}\\
		\hline  
		&  & $X_1$ \& $X_2$ & $X_1$ & $X_2$ & no $X$ \\
		\hline
		\cities & 0.86 & - & - & 0.98 & 0.86 \\
		\topart & 0.6 & 1.0 & 0.59 & 0.85 & 0.47
	\end{tabular}
	\caption{Matched NMI for \bisc and \mbisbm on the two Wikipedia networks.}
	\label{tab:wiki:res}
\end{table}

\begin{figure}[t]\centering
	\begin{tabular}{cc}
		\includegraphics[scale=0.43]{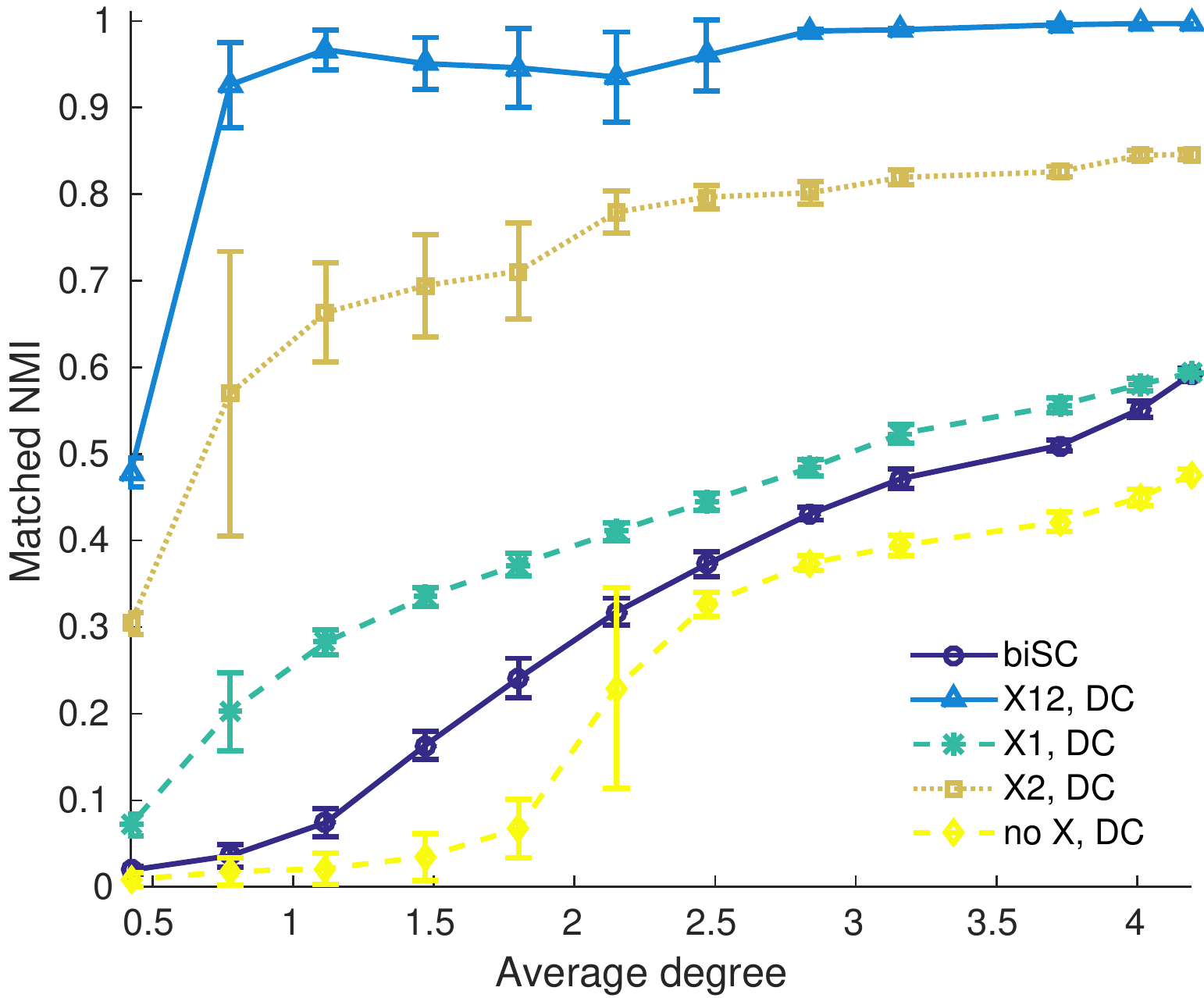} &
		\includegraphics[scale=0.43]{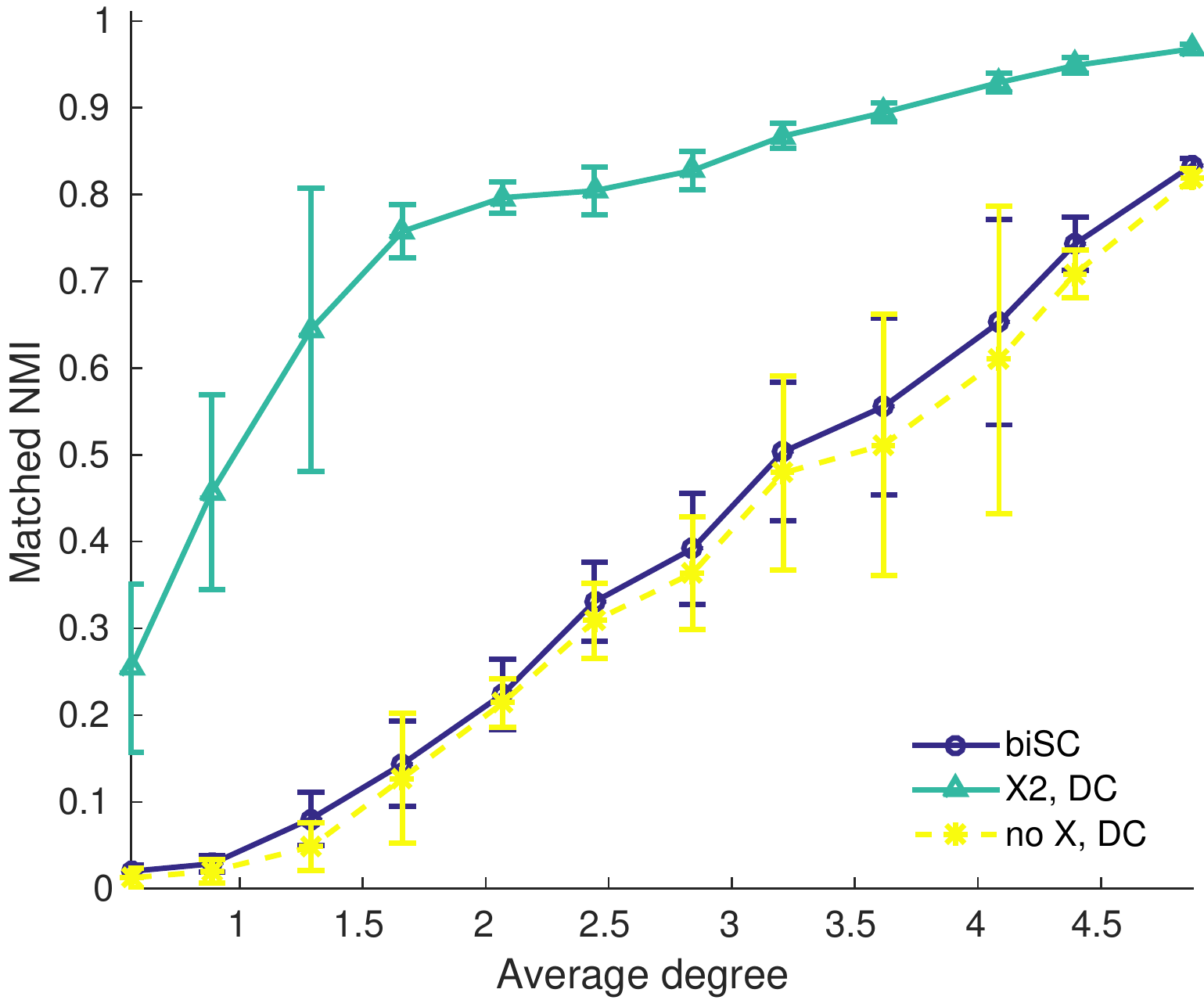}
	\end{tabular}
	\caption{Effect of subsampling on Wikipedia networks: (left) \cities (right) \topart. }
  \label{fig:wiki:subsamp}
\end{figure}

To get a more refined understanding of the relative standing of \bisc and \mbisbm~(\bisc), we have run two Monte Carlo analyses based on these real networks, one using \emph{subsampling} and the other by adding \ER  noise. Figure~\ref{fig:wiki:subsamp} shows the results when we subsample the network to retain a fraction of the nodes on each side (from 95\% down to 10\%). The x-axis shows the resulting overall average degree of the network at each subsampling level. The results are averaged over 50 replications and the interquartile range (IQR) is also shown as a measure of variability. The figures in Table~\ref{tab:wiki:res} correspond to the rightmost point of these plots. (Average degrees of \cities vary in these ranges: overall $\in [0.4,4.2]$, page $\in [1,9.7]$ and user $\in [0.3,2.7]$, whereas for \topart the ranges are:  overall $\in [0.6,4.9]$, page $\in [1.6,13.8]$ and user $\in [0.3,3]$.)

The plots show that \mbisbm~(\bisc) with covariates outperforms \bisc, and the improvement is quite significant the sparser the network becomes.  Note for example that in \cities adding $X_1$ does not have much effect in the original network, however, there is a considerable improvement when average degree starts to drop under subsampling. In the \cities case, the two covariates $X_1$ and $X_2$ together are quite strong leading to an NMI $\approx 1$ and masking the effect of the network to some extent. 

However, by looking at cases where only one of $X_1$ and $X_2$ is present, we observe that \mbisbm~(\bisc) manges to pass the covariate information via the network to the side without covariates, thus improving matched NMI significantly. To see this, consider for example the \topart, where only $X_2$ is present. In this case, even if a method could cluster $X_2$ perfectly and, in the absence of network information randomly guessed the labels of the other side, the NMI would be 0.42. That in Figure~\ref{fig:wiki:subsamp}(b), the NMI starts at 0.98 and remains much above 0.42 for most of the range of subsampling illustrates the ability of \mbisbm to effectively utilize both covariate and network information to correctly infer the labels of the other side. The same can be observed in the case of \cities. Finally, we note that without covariates, \bisc usually performs better. We expect this since it is hard for local methods starting from \bisc to improve upon it. The strength comes when we use the covariate information.

\begin{figure}[t]
	\centering
	\begin{tabular}{cc}
		\includegraphics[scale=0.43]{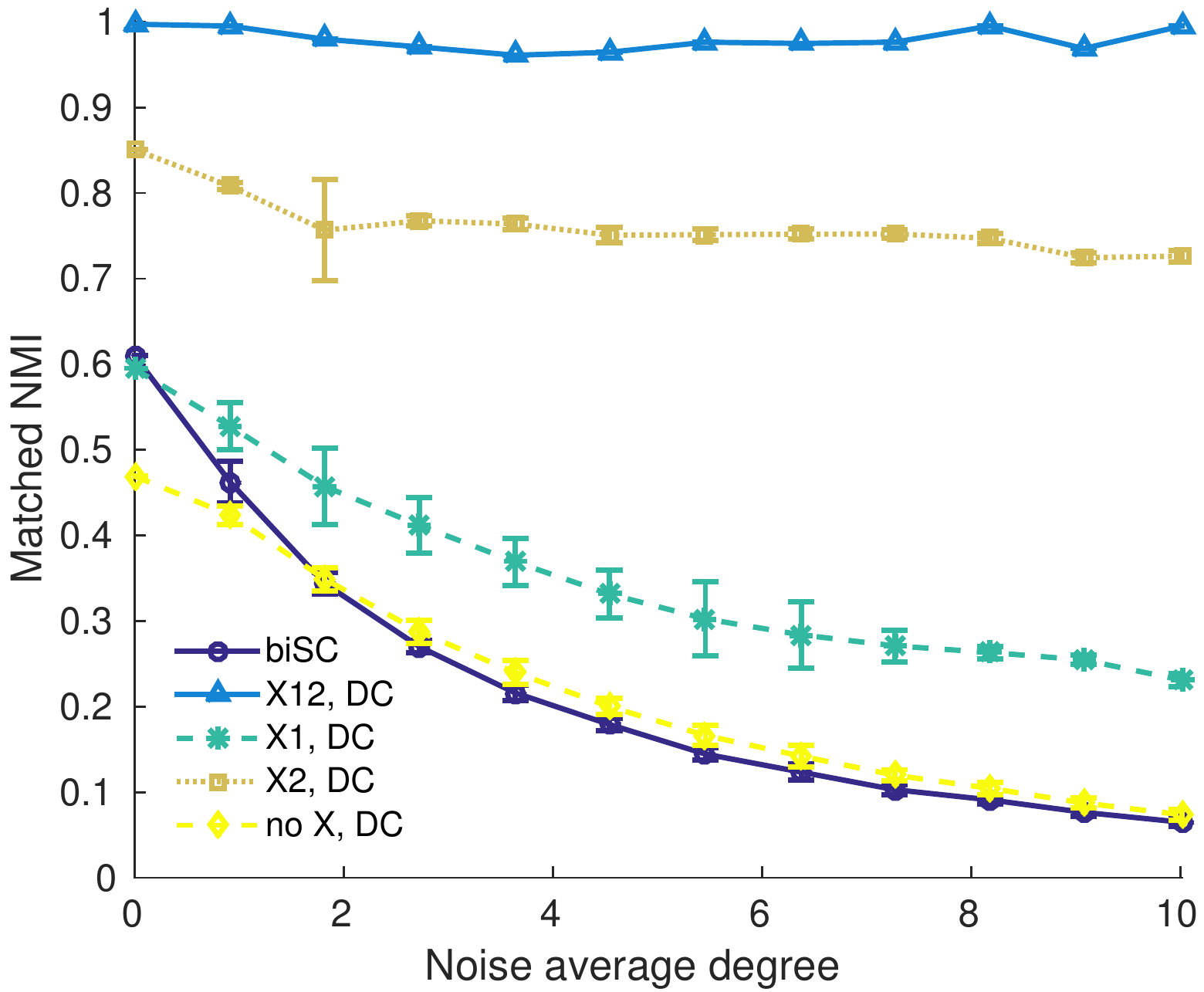} &
		\includegraphics[scale=0.43]{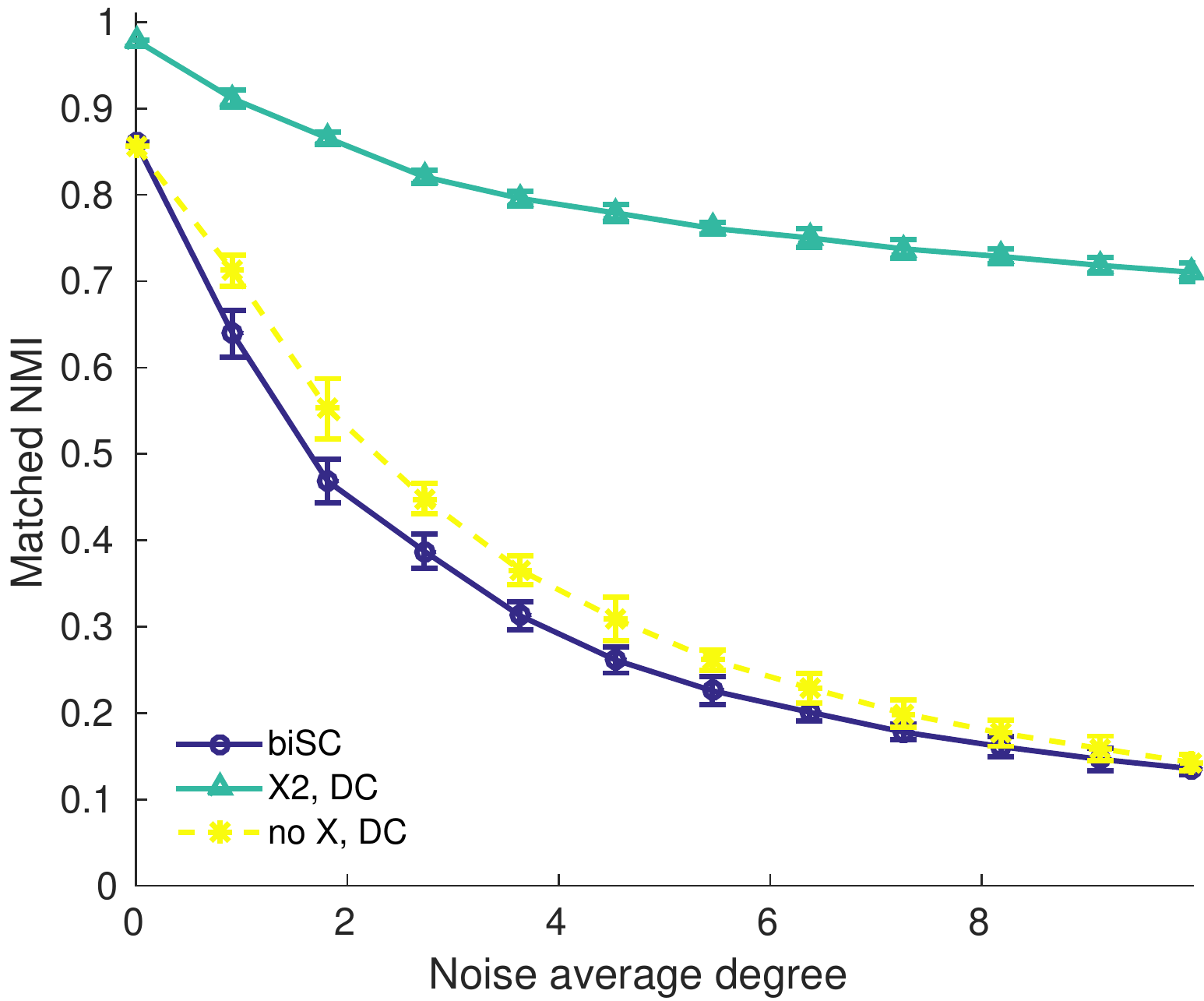}
	\end{tabular}
	\caption{Effect of adding \ER noise on Wikipedia networks: (left) \cities (right) \topart.}
	\label{fig:wiki:addnoise}
\end{figure}

Figure~\ref{fig:wiki:addnoise} shows another experiment where we added \ER noise of average degree from 0 to 10. Again, the advantage gained by \mbisbm from using covariates can be quite clearly observed when one increases the noise. The covariates mitigate the effect of noise and lead to a much graceful degradation of performance for \mbisbm relative to \bisc.

\section{Discussion}\label{sec:discussion}

% \section{Future Directions}
% \begin{itemize}
% \item The model we proposed was motivated by a real application in which we had gene expression of two species and an ortholog network between the two. We can use similar idea and modify the model so that we can detect communities in multilayer network in which the whole data is in network format.
% \item We can generalize our model to find communities for more than two groups.
% %\item As the bipartite graph becomes denser, the error becomes smaller. This relationship needs to be further investigated.
% \item Finding the optimal number of clusters(communities) has been always a challenge for clustering problems. Our model is no exception.
% \item This model needs to be tested with our real data.
% \item Degree-corrected block models
% \item Do we experience phase transition in the existence of node covariates?
% \end{itemize}
In this paper, we considered the problem of matched community detection in the bipartite setting, where one assumes a latent one-to-one correspondence between communities of the two sides. This matching is built into the model and inferred simultaneously with the communities in the process of fitting the model.
Our model is an extension of the stochastic block model (SBM) and its degree-corrected version (DC-SBM). We extended our proposed model to allow for the presence of node covariates that are aware of the matching between communities: Covariates corresponding to nodes in matched communities are statistically linked using hierarchical Bayesian modeling ideas. 

Although we only considered Gaussian distributions in generating covariates, our hierarchical mixture approach has the potential for extension to more general settings. For example, one can easily model discrete covariates as mixtures of multinomial distributions. Some care however is needed when deciding the distribution of the top layer to allow for proper information sharing among the lower level variables ($v_{rk}$ in our notation). In addition, as mentioned (cf.~Section~\ref{sec:covariate:corr}), our current statistical linkage carries weak information about the matching and it would be interesting to design models in which the degree of covariate information about the matching can be tuned more effectively.

%We use hierarchical Bayesian modeling ideas to jointly model the covariate of the two sides. To fit the model, we derive an algorithm based on variational inference. We show through simulation and real data that covariate information when correlated with community structure, can drastically boost performance of community detection.

Our model has natural extensions to $r$-partite ($r>2$) networks where some of the modes may or may not have node covariates. We note that the general $r$-partite case is related to the so-called multilayer or multiplex community detection problem~\cite{Kivela2014}. Finally, one would like to allow for edge covariates to accommodate many cases in real data, where edges are annotated in some way, say by the ratings as in recommender systems, by time-stamps as in our Wikipedia user-page examples, and so on. Poisson model of edge generation can, to some extent, take simple edge weights into account. Whether one can go beyond that in modeling more complex edge information is an interesting avenue for future work.

\printbibliography

\appendix

%!TEX root = sbm_bip_arxiv.tex
\section{Details of Section~\ref{sec:model:fitting}}
\subsection{Derivation of~\eqref{eq:elbo:expr:1}}\label{sec:elbo:deriv}

  We write $\E_q$ for expectation w.r.t. the above joint distribution on $(Z,V)$. Similarly, we write $\E_{q_Z}$ and $\E_{q_V}$ for the expectation (or integration) w.r.t. to each of $q_Z$ and $q_V$. Note that $\E_q[\,\cdot\,]  = \E_{q_V} \E_{q_Z}[\,\cdot\,]$. Plugging in the variational distribution~\eqref{eq:var:dist} into the variational likelihood~\eqref{eq:elbo:def}. we have
\begin{align*}
J &= \E_q\Big[\ell(\mu,\Sigma,\sigma,Q)] -\log q(Z,V)\Big] = T_1 + T_2 + T_3 - T_4 - T_5
\end{align*}
where
\begin{align*}
T_1 &= \E_{q_Z} \Big[\sum_{i,j} \psi(A_{ij},y_{ij}) \Big], \quad 
T_2 = \E_{q_V} \E_{q_Z} \Big[\sum_{r,i,k} z_{rik} \big[ \log f_r(x_{ri}; v_{rk}) + \log \pi_{rk}\big]  \Big] \\
T_3 &= \E_{q_V} \Big[ \sum_{k} \log p(v_{*k} | \mu, \Sigma) \Big], 
  \quad T_4 =  \E_{q_Z} \log q(Z), \quad T_5 = \E_{q_V} \log q(V)
 \end{align*}
 Let 
$ \gamma_{ij}(\tau) := \E_{q_Z}(y_{ij}) = \sum_{k=1}^K \tau_{1ik} \tau_{2jk}$,
so that
\begin{align}\label{eq:T1:expr:1}
T_1 = \sum_{i,j} \Big[\gamma_{ij}(\tau ) \log(p^{A_{ij}} (1-p)^{1-A_{ij}}) + (1-\gamma_{ij}(\tau))  \log(q^{A_{ij}} (1-q)^{1-A_{ij}}) \Big]
\end{align}
We frequently use the following elementary result in the sequel. Let $x \mapsto N(x;\mu,\Sigma)$ be the PDF of the multivariate normal distribution with mean $\mu$ and covariance $\Sigma$.
\begin{lem}\label{lem:ex:quad:form}
  Let $\eps$ be a random vector with mean $\mut$ and covariance $\Sigt$, and $x$ a nonrandom vector. Then,
  \begin{align}
    \E[\eps^T  \Lambda \eps] &= \tr[\Lambda \Sigt] + \mut^T \Lambda \mut, \label{eq:quad:form:1}\\
    \E[ \log N(x;\eps,\Sigma)] =  \E[ \log N(\eps; x,\Sigma)] &=  -\frac12 \Big\{ \log |\Sigma| +  (x-\mut)^T \Sigma^{-1} (x-\mut) + \tr(\Sigma^{-1} \Sigt) \Big\}  \notag \\
    &=-\frac12 \Big\{ \log |\Sigma| +   \tr(\Sigma^{-1} \Psi ) \Big\}  \label{eq:quad:form:2},
  \end{align}
  where $\Psi = \Sigt + (x-\mut)(x-\mut)^T$.
\end{lem}
\begin{proof}
  Let us prove~\eqref{eq:quad:form:2}. We have $\log N(x;\eps,\Sigma) = -\frac12 \log |\Sigma| - \frac12 (x-\eps)^T \Sigma^{-1} (x-\eps)$. Noting that $x-\eps$ has mean $x -\mut$ and covariance $\Sigt$ and applying~\eqref{eq:quad:form:1} gives the desired result.
\end{proof}

Recall that $f_r(x_{ri}; v_{rk}) = N(x_{ri}; v_{rk},\sigma^2_r I_{d_r})$, and note that under $q_V$, $v_{rk}$ has mean $\mut_{rk}$ and covariance $(\Sigt_k)_{rr}$. Note that we are partitioning $\Sigt_k$ into four blocks of sizes $\{d_1,d_2\} \times \{d_1,d_2\}$ and $(\Sigt_k)_{rr}, r=1,2$ correspond to the two diagonal blocks in this partition. Using Lemma~\ref{lem:ex:quad:form}, we have
\begin{align*}
T_2 & = \E_{q_V} \Big\{\sum_{r,i,k} \tau_{rik} \big[ \log N(x_{ri}; v_{rk},\sigma^2_r I_{d_r}) + \log \pi_{rk}\big]\Big\} \\
&= \sum_{r,i,k} \tau_{rik} \big[ - \frac{d_r}{2}\log\sigma_r^2 
  -\frac{\tr\big((\Sigt_k)_{rr}\big) + \|x_{ri} - \mut_{rk}\|^2} {2\sigma^2_r} + \log \pi_{rk}\big].
\end{align*}
Recall that  $\Gamt := ((\Sigt_k, \mut_k), k=1,\dots,K)$ and
$	\Psib(\Gamt) := \frac{1}{K} \sum_{k=1}^K \big[\Sigt_k + (\mut_k - \mu) (\mut_k - \mu)^T \big].$
Another application of Lemma~\ref{lem:ex:quad:form} gives
\begin{align*}
T_3 = \E_{q_V}\Big[\sum_{k} \log N(v_{*k}; \mu, \Sigma)\Big]
&= -\frac12 \sum_k \big[ \log |\Sigma| + \tr(\Sigma^{-1}\tilde{\Sigma}_k) +(\tilde{\mu}_k-\mu)^T\Sigma^{-1}(\mu_k-\mu)\big] \\
&= -\frac12 \sum_k \big[ \log |\Sigma| + \tr(\Sigma^{-1} \Psi_k(\Gamt)) \big] \\
&= -\frac{K}2 \big[ \log |\Sigma| + \tr(\Sigma^{-1} \Psib(\Gamt)) \big].
%
%-\frac{1}{2}\Big\{ 
%  K\log|\Sigma|+\sum_k\big[\tr(\Sigma^{-1}\tilde{\Sigma}_k) +(\tilde{\mu}_k-\mu)^T\Sigma^{-1}(\tilde{\mu}_k-\mu)\big]\Big\} 
\end{align*}
Using Lemma~\ref{lem:ex:quad:form} once more, we have
\begin{align*}
	T_5 = \E_{q_V} \log q(V) &= \sum_k \E_{q_V} \log N(v_{*k};\mut_k,\Sigt_k) \\
	&= \sum_k -\frac12 \big[ \log |\Sigt_k| +  \tr(I_{d_1+d_2}) \big] = -\frac12 K(d_1 +d_2) -\frac{1}{2}\sum_k\log|\Sigt_k|
\end{align*}
Finally,  we have 
\begin{align}\label{eq:T4:expr1}
	T_4 = \E_{q_Z} \log q(Z) = \E_{q_Z} \sum_{r,i,k} z_{rik} \log \tau_{rik} = \sum_{r,i,k} \tau_{rik} \log \tau_{rik}.
\end{align}
Putting the pieces together we get expression~\eqref{eq:elbo:expr:1}.

\subsection{Updates of $\Sigt$ and $\mut$}\label{sec:details:Gamt}
%Let us define $\tauc_{rk} := \sum_{i=1}^{N_r}\tau_{rik}$ and $D^{-1}_k := \diag\big( \frac{\tauc_{1k}}{\sigma^2_1}I_{d_1}, \frac{\tauc_{2k}}{\sigma^2_2}I_{d_2}\big)$
%%\begin{align}\label{eq:tau:check:def}
%%  \tauc_{rk} = \sum_{i=1}^{N_r}\tau_{rik}, \quad 
%%  \text{and} \quad
%%  D^{-1}_k = \diag\Big( \frac{\tauc_{1k}}{\sigma^2_1}I_{d_1}, \frac{\tauc_{2k}}{\sigma^2_2}I_{d_2}\Big)
%%\end{align}
%Then, as a function of $\Sigt$, $J$ can be written as (see Appendix~?)
%\begin{align}
%J(\Sigt)\doteq -\frac12 \sum_k \tr\big[(D^{-1}_k +\Sigma^{-1})\Sigt_k\big] -\log|\Sigt_k|.
%\end{align}
%%
From~\eqref{eq:elbo:expr:1}, the relevant portion of $J$ which is a function of $\Gamt = (\mut,\Sigt)$ is given by
\begin{align*}
	J(\mut,\Sigt) \doteq \sum_{r,i,k} \tau_{rik}  \beta_{rik}(\Gamt, \sigma^2)
	-\frac{K}2 \tr[\Sigma^{-1} \Psib(\Gamt,\mu)]
	+ \frac12\sum_k\log|\Sigt_k|
\end{align*}
Substituting $\beta_{rik}(\Gamt, \sigma^2) := 
- \frac1{2\sigma^2_r}[\tr\big((\Sigt_k)_{rr}\big) + \|x_{ri} - \mut_{rk}\|^2 ], $ and $	\Psib(\Gamt,\mu) := \frac{1}{K} \sum_{k=1}^K \big[\Sigt_k + (\mut_k - \mu) (\mut_k - \mu)^T \big]$ from their definitions, and looking at the result only as a function of $\Sigt$, we obtain
\begin{align}
	J(\Sigt) \doteq 
		- \frac12 \sum_{k} \Big[ 
	 \sum_{r,i}\tau_{rik} \frac{\tr\big((\Sigt_k)_{rr}\big)}{\sigma^2_r}
	+ \tr[\Sigma^{-1} \Sigt_k]
	- \log|\Sigt_k| \Big]
\end{align}
Recalling $\tauc_{rk} := \sum_{i=1}^{N_r}\tau_{rik}$ and $D^{-1}_k := \diag\big( \frac{\tauc_{1k}}{\sigma^2_1}I_{d_1}, \frac{\tauc_{2k}}{\sigma^2_2}I_{d_2}\big)$, we have
\begin{align*}
	\sum_{r,i}\tau_{rik} \frac{\tr\big((\Sigt_k)_{rr}\big)}{\sigma^2_r} = 
	\sum_{r}\tauc_{rk} \frac{\tr\big((\Sigt_k)_{rr}\big)}{\sigma^2_r} =
	 \tr \Big[	\sum_{r}\frac{\tauc_{rk}}{\sigma^2_r} (\Sigt_k)_{rr} \Big] = \tr(D^{-1}_k \Sigt_k).
\end{align*}
Hence, we obtain $J(\Sigt) \doteq -\frac12 \sum_k \tr[(\Sigma^{-1} + D_k^{-1}) \Sigt_k] - \log|\Sigt_k| $ which is the desired result.

\medskip Similarly, by substituting $\beta_{rik}(\Gamt, \sigma^2)$ and $\Psib(\Gamt,\mu)$ and looking at the result as a function only of $\mut$, we obtain
\begin{align*}
J(\tilde{\mu}) = -\frac12\sum_{k} \Big[\sum_{r,i} \tau_{rik} 
\frac{\|x_{ri} - \mut_{rk}\|^2} {\sigma^2_r} + 
 (\mut_k-\mu)^T \Sigma^{-1} (\mut_k-\mu)\Big]  
\end{align*}
Let us simplify the sum over $r$ and $i$. Up to constants as function of $\mut$, we have
\begin{align*}
\sum_{i} \tau_{rik} 
\|x_{ri} - \mut_{rk}\|^2 &\doteq \sum_{i} \tau_{rik} \big(
\|\mut_{rk}\|^2 - 2\ip{x_{ri}, \mut_{rk}} \big) \\&= \tauc_{rk} \|\mut_{rk}\|^2 - 2\ip{\xc_{rk}, \mut_{rk}} \\
&= \tauc_{rk} \big( \|\mut_{rk}\|^2 - 2\ip{\muc_{rk}, \mut_{rk}} \big) \\
&\doteq \tauc_{rk} \|\mut_{rk} - \muc_{rk}\|^2 
\end{align*}
where the second to last equality is by definition of $\muc_{rk} := \xc_{rk}/\tauc_{rk}$.
Recalling the definitions of $\tauc_{rk}$ and $\xc_{rk} := \sum_{i=1}^{N_r} \tau_{rik} x_{ri}$. Hence,
\begin{align*}
	\sum_{r,i} \tau_{rik} 
	\frac{\|x_{ri} - \mut_{rk}\|^2} {\sigma^2_r} \doteq 
	\sum_r  \frac{\tauc_{rk}}{\sigma_r^2}\|\mut_{rk} - \muc_{rk}\|^2 
	= (\mut_k - \muc_k)^T D_k^{-1} (\mut_k - \muc_k)
\end{align*}
Thus, we obtain
\begin{align*}
J(\mut) \doteq -\frac12 \sum_k \big[ (\mut_{k}-\muc_{k})^T D^{-1}_k (\mut_{k}-\muc_{k}) + (\mut_k-\mu)^T \Sigma^{-1} (\mut_k-\mu) \big].
\end{align*}
Desired expression~\eqref{eq:J:mut:expr} follows by applying the following lemma.
\begin{lem}[Sum of quadratic forms]\label{lem:ex:quad:sumofquadform}
For symmetric matrices $Q_1,Q_2, \dots$,
\begin{align*}
\sum_{r} (x-m_r)^TQ_r^{-1}(x-m_r) 
=(x - m)^TQ^{-1}(x - m) + \text{const.}, \quad \forall x
\end{align*}
where $Q = (\sum_r Q^{-1}_r)^{-1}$ and $m =  \sum_r  Q Q_r^{-1} m_r$.
\begin{proof}
	Since the two sides are quadratic functions, they are equal up to constants if their derivatives up to second-order match. Equating the Hessians gives $\sum_r Q_r^{-1} = Q^{-1}$. Then, equating the gradients gives $\sum_r Q_r^{-1}(x-m_r) = Q^{-1}(x-m)$, which simplifies to $\sum_r Q_r^{-1} m_r = Q^{-1} m$ in light of the Hessian equality.
\end{proof}

%\begin{align*}
%\Sigma_c &= (\Sigma^{-1}_1 + \Sigma^{-1}_2)^{-1}\\
%m_c &= \Sigma_c (\Sigma^{-1}_1 m_1 + \Sigma^{-1}_2 m_2)\\
%%C&= \frac12(m_1^T\Sigma^{-1}_1 + m_2^T\Sigma^{-1}_2) (\Sigma^{-1}_1 + \Sigma^{-1}_2)^{-1}(\Sigma^{-1}_1 m_1 + \Sigma^{-1}_2 m_2) -\frac12(m_1^T\Sigma^{-1}_1 m_1 + m_2^T
%%\Sigma^{-1}_2 m_2)
%\end{align*}
\end{lem}
%By Lemma~\ref{lem:ex:quad:sumofquadform}, $J(\mut) \doteq -\frac12\sum_k (\mut_k - m_k)^T (D_k^{-1}+\Sigma^{-1}) (\mut_k - m_k)$
%
%  we can write $J(\mut)$ as follows:\\
%\begin{align*}
%J(\mut)&= -\frac12\sum_k (\mut_k - m_k)^T (D_k^{-1}+\Sigma^{-1}) (\mut_k - m_k) +C\\
%\text{where,}\; m_k &= (D_k^{-1}+\Sigma^{-1})^{-1} (D_k^{-1}\muc_k +\Sigma^{-1}\mu)
%  \end{align*}
%So the optimal value of $\mut_k = m_k = (D_k^{-1}+\Sigma^{-1})^{-1} (D_k^{-1}\muc_k +\Sigma^{-1}\mu)$. \\
%Note that based on our derivation in the last section, the optimal value of $\Sigt_k$ was found to be $(D_k^{-1}+\Sigma^{-1})^{-1}$. 

\subsection{Updates of $\sigma^2$, $\pi$, $p$ and $q$}\label{sec:sig2:etc}
The relevant portion of $J$ as a function $(\sigma_r^2)$ is
\begin{align}
J((\sigma_r^2)) &=  -\frac12\sum_r \Big[ \frac1{\sigma^2_r} \sum_{i,k} \tau_{rik} \big[ \tr\big((\Sigt_k)_{rr}\big) + \|x_{ri} - \mut_{rk}\|^2 \big]
+d_r N_r \log \sigma_r^2 \Big].
\end{align}
The maximizer of the function $x \mapsto Ax^{-1} + B \log x$ is $A/B$ (assuming $A,B > 0$), from which~\eqref{eq:sig2:update} follows.

%We first update $\sigma^2_1$. 
%\begin{align*}
% J(\sigma_1) & = - \frac{N_1 d_1}{2}\log\sigma_1^2 - \frac{\sum_k \tauc_{1k}\tr\big((\Sigt_k)_{11}\big)+\sum_{i,k} \tau_{1ik}\|x_{1i} - \mut_{1k}\|^2}{2\sigma^2_1}  \\
%  \frac{\partial J(\sigma^2)}{\partial \sigma_1^2}&= -\frac{N_1 d_1}{2} \frac{1}{\sigma^2_1}  
%  + \frac{\sum_k \tauc_{1k}\tr\big((\Sigt_k)_{11}\big)+\sum_{i,k} \tau_{1ik}\|x_{1i} - \mut_{1k}\|^2}{2\sigma^4_1}=0\\
%  \sigma^2_1 &= \frac{\sum_k \tauc_{1k}\tr\big((\Sigt_k)_{11}\big)+\sum_{i,k} \tau_{1ik}\|x_{1i} - \mut_{1k}\|^2}{N_1 d_1}
%\end{align*}
%Similarly, we can update $\sigma^2_2$
% 
%\subsection{Updates of $\pi$, $p$ and $q$}
As a function $\pi$, $J$ has the form $J(\pi) \doteq \sum_{r,i,k} \tau_{rik} \log \pi_{rk}
 = \sum_{r,k} \tauc_{rk} \log\pi_{rk}
$ using the definition of $\tauc_{rk}$.  The following lemma is standard. (Recall that $\Pc_K$ is the set of probability $K$-vectors.)
\begin{lem}
For any nonnegative vector $(a_1,\dots,a_K)$, 
  \begin{align*}
    \argmax_{p\; \in \; \Pc_K} \; \sum_{k} a_k \log p_k = \frac{1}{\sum_k {a_k}} \big(a_1,\ldots, a_K).
  \end{align*}
\end{lem}
%\begin{proof}
%Assume $a_i > 0$ for all $i \in [K]$.
%Introducing Lagrange multiplier $\lambda_k^*$ for the inequality constraints $p_k\geq 0 ;\ \forall k$, and a multiplier $\nu^*$ for the equaity constraint $\sum_k p_k = 1$, we obtain {\it{KKT}} conditions:
%\begin{align}
% \frac{a_k}{p_k} + \lambda_k^* +\nu^* &=0 \quad \forall k\in[K]\label{eq:pi:KKT:1}\\
%\lambda_k^* p_k^* &= 0 \quad \forall k\in[K]\label{eq:pi:KKT:5}
%\end{align}
%together with $\sum_k p_k^* = 1$, $p_k^* \geq 0, \;\forall k\in[K]$, and $\lambda^* \ge 0$. 
%%
%From~\eqref{eq:pi:KKT:1}, we infer that $p_k^* = -\frac{a_k}{\lambda^*_k +\nu^*}$
%From~\eqref{eq:pi:KKT:5}  we see that $\lambda_k^* = 0$ since $p_k^* \neq 0$. Hence, $p^*_k \propto a_k$ and $\nu^*$ is determined by the normalization constraint.
%\end{proof}

Based on the lemma, $\pi_1$-update is $\pi_1= (\tauc_{11} ,\ldots, \tauc_{1K})/(\sum_k \tauc_{1k})$.
The update for $\pi_{2}$ is similar. 
%
%\subsection{Q = (p,q)}
To update $p$ we note that $J(p) \doteq \sum_{i,j} \gamma_{ij}(\tau) (A_{ij}\log p + (1-A_{ij}) \log(1-p))$. The update is obtained by setting the derivative to zero. The $q$-update is similar.
% We first find the optimal value for $p$.
% \begin{align*}
% J(p) &= \sum_{i,j} \gamma_{ij}(\tau) (A_{ij}\log p + (1-A_{ij}) \log(1-p))\\
% \frac{\partial J(p)}{\partial p} &= \sum_{i,j} \gamma_{ij}(\tau) \big(\frac{A_{ij}}{p} - \frac{1-A_{ij}}{1-p}\big) = 0\\
% \text{So};\ p &= \frac{\sum_{ij} \gamma_{ij}(\tau) A_{ij}}{\sum_{i,j}\gamma_{ij}(\tau)}
% \end{align*}

% Similarly, the optimal value for $q$ is found to be\\
% \begin{align*}
% q&= \frac{\sum_{i,j} \big(1-\gamma_{ij}(\tau)\big) A_{ij}}{\sum_{i,j} \big(1-\gamma_{ij}(\tau)\big)}
% \end{align*}

\section{Details for degree-corrected algorithm}
\subsection{$\tau$-update with degree restriction}\label{sec:dc:tau:details}
In this section, we derive a dual ascent algorithm for the optimization problem~\eqref{eq:dc:tau:update:optim} which has to be solved for updating $\tau$ under the degree corrected model. Letting $a_{ik} := \phi_1 [A \tau_2]_{ik} + \phi_0\tauc_{2k}  + \xi_{1ik}$, and with some notational simplifications,  problem~\eqref{eq:dc:tau:update:optim} can be stated as 
\begin{align}\label{eq:constrined:label:optim}
\min_{X = (x_{ik})} -\sum_{ik} x_{ik} (a_{ik} - \log x_{ik}), \quad\text{s.t.}\;\; X \in \Pc_{n, K}, \; \;\sum_i x_{ik}(\theta_i - 1) = 0, \, \forall k
\end{align}
where $\Pc_{n, K} := \{(x_{ik}) \in \reals_+^{n \times K}:\; \sum_{k} x_{ik} = 1,\forall i\}$. Let $f(X)$ be the objective function in~\eqref{eq:constrined:label:optim}, and let us write the constraint in vector form $\sum_i x_{ik}(\theta_i - 1) = X^T (\theta - \onev) = 0$. With the Lagrangian $L(X,\lambda)  = f(X) - \lambda^T [ X^T (\theta - \onev)]$, the dual function is 
\begin{align*}
	\Phi(\lambda) := \min_{X \in \Pc_{n,K}} L(X,\lambda),
\end{align*}
and the dual problem is $\max_{\lambda} \Phi(\lambda)$.
A dual-descent algorithm maximizes $\Phi$ by performing a gradient ascent on $\Phi$: $\lambda^{+} = \lambda + \mu \nabla \Phi(\lambda)$. We know that the gradient of $\Phi$ is given by $\partial_\lambda L(X,\lambda)$ evaluated at $X^*(\lambda)$, the optimizer of the Lagrangian. More precisely,
\begin{align*}
\nabla \Phi(\lambda) = [X^*(\lambda)]^T (\theta - \onev), \quad \text{where}\;\; X^*(\lambda) = \argmax_{X \,\in\, \Pc_{n,K}} -L(X,\lambda)
\end{align*}
Solving for $X^*(\lambda)$ is an instance of the problem in Lemma~\ref{lem:tau:update}. The problem is separable over $i$ and for fixed $i$, we are maximizing $\sum_k  x_{ik}(a_{ik}-  \log x_{ik}) + \sum_k \lambda_k x_{ik}(\theta_i-1) = \sum_k x_{ik} \{[a_{ik} +\lambda_k (\theta_i-1)]-  \log x_{ik}\}$ over $x_{i*} \in \Pc_{1,K}$, the solution of which is given by the softmax operation
\begin{align}\label{eq:xs:dual:descent}
x_{ik}^*(\lambda) \propto_k \exp(a_{ik} +\lambda_k (\theta_i-1)).
\end{align}
Thus, the update for the dual descent can be written
\begin{align}
\lambda^{+}_k = \lambda_k - \mu \sum_{i} x^*_{ik}(\lambda) (\theta_i-1), \quad \forall k.
\end{align}
%where $x^*_{ik}(\lambda)$ is computed from~\eqref{eq:xs:dual:descent}.

\newcommand{\prox}{\text{prox}}
\subsection{$\theta$-update}\label{sec:dc:theta:details}
Simplyfying the notation, let $h(\theta) = - \sum_{i}  d_i \log \theta_{i}$ and $V := \{\theta:\; \sum_{i} \tau_{ik}(\theta_{i}-1) = 0\}$. The problem is equivalent to minimizing $h(\theta) + \delta_V(\theta)$ over $\theta$ where $\delta_V$ is the indicator of $V$ in the sense of convex analysis. Douglas-Rachford algorithm, also known as  Spingarn's method of partial inverses in this special case, is given by
\begin{align}
	\theta^+ &= \prox_{th}( \xi) \notag \\
	\xi^+ &= \xi + P_V(2 \theta^+ - \xi) - \theta^+\label{eq:temp:5676}
\end{align}
where $\prox_{th}$ is the proximal operator of $t h(\cdot)$ and $P_V$ is the projection onto $V$. Due to separability, it is not hard to see that $[\prox_{th}(\theta)]_i = \prox_{t d_i \log(\cdot)}(\theta_i)$. This  univariate proximal operator can be easily shown to coincide with $[f_t(\theta,d)]_i$ as given in~\eqref{eq:prox:of:log}.

As for the projection, in general with $C = \{x:\; Ax = b\}$, we have $P_C(x) = x + A^T (AA^T)^{-1}(b-Ax)$. Note that $V = \{\theta :\; \tau^T(\theta - \onev) = 0\}$. Applying the general result with $A = \tau^T$ and $b = \tau^T \onev$, we get $P_V(\theta) = \theta + \tau (\tau^T \tau)^{-1} \tau^T(\onev - \theta) = \theta - H( \theta- \onev)$, with the obvious choice for $H$. Thus~\eqref{eq:temp:5676} simplifies to
\begin{align*}
		\xi^+ = \xi + (2 \theta^+ - \xi) -H(2 \theta^+ - \xi - \onev) - \theta^+
\end{align*}
which gives the claimed update.

\section{Expected average degree}
\label{sec:expec:avg:degree}
Let $\hat\lambda_i = \sum_j A_{ij}$ and $\hat\lambda_j = \sum_i A_{ij}$ be the degree of node $i$ from group 1, and node $j$ from group 2, respectively. Then, the average degree of the network is 
\begin{align*}
\hat{\lambda} &= \frac{\sum_i \hat\lambda_i + \sum_j \hat\lambda_j}{N_1 + N_2}
= \frac{2 \sum_{i,j} A_{ij}}{N_1 + N_2}.
\end{align*}
Recall the definition of clusters $C_{rk}$ from~\eqref{eq:C:match}. The expected average degree can be derived as follows: Assume that $i \in C_{1k}$, then $\ex[A_{ij}] =  \theta_{1i} \theta_{2j} (p 1\{j \in C_{2k}\} + q 1\{j \notin C_{2k}\})$. Then
\begin{align*}
	\E(\hat\lambda_i) = \theta_{1i} \Big( p\sum_{j \in C_{2k}} \theta_{2j} + q \sum_{j \notin C_{2k}} \theta_{2j}\Big)
	= \theta_i r_k, \quad\text{where}\quad r_k := |C_{2k}|p + (N_2 - |C_{2k}|)q
\end{align*}
 Here, we have used the normalization~\eqref{eq:theta:normalization}: $\sum_{j \in C_{2\ell}} \theta_{2j} = |C_{2\ell}|$ for all $\ell \in [K]$. Now,
 \begin{align*}
 	\sum_{i=1}^{N_1} \E(\hat\lambda_i) = \sum_{i=1}^{N_1} \sum_{k} 1\{i \in C_{1k}\}\E(\hat\lambda_i) 
 	= \sum_k r_k \sum_{i=1}^{N_1} \theta_{1i} 1\{i \in C_{1k}\} = \sum_k r_k |C_{1k}|
 \end{align*}
using the normalization of $\theta_{1i}$. Dividing by $N_1 N_2$ and using $|C_{rk}|/N_r = \pi_{rk}$ for $r=1,2$, 
 \begin{align*}
 	\sum_{i=1}^{N_1} \E(\hat\lambda_i) = N_1 N_2 \sum_k   \big[\pi_{1k} q + \pi_{1k}\pi_{2k} (p-q)\big] =  N_1 N_2( q + \Pi (p-q)),
 \end{align*}
% 
%\begin{align*}
%  \sum_{i=1}^{N_1} \E(\hat\lambda_i) = \sum_{i=1}^{N_1} \sum_{k} 1\{i \in C_{1k}\}\E(\hat\lambda_i) 
%   &=   \sum_{k} \Big[ |C_{2k}|p + (N_2 - |C_{2k}|)q \Big] \sum_{i=1}^{N_1} 1\{i \in C_{1k}\}\\
%   &=  \sum_{k} \Big[ |C_{2k}|p + (N_2 - |C_{2k}|)q \Big] | C_{1k}| \\
%   &= N_1 N_2 \sum_k \Big[\frac{|C_{2k}|}{N_2} p + \big(1- \frac{|C_{2k}|}{N_2}\big) q \Big] \frac{|C_{1k}|}{N_1} \\
%  &= N_1 N_2 \sum_k \Big[\pi_{2k} p + \big(1- \pi_{2k}\big) q \Big] \pi_{1k}\\
%   & = N_1 N_2 \sum_k \big[\pi_{1k} q + \pi_{1k}\pi_{2k} (p-q)\big] =  N_1 N_2( q + \Pi (p-q)),
%\end{align*}
where $\Pi := \sum_k \pi_{1k} \pi_{2k}$. By symmetry, $\sum_{j=1}^{N_2} \E(\hat\lambda_j) =  N_1 N_2( q + \Pi (p-q))$. Hence,
% Similarly
% \begin{align*}
%   \sum_{j=1}^{N_2} \E(\hat\lambda_j) = \sum_k \Big[|C_{1k}| p + (N_1 - |C_{1k}|)q \Big] |C_{2k}|
%  &=N_1 N_2 \sum_k \big[\pi_{2k} q + \pi_{2k}\pi_{1k} (p-q)\big]  \\
%  &=  N_1 N_2( q + \Pi (p-q))
% \end{align*}
% \begin{align*}
% %\E(k_i)&= \sum_j \E(A_{ij})\\
% \sum_i \E(\hat\lambda_i)& = \frac{N_1\sum_k \Big[\big(N_2 - |C_{2k}|\big) q + |C_{2k}| p\Big] |C_{1k}|}{N_1} \\
% & = N_1 N_2 \sum_k \Big[\big(1- \frac{|C_{2k}|}{N_2}\big) q + \frac{|C_{2k}|}{N_2} p\Big] \frac{|C_{1k}|}{N_1}\\
% & = N_1 N_2 \sum_k \Big[\big(1- \pi_{2k}\big) q + \pi_{2k} p\Big] \pi_{1k}\\
% & = N_1 N_2 \sum_k \big[\pi_{1k} q + \pi_{1k}\pi_{2k} (p-q)\big]\\
% \text{Similarly,} \quad 
% \end{align*}
% So 
\begin{align*}
\lambda & = \E[\hat{\lambda}] = \frac{2 N_1 N_2}{N_1+ N_2}\big(q + \Pi(p-q)\big), \quad \text{where}\; \Pi := \sum_k \pi_{1k} \pi_{2k}.
\end{align*}
Note that $\frac{2 N_1 N_2}{N_1+ N_2} = 2/(N_1^{-1} + N_2^{-1})$ is the harmonic mean of $N_1$ and $N_2$.

\section{More simulations}\label{sec:extra:sim}
Figure~\ref{fig:nmidegcovaradim} shows the effect of varying the dimension of the covariates $d = (d_1,d_2)$ and the scale of the covariance matrix $\nu$. The setup is as in~Section~\ref{sec:simulation}, and in particular $\Sigma = \nu I$ controls how far apart the centers of the covariate clusters $v_{*k} \in \reals^{d_1+d_2}$ are. Only the \mbisbm algorithm is shown (with no degree correction and) initialized with Dirichlet-perturbed truth (\simrnd). As one expects, increasing the dimensions of the covariates increases the performance (seemingly without bound). Increasing $\nu$ improves the performance up to a point, but there is a saturation effect beyond that point, where the performance remains more or less the same.

\begin{figure}[H]
	\centering
	\includegraphics[width=0.43\linewidth]{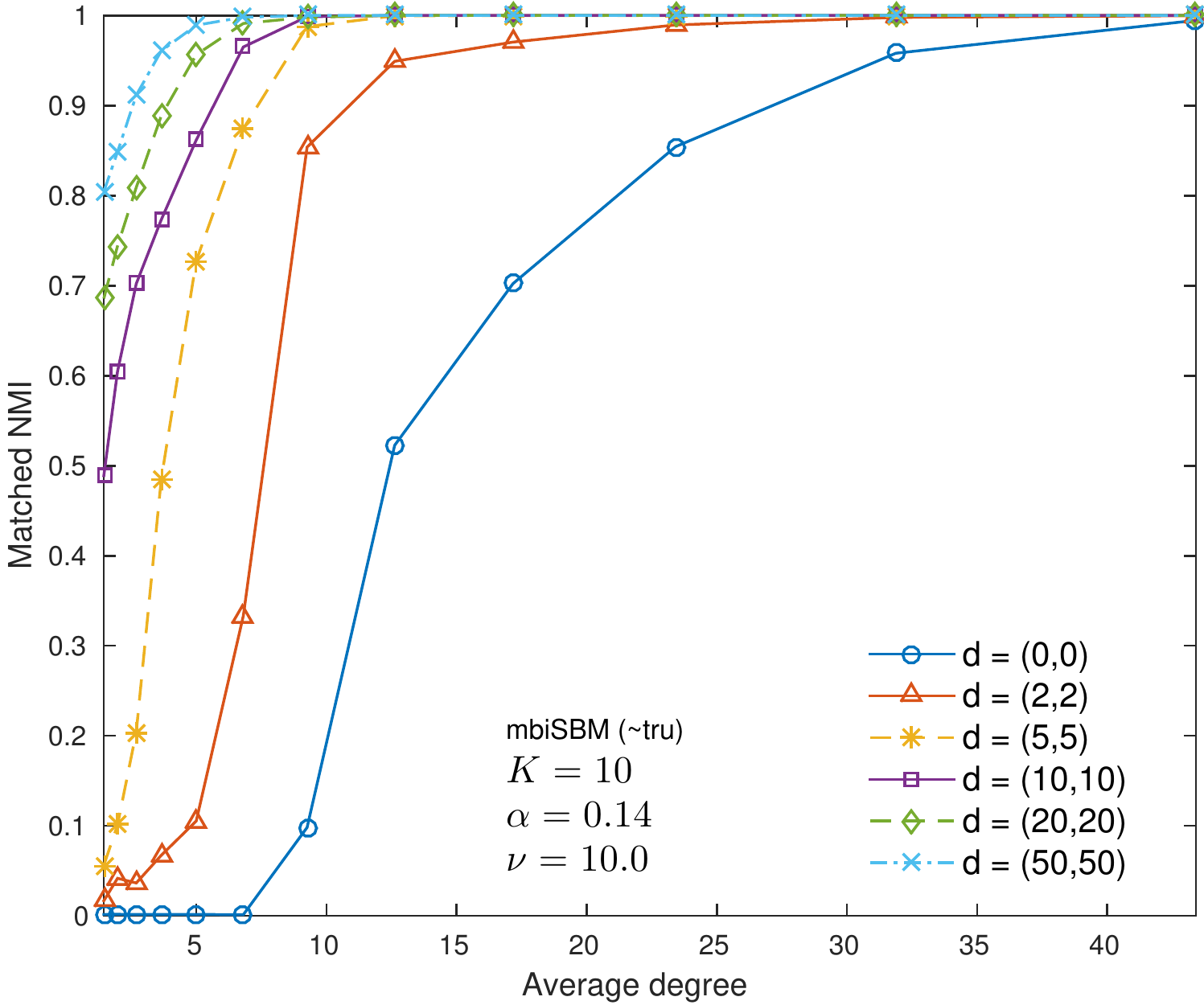}
	\includegraphics[width=0.43\linewidth]{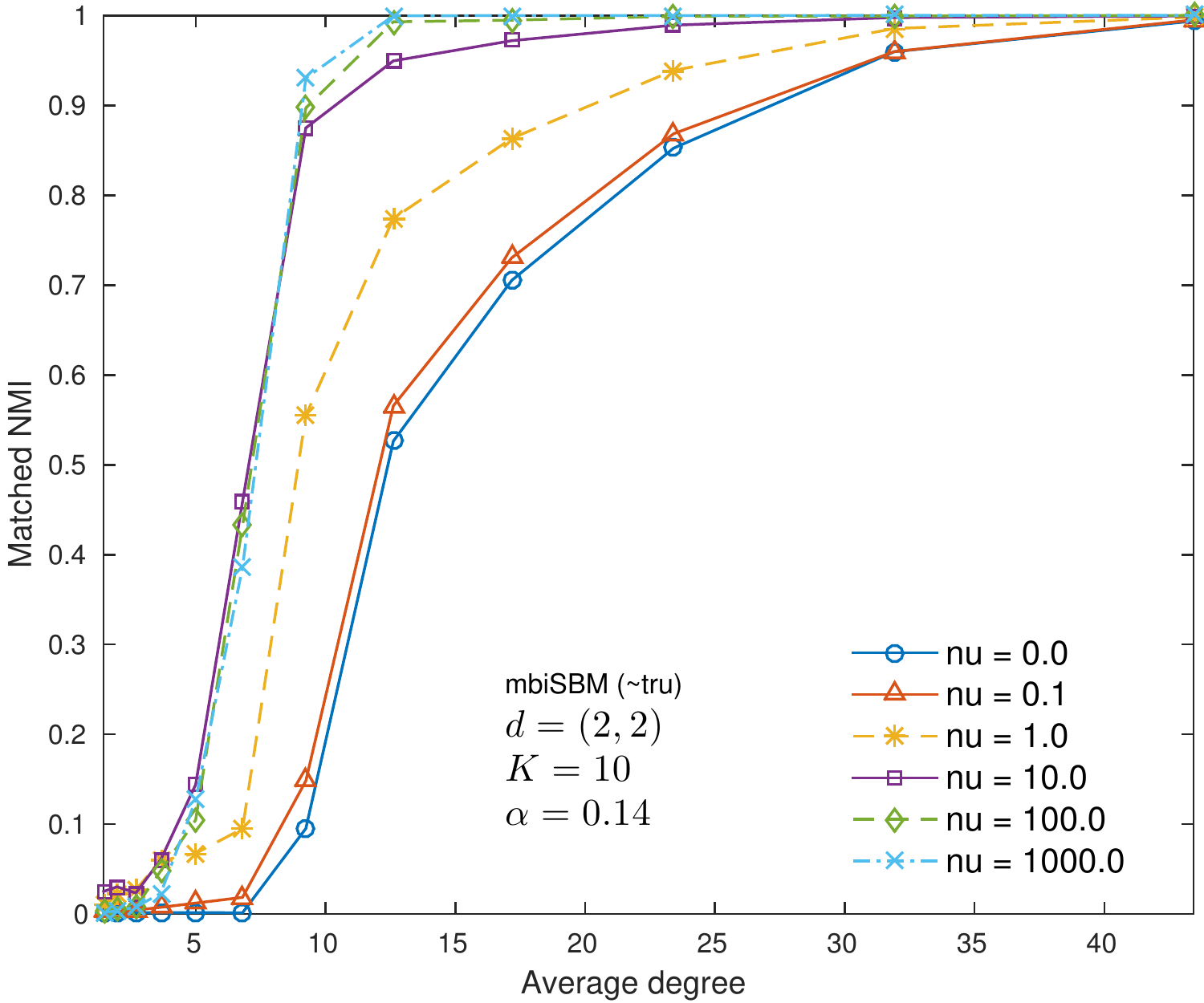}
	\caption{(Left) Effect of changing the dimension of covariates $d = (d_1,d_2)$ and (right) the parameter $\nu$ where $\Sigma = \nu I_{d_1 + d_2}$.}
	\label{fig:nmidegcovaradim}
\end{figure}

%\begin{figure}[ht]
%\centering
%%\subfigure{%
%\includegraphics[scale=.32]{Figs/NMI_odds_25_2}
%%\label{fig:subfigure21}}
%%\subfigure{%
%\includegraphics[scale=.32]{Figs/NMI_odds_27_2}
%%\label{fig:subfigure22}}
%% \subfigure{%3
% \includegraphics[scale=.32]{Figs/NMI_odds_33_2}
%% \label{fig:subfigure23}}
%% \subfigure{%
% \includegraphics[scale=.32]{Figs/NMI_odds_31_2}
%% \label{fig:subfigure24}}\textbf{}
% \caption{NMI versus $\lambda$ for various value of $N_1 , N_2, K , \alpha , \omega , \nu$}
% \label{fig:meandeg}
%  \end{figure}
%
%\begin{figure}[ht]
%\centering
%%\subfigure{%
%\includegraphics[scale=.32]{Figs/NMI_odds_21_2.eps}
%%\label{fig:subfigure21}}
%%\subfigure{%
%\includegraphics[scale=.32]{Figs/NMI_odds_22_2.eps}
%%\label{fig:subfigure22}}
%% \subfigure{%3
% \includegraphics[scale=.32]{Figs/NMI_odds_23_2.eps}
%% \label{fig:subfigure23}}
%% \subfigure{%
% \includegraphics[scale=.32]{Figs/NMI_odds_24_2.eps}
%% \label{fig:subfigure24}}
% \caption{NMI versus $\alpha = q/p$ for various value of $N_1 , N_2, K , \lambda , \omega , \nu$}
% \label{fig:ratio}
%  \end{figure}

\end{document}